\newif\iffinal
\finalfalse	
\finaltrue 

\newif\ifmarek
\marektrue

\documentclass[reqno,twoside,11pt]{amsart}
\iffinal\else\usepackage[notref,notcite]{showkeys}\fi
\usepackage{cite}
\usepackage[toc,title,page]{appendix}
\usepackage{subcaption}
\usepackage{amsmath}
\usepackage{amsmath,lipsum}
\usepackage{amsfonts}
\usepackage{amssymb}
\usepackage{graphicx}
\usepackage{enumitem}
\usepackage{graphics}
\usepackage{booktabs}
\usepackage{multirow}
\usepackage{tikz}
\usetikzlibrary{arrows,automata,positioning}
\usepackage{verbatim}
\usepackage{color}
\usepackage{float}
\usepackage{dirtytalk}
\IfFileExists{epsf.def}{\input epsf.def}{\usepackage{epsf}}
\ifmarek
\IfFileExists{myowntimes.sty}{\usepackage{myowntimes}\usepackage{mathrsfs}}
	{\usepackage{times}\newcommand{\mathscr}{\mathcal}}

\usepackage{mathptmx}
\fi

\DeclareFontFamily{OT1}{eusb}{} \DeclareFontShape{OT1}{eusb}{m}{n}
{<5> <6> <7> <8> <9> <10> <11> <12> <14.4> eusb10}{}
\DeclareMathAlphabet{\eusb}{OT1}{eusb}{m}{n}

\DeclareFontFamily{OT1}{eusm}{} \DeclareFontShape{OT1}{eusm}{m}{n}
{<5> <6> <7> <8> <9> <10> <11> <12> <14.4> eusm10}{}
\DeclareMathAlphabet{\eusm}{OT1}{eusm}{m}{n}

\DeclareFontFamily{OT1}{eufm}{} \DeclareFontShape{OT1}{eufm}{m}{n}
{<5> <6> <7> <8> <9> <10> <11> <12> <14.4> eufm10}{}
\DeclareMathAlphabet{\mathfrak}{OT1}{eufm}{m}{n}

\DeclareFontFamily{OT1}{fraktura}{}
\DeclareFontShape{OT1}{fraktura}{m}{n} {<5> <6> <7> <8> <9> <10>
<11> <12> <13> <14.4> [1.1] eufm10}{}
\DeclareMathAlphabet{\fraktura}{OT1}{fraktura}{m}{n}

\DeclareFontFamily{OT1}{cmfi}{} \DeclareFontShape{OT1}{cmfi}{m}{n}
{<5> <6> <7> <8> <9> <10> <11> <12> <13> <14.4> [0.9] cmfi10}{}
\DeclareMathAlphabet{\cmfi}{OT1}{cmfi}{b}{n}

\DeclareFontFamily{OT1}{cmss}{} \DeclareFontShape{OT1}{cmss}{m}{n}
{<5> <6> <7> <8> <9> <10> <11> <12> <13> <14.4> cmss10}{}
\DeclareMathAlphabet{\cmss}{OT1}{cmss}{m}{n}

\newtheoremstyle{thm}{1.5ex}{1.5ex}{\itshape\rmfamily}{}
{\bfseries\rmfamily}{}{2ex}{}

\newtheoremstyle{def}{1.5ex}{1.5ex}{\slshape\rmfamily}{}
{\bfseries\rmfamily}{}{2ex}{}

\newtheoremstyle{rem}{1.3ex}{1.3ex}{\rmfamily}{}
{\itshape}
{} {1.5ex}{}

\theoremstyle{thm}
\newtheorem{theorem}{Theorem}[section]
\newtheorem{lemma}[theorem]{Lemma}

\newtheorem{corollary}[theorem]{Corollary}

\theoremstyle{def}

\theoremstyle{rem}

\numberwithin{equation}{section}

\renewcommand{\subsection}{\secdef\subsct\sbsect}
\newcommand{\subsct}[2][default]{\refstepcounter{subsection}
\addcontentsline{toc}{subsection}
{{\tocsection{\!\!}{\hspace{1.2em}\thesubsection}{\!\!\!\!#1\dotfill}}{}}
\nopagebreak\vspace{0.45\baselineskip} {\flushleft\bf
\thesection.\arabic{subsection}~\bf #1.~}
\\*[3mm]\noindent
\nopagebreak}
\newcommand{\sbsect}[1]{\vspace{0.1cm}\noindent
\textbf{#1.~}\vspace{0.1cm}}

\renewcommand{\subsubsection}{%
\secdef \subsubsect\sbsbsect}
\newcommand{\subsubsect}[2][default]{%
\refstepcounter{subsubsection} 
\addcontentsline{toc}{subsubsection}{{\tocsection{\!\!}
{\hspace{3.05em}\thesubsubsection}{\!\!\!\!#1\dotfill}}{}}
\nopagebreak
\vspace{0.15\baselineskip} \nopagebreak {\flushleft\rmfamily
\itshape\arabic{section}.\arabic{subsection}.\arabic{subsubsection}
\ \rmfamily #1\/.}\ }
\newcommand{\sbsbsect}[1]{\vspace{0.1cm}\noindent
\rmfamily \itshape
\arabic{section}.\arabic{subsection}.\arabic{subsubsection} \
\sffamily #1\/.\ }
\iffinal

\else

\fi
   
\usepackage[toc,title,page]{appendix}

\newcommand{\scrC}{\mathscr{C}}

\newcommand{\scrF}{\mathscr{F}}

\title[]
{Pricing  European Options under Stochastic Volatility Models: \\Case of five-Parameter Variance-Gamma Process}

\author [] {A.H.Nzokem$^1$}
\thanks {$^1$ hilaire77@gmail.com}
\begin{document}

\maketitle
\begin{abstract}

\noindent
The paper builds a Variance-Gamma (VG) model with five parameters: location ($\mu$), symmetry ($\delta$), volatility ($\sigma$), shape ($\alpha$), and scale ($\theta$); and studies its application to the pricing of European options. The results of our analysis show that the five-parameter VG model is a stochastic volatility model with a $\Gamma(\alpha, \theta)$ Ornstein--Uhlenbeck type process; the associated L\'evy density of the VG model is a KoBoL family of order $\nu=0$, intensity $\alpha$, and steepness parameters $\frac{\delta}{\sigma^2} - \sqrt{\frac{\delta^2}{\sigma^4}+\frac{2}{\theta \sigma^2}}$ and $\frac{\delta}{\sigma^2}+ \sqrt{\frac{\delta^2}{\sigma^4}+\frac{2}{\theta \sigma^2}}$; and the VG process converges asymptotically in distribution to a L\'evy process driven by a normal distribution with mean $(\mu + \alpha \theta \delta)$ and variance $\alpha (\theta^2\delta^2 + \sigma^2\theta)$. The data used for empirical analysis were obtained by fitting the five-parameter Variance-Gamma (VG) model to the underlying distribution of the daily SPY ETF data. 
 Regarding the application of the five-parameter VG model, the twelve-point rule Composite Newton--Cotes Quadrature and Fractional Fast Fourier (FRFT) algorithms were implemented to compute the European option price. Compared to the Black--Scholes (BS) model, empirical evidence shows that the VG option price is underpriced for out-of-the-money (OTM) options and overpriced for in-the-money (ITM) options. Both models produce almost the same option pricing results for deep out-of-the-money (OTM) and deep-in-the-money (ITM) options..\\

\noindent
\end{abstract}

 \section {Introduction}
\noindent
Black-Scholes (BS) model \cite{black1973} is considered the cornerstone of option pricing theory. The model relies on the fundamental assumption that the asset returns have a normal distribution with a known mean and variance. However, based on empirical studies, the Black-Scholes (BS) model is inconsistent with a set of well-established stylized features \cite{Cont_2001}. Due to the subsequent development of the option pricing theory, a new class of models has emerged in the literature to address the stylized characteristics of the markets. The probabilistic property of infinity divisibility is the main characteristic of these new models, and they belong to the family of Levy processes \cite{kyprianou2014fluctuations}. \\
The new class of models can be divided into two subclasses of Jump-Diffusion models and Stochastic Volatility models. The Jump-Diffusion process is modeled as an independent  Brownian motion plus a Compound Poisson Process. The popular models in the literature are Merton's jump-diffusion model \cite{matsuda2004introduction} and Kou's jump-diffusion model \cite{kou2002jump}, where logarithmic jump size follows a normal distribution and an asymmetric double exponential distribution respectively. The stochastic volatility (SV) model is another extension of the standard geometric Brownian motion (GBM) model, where the observed volatility is modeled as a stochastic process. In a stochastic volatility framework \cite{alhagyan2020discussions}, the constant volatility ($\sigma$) in a standard geometric Brownian motion (GBM) model is replaced by a deterministic function of a stochastic process ($\sigma(Y_{t})$) where $Y_{t} $ represents the solution of the stochastic differential equation (SDE). We have two main types of SV models in the literature:  Diffusion based SV models and non-Gaussian Ornstein-Uhlenbeck-based SV models.  In the popular diffusion-based SV models, $Y_{t}$ follows a Cox-Ingersoll-Ross(CIR) process \cite{heston1993closed}  or a Log-normal process \cite{hull1987pricing} and the deterministic function is a squared root of the stochastic process ($\sigma(Y_{t})=\sqrt{Y_{t}}$). The non-Gaussian Ornstein-Uhlenbeck-based SV models have been introduced and thoroughly studied in \cite{barndorff2002financial,barndorff2003integrated,barndorff1998some,barndorff1999non}. The SV model with the Ornstein-Uhlenbeck type process is mathematically tractable and has many appealing features.\\
 From the perspective of derivative asset analysis, we will build a five-parameter VG model as a Stochastic volatility model with $\Gamma(\alpha, \theta)$ Ornstein-Uhlenbeck type process. While there is a great number of studies on option pricing under the VG Model, most of the VG Model in the literature has three parameters \cite{li2020pricing, mozumder2015revisiting, madan1998variance,adeosun2016variance}, certainly due to technical issues of fitting a high parametric model to the marginal distribution of asset returns.  The amount of literature considering the VG Model with five parameters is rather limited. Using the five-parameter Variance-Gamma model as an underlying distribution of the European option will allow controlling both the excess kurtosis and the skewness of the underlying data. In the option pricing theory, the challenge is often the existence of the Equivalent Martingale Measure (EMM) and whether it preserves the structure of the Variance-Gamma measure. The Variance-Gamma (VG) process is not a Gaussian process, and the market is incomplete; therefore, the Equivalent Martingale Measure is not unique. The Esscher transform of the historical measure is considered optimal with respect to some optimization criterion \cite{boyarchenko2002non}. The Esscher martingale measure was shown \cite{andrii2020}  to coincide with the minimal entropy martingale measure for L\'evy processes.\\
The remainder of this paper is organized as follows. Section 2 is devoted to building a five-parameter VG process and presenting parameter estimations and simulations of the VG process. Section 3 investigates the L\'evy density and the asymptotic distribution of the VG process. And section 4 extends the Black-Scholes framework, provides the integral representation for the option price, and computes the VG option price numerically.

\section{Variance - Gamma Process: Stochastic Volatility Model}
\subsection{L\'evy Framework and Asset Pricing }
\noindent
Let $(\Omega, \scrF, \left\{ \scrF_{t} \right \}_{t \geq 0},  \mathbf{P})$ be a filtered probability space, with $\scrF =\vee_{t>0}\scrF_{t}$ and $\left\{ \scrF_{t} \right \}_{t \geq 0}$ is a filtration.  $\scrF_{t}$ is a $\sigma$-algebra included in $\scrF$ and for $\tau<t$, $\scrF_{\tau}\subseteq \scrF_{t}$.\\
\noindent
A stochastic process $Y=\{Y_{t}\}_{t\geq0}$ is a L\'evy process, if it has the following properties\\
(L1): $Y_{0} = 0$ a.s;\\
(L2): $Y_{t }$ has independent increments, that is,  for $0<t_{1} < t_{2} <\dots< t_{n}$, the random variables $Y_{t_{1} }$, $Y_{t_{1} } - Y_{t_{2} }$, \dots, $Y_{t_{n} } - Y_{t_{n-1} }$  are independent;\\
(L3): $Y_{t }$ has stationary increments, that is, 
 for any $t_{1} < t_{2} < +\infty$, the probability distribution  of $Y_{t_{1} } - Y_{t_{2} }$   depends only on $t_{1}  - t_{2} $ ;\\
(L4): $Y_{t }$ is stochastically continuous: for any $t$ and $\epsilon>0$, $\lim_{s \to t} P\left(|Y_{s} - Y_{t}|>\epsilon \right)=0$;\\
(L5): $c\grave{a}dl\grave{a}g$ paths, that is, $t \mapsto Y_{t}$ is a.s right continuous with left limits\\

\noindent
Given a L\'evy process $Y=\{Y_{t}\}_{t\geq0}$  on the filtered probability space $(\Omega, \scrF, \left\{ \scrF_{t} \right \}_{t \geq 0}, \mathbf{P})$, we define the asset value process $S=\{S_{t}\}_{t\geq0}$ such as $S_{t}=S_{0}e^{Y_{t}}$.

 \begin{theorem}\label{lem1} (L\'evy-Khintchine representation) \ \\
 Let $Y=\{Y_{t}\}_{t\geq0}$ be a L\'evy process on $\mathbb{R}$. Then the characteristic exponent admits the following representation.
 \begin{align}
\varphi(\xi)= -Log\left(Ee^{i Y_{1}\xi}\right)= - i\gamma \xi + \frac{1}{2} \sigma^{2}\xi^{2} + \int_{0}^{t}\left(e^{i \xi y} -1 - y\xi 1_{|y|\leq 1}\right)\Pi(dy)
\label {eq:l01}
  \end{align}
where $\gamma \in \mathbb{R}$,  $\sigma \geq 0$ and $\Pi$  is a $\sigma$-finite measure called the Lévy measure of Y, satisfying the property
$$\int_{-\infty}^{+\infty}Min(1,|y|^{2})\Pi(dy) <+\infty$$
\end{theorem} 

\noindent
For the Theorem-proof, see \cite{applebaum2009levy,ken1999levy,tankov2003financial}\\
\noindent
Each L\'evy process is uniquely determined by the L\'evy–Khintchine triplet  $(\gamma, \sigma^{2},\Pi )$. The terms of this triplet suggest that a L\'evy process can be seen as having three independent components: a linear drift, a Brownian motion, and a L\'evy jump process. When the diffusion term $\sigma=0$, we have a  L\'evy jump process; in addition, if $\gamma=0$, we have a pure jump process.

\subsection{$\Gamma(\alpha, \theta)$ Ornstein-Uhlenbeck proces} 
\noindent
The Ornstein-Uhlenbeck process is a diffusion process introduced by Ornstein and Uhlenbeck \cite{uhlenback1930theory} to model the stochastic behavior of the velocity of a particle undergoing Brownian motion. The Ornstein-Uhlenbeck diffusion $\sigma^{2}=\{\sigma^{2}(t), t\geq 0\}$ is the solution of the Langevin Stochastic Differential Equation (SDE) (\ref{eq:l01})
\begin{align}
d\sigma^{2}(t)=-\lambda \sigma^{2}(t)dt + dB(\lambda t) .
\label {eq:l01}
  \end{align}
where $\lambda >0$ and $B=\{B_{t},t\geq 0\}$ is a Brownian motion.
\noindent
In recent years, the Ornstein-Uhlenbeck process has been used in finance to capture important distributional deviations from Gaussianity and to model dependence structures. The extension of the Ornstein-Uhlenbeck processes was obtained by replacing the Brownian motion in (\ref{eq:l01}) by z(t), a background driving L\'evy process (BDLP) \cite{barndorff2003integrated,barndorff2001modelling,barndorff2002financial}. The SDE (\ref{eq:l01}) becomes
\begin{align} 
d\sigma^{2}(t)=-\lambda \sigma^{2}(t)dt + dz(\lambda t)  \quad  \lambda >0.
\label {eq:l02}
  \end{align}
where the process  $z(t)=\{z(t), t \geq 0, z(0)=0 \}$ is subordinator; a process with non-negative, independent and stationary increments, which implies  $\sigma^{2}(t) \geq 0$. Correspondingly z(t) moves up entirely
by jumps and then tails off exponentially \cite{barndorff1999non}.
\begin{lemma} \label{lem2} \ \\
The general form of the stationary process $\sigma^{2}(t)$ , solution of (\ref {eq:l02}) is given by :
\begin{align}
\sigma^{2}(t)=-\int_{0}^{+\infty}e^{-s} dz(\lambda t -s)  \quad  \lambda >0. 
\label {eq:l03}
  \end{align}
\end{lemma} 
\begin{proof} 
\begin{align}
\sigma^{2}(t)=-\int_{0}^{+\infty}e^{-s} dz(\lambda t -s)= -\int_{0}^{+\infty}e^{-\lambda s} dz(\lambda (t -s))=\int_{-\infty}^{t}e^{-\lambda(t-s)} dz(\lambda s) 
\label {eq:l04}
  \end{align}
By using the variable  changing method, we can have different expressions of (\ref {eq:l03}).
\begin{align}
\sigma^{2}(t)=\int_{-\infty}^{t}e^{-\lambda(t-s)} dz(\lambda s) \quad \Longrightarrow  \quad d\sigma^{2}(t)=-\lambda \sigma^{2}(t)dt + dz(\lambda t)
\label {eq:l05}
  \end{align}
\end{proof}

Expression (\ref{eq:l04}) can be written as follows:
\begin{align}
\sigma^{2}(t)=e^{-\lambda t}\sigma^{2}(0)+\int_{0}^{t}e^{-\lambda(t-s)} dz(\lambda s) \quad \quad \sigma^{2}(0)=\int_{-\infty}^{0}e^{\lambda s} dz(\lambda s) \label {eq:l06}
  \end{align}

 \begin{theorem}\label{lem3} \ \\
Assume $z(t)=\sum_{k=1}^{N(t)}\xi_{k}$ is a compound poison process, that is,  N(t)  is Poisson process with the instantaneous rate $\alpha$, and  $\xi_{k}$ follows an exponential distribution with the rate $\theta$. \\
\noindent
The stationary marginal distribution of $\sigma^{2}(t)$ is  Gamma distribution $\Gamma(\alpha, \theta)$
\end{theorem} 
\begin{proof}

\begin{align}
\sigma^{2}(t+u)=\int_{-\infty}^{t+u}e^{-\lambda(t+u-s)} dz(\lambda s)=e^{-\lambda u}\sigma^{2}(t)+\int_{0}^{u}e^{-\lambda(u-s)} dz(\lambda s)  \quad u\geq0 \label {eq:l07}
\end{align}
The stationary solution $\sigma^{2}(t)$ of (\ref {eq:l02}) can be written as in (\ref {eq:l07}). Because of the stationarity, we have 
\begin{align}
\vartheta(\xi)&=\vartheta(\xi e^{-\lambda u})\Phi(u,\xi) \label{eq:l08}
\end{align}
\noindent
$\vartheta(\xi)$ is the characteristic function of the stationary distribution of $\sigma^{2}(t)$ and $\Phi(u,\xi)$ is the characteristic function of $\int_{0}^{u}e^{-\lambda(u-s)} dz(\lambda s)$. We have $0 \leq e^{-\lambda u} \leq 1$ for $ u \geq 0 $, and the relation (\ref{eq:l07}) shows that $\sigma^{2}(t)$ is self-decomposable. \\
$z(t)$ is a compound poison process with the function characteristic. 
\begin{align}
g(\xi)=E(e^{i \xi z(1)})=exp\left\{\int_{0}^{\infty}(e^{i\xi x} -1)\alpha f(x)dx\right\}=exp({\rho(\xi)})  \quad \quad \rho(\xi)=\frac{i \xi\alpha}{\theta - i\xi}\label{eq:l09}
\end{align}
\noindent
It was shown in \cite{barndorff1998some} that $\Phi(u,\xi)$ can be expressed as follows 
\begin{align}
\Phi(u,\xi)=exp\left\{\lambda\int_{0}^{u}\rho(\xi e^{-\lambda(u-s)})ds\right\}=exp\left\{\int_{\xi e^{-\lambda u}}^{\xi}\frac{\rho(w)}{w}dw\right\}  \label{eq:l10}
\end{align}
By replacing, $\frac{\rho(w)}{w}=\frac{i\alpha}{\theta-iw}$, we have 
\begin{align}
\Phi(u,\xi)=\left(\frac{\theta - i\xi e^{-\lambda u}}{\theta -i \xi}\right)^{\alpha} \label{eq:l11}
\end{align}
$\vartheta(\xi)$ is continuous at zero, and we have: 
\begin{align}
\vartheta(\xi)=\lim_{u\to\infty}\vartheta(\xi e^{-\lambda u})\Phi(u,\xi)=\left(\frac{1}{1-i \frac{1}{\theta}\xi}\right)^{\alpha}=\left(1-i{\theta}^{-1} \xi\right)^{-\alpha} \label {eq:l12}
\end{align}
From (\ref{eq:l12}), $\vartheta(\xi)$ is the function characteristics of the gamma distribution; and the stationary marginal distribution of $\sigma^{2}(t)$ is the $\Gamma(\alpha, \theta)$ Gamma distribution. \\ Another method developed in \cite{barndorff1999non,barndorff2001non,barndorff2003integrated,barndorff2001modelling,barndorff2002financial} uses the relationship between the $z(t)$ L\'evy density $w(x)$ and the L\'evy density $u(x)$ of $\sigma^{2}(t)$. 
\begin{align}
u(x)=\int_{1}^{\infty}w(xr) dr \label {eq:l13}
  \end{align} 
From (\ref {eq:l09}), we have the L\'evy density $W(x)=\alpha f(x)=\alpha\theta e^{-\theta x}$ and the L\'evy density $u(x)$ of $\sigma^{2}(t)$ can be deduced as follows.
\begin{align}
u(x)=\int_{1}^{\infty}\alpha\theta e^{-\theta xr}dr=\frac{\alpha}{ x}e^{-\theta x} \quad  x>0 \label {eq:l14}
\end{align} 
u(x) is the L\'evy density of Gamma distribution $\Gamma(\alpha, \theta)$.
\end{proof}
\noindent
We can integrate the stationary non-negative process $\sigma^{2}(t)$.
\begin{align}
\sigma^{2*}(t)=\int_{0}^{t}\sigma^{2}(s)ds \label {eq:l6}
\end{align}  
\noindent
By  integration by part method, (\ref {eq:l6}) becomes
\begin{align}
\sigma^{2*}(t)&=\lambda^{-1}\sigma^{2}(0)(1-e^{-\lambda t}) + \lambda^{-1}\int_{0}^{t}\left(1- e^{-\lambda (t-s)}\right)dz(\lambda s)\\
&=\lambda^{-1}\left( -\sigma^{2}(t) + z(\lambda t) + \sigma^{2}(0)\right)\label {eq:l7}
\end{align}  
It results from  (\ref {eq:l7}) that the process $\sigma^{2*}(t)$ is continuous as $z(\lambda t)$ and $\sigma^{2}(t)$ co-break \cite{barndorff2002financial, barndorff2001non}. In addition,  the shape of $\sigma^{2*}(t)$ is determined by $z(\lambda t)$. In fact, $\sigma^{2*}(t)$ and $z(\lambda t)$ co-integrate. The co-integration can be shown by transforming the equation (\ref {eq:l7}) into (\ref {eq:l7a}). $\lambda\sigma^{2*}(t) - z(\lambda t)$ is a stationary process such that.
\begin{align}
\lambda \sigma^{2*}(t) - z(\lambda t) = -\sigma^{2}(t) + \sigma^{2}(0) \label {eq:l7a}
\end{align} 
\noindent
For $\lambda=1$ and $\sigma^{2}(0)=0$, the compound poison process ($z(t)$), the $\Gamma(\alpha, \theta)$ Ornstein-Uhlenbeck process in ($\sigma^{2}(t)$), and $\sigma^{2*}(t)$ in (\ref {eq:l8a}) were simulated and the results are in  Fig \ref{fig01}, Fig \ref{fig02}, and Fig \ref{fig03} respectively. 
\begin{align}
z(t)=\sum_{k=1}^{N(t)}\xi_{k}  \quad \quad \sigma^{2}(t)=\sigma^{2}(0)e^{\lambda t} + \sum_{k=1}^{N(t)}exp(-\lambda (t-a_{k}))\xi_{k} \quad \quad \sigma^{2*}(t)=\int_{0}^{t}\sigma^{2}(s)ds \label {eq:l8a} 
\end{align}
\begin{figure}[ht]
    \centering
\hspace{-1.2cm}
  \begin{subfigure}[b]{0.34\linewidth}
    \includegraphics[width=\linewidth]{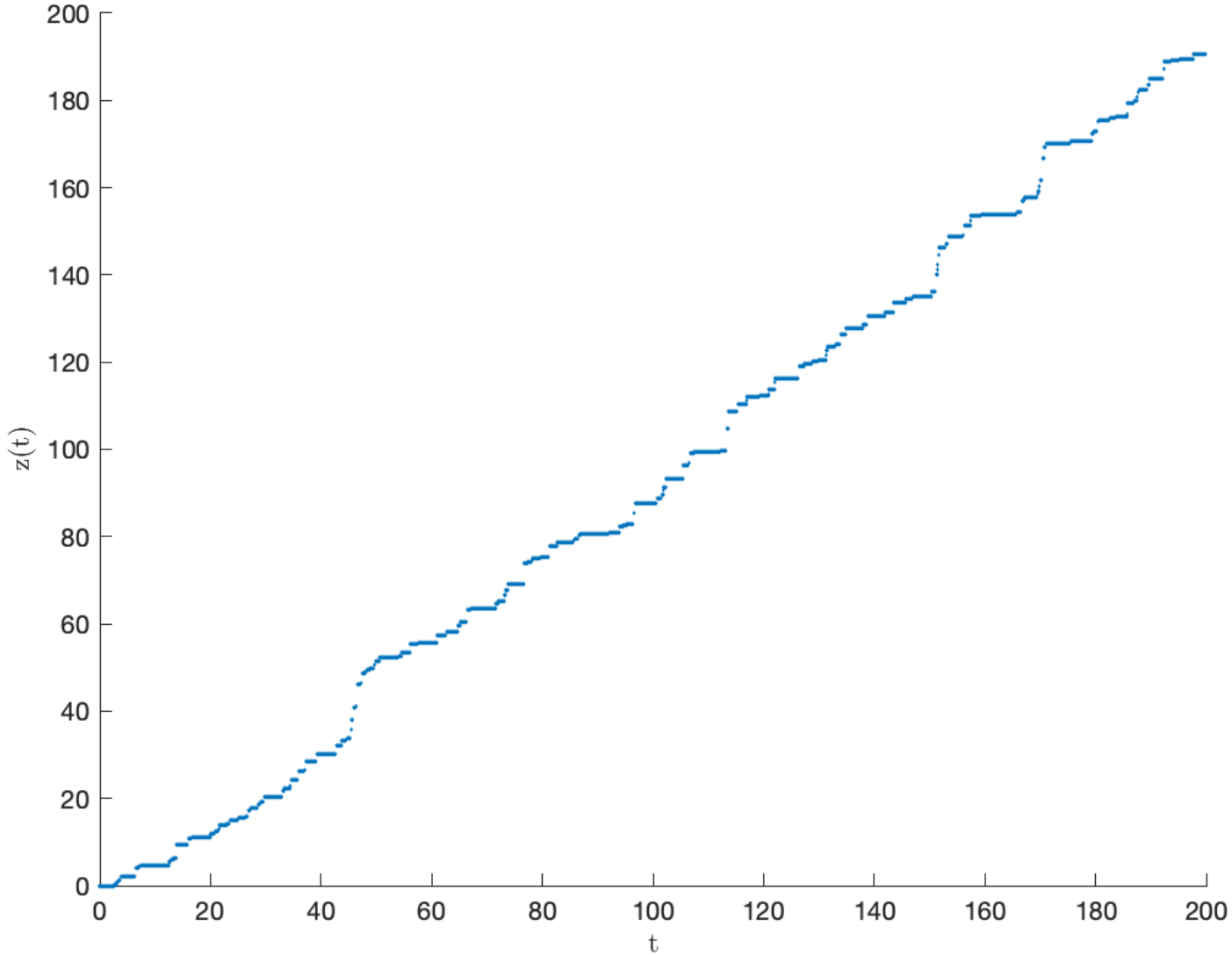}
\vspace{-0.5cm}
     \caption{ Compound Poison process: $\hat{z}(t)$}
         \label{fig01}
  \end{subfigure}
\hspace{-0.5cm}
  \begin{subfigure}[b]{0.34\linewidth}
    \includegraphics[width=\linewidth]{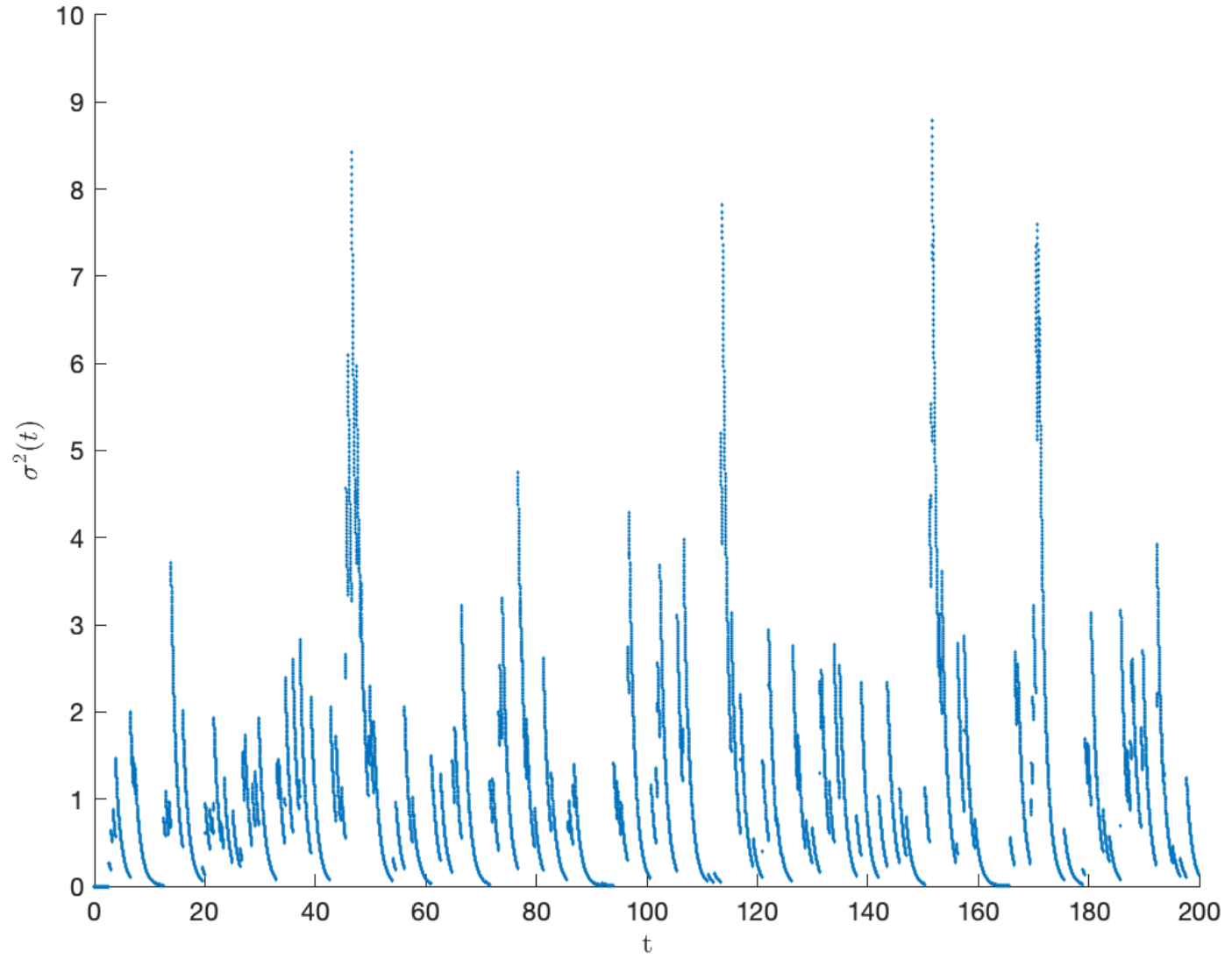}
\vspace{-0.5cm}
     \caption{Gamma process: $\hat{\sigma}^{2}(t)$}
         \label{fig02}
          \end{subfigure}
\hspace{-0.5cm}
  \begin{subfigure}[b]{0.34\linewidth}
    \includegraphics[width=\linewidth]{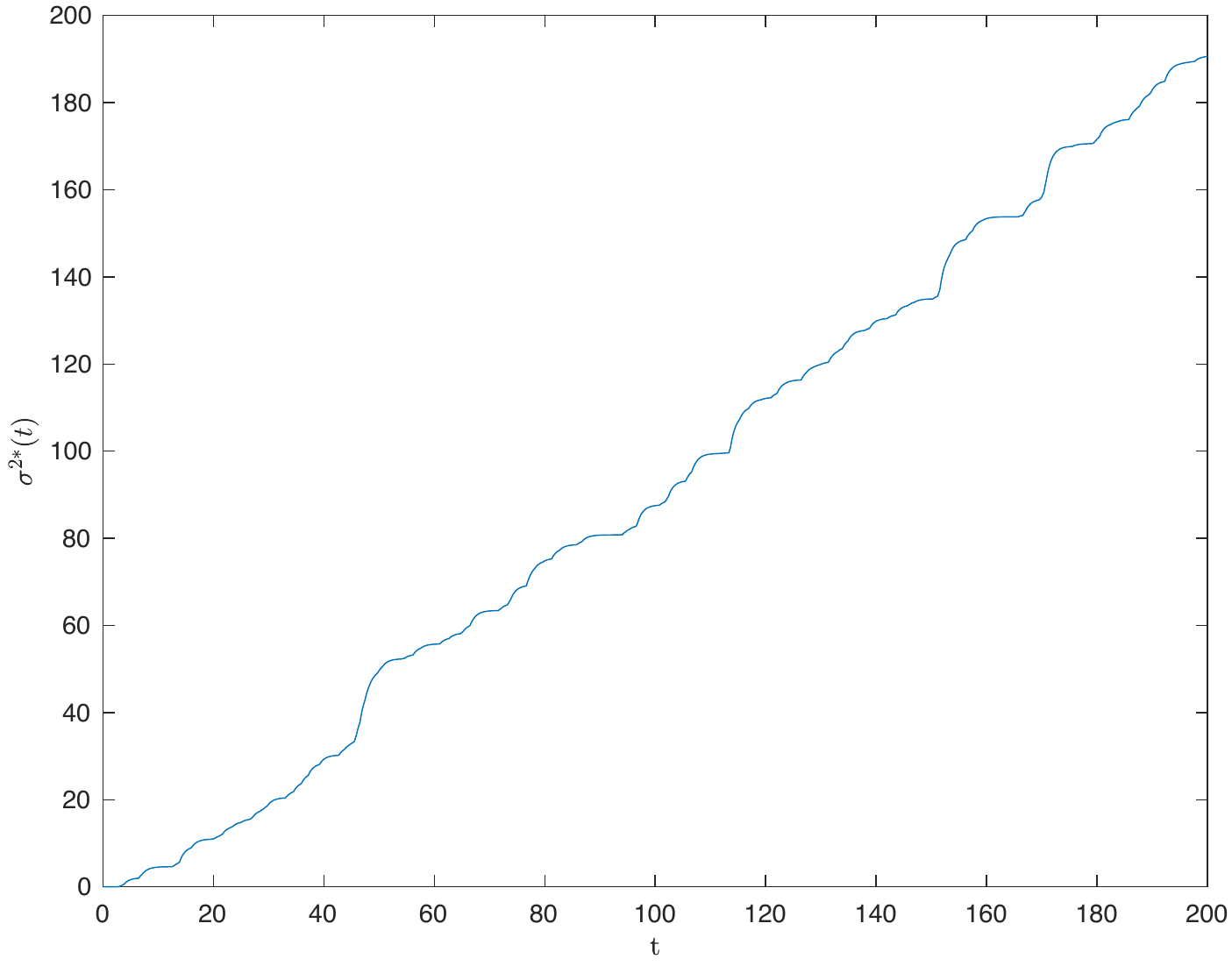}
\vspace{-0.5cm}
     \caption{Subordinator: $\hat{\sigma}^{2*}(t)$}
         \label{fig03}
          \end{subfigure}
\vspace{-0.9cm}
  \caption{Simulations: $\hat{\alpha}=0.8845$, $\hat{\theta}= 0.9378$}
  \label{fig0}
\vspace{-0.5cm}
\end{figure}

\noindent
The estimations of the gamma distribution parameter $\Gamma(\alpha, \theta)$ were performed by the FRFT Maximum likelihood on the daily SPY prices \cite{nzokem2021fitting}.
\subsection{Variance - Gamma Process: Semi-Martingale}

\noindent
Let  $Y^{*}= \{Y_{t}^{*}\}$, a stochastic process used to model the log of an asset price.
 \begin{align}
Y_{t}^{*}= A_{t} + M_{t}  \quad  & \quad  A_{t}=\beta t + \delta \sigma^{2*} (t)  \label {eq:l911} \\
M_{t}&= \sigma \int_{0}^{t}\sigma(s) dW(s) \label {eq:l913}
\end{align} 
where $\beta$ and $\delta$ are the drift parameters, $t$ represents the continuous time clock, and $W(t)$ is the standard Brownian motion and independent of  $\sigma^{2}(t)$.
\begin{align}
\sigma(t)=\sqrt{\sigma^{2} (t)} \quad \quad
\sigma^{2*}(t)&=\int_{0}^{t}\sigma^{2}(s)ds \label {eq:l914}
\end{align}
 $\sigma (t)$ is the spot or instantaneous volatility, and $\sigma^{2*}(t)$ is the chronometer or the integrated variance of the process. As shown in Fig \ref{fig03}, the Gamma process  ($\sigma^{2*}(t)$) is a strictly increasing process of the stationary process ($\sigma^{2}(t)$). \\
\noindent
The mean process $A_{t}$ is a predictable process with locally bounded variation. In fact, $ A_{t}$ is continuous  and differentiable because of $\sigma^{2*}(t)$.\\
\noindent
$M_{t}$ is a local martingale. The derivative of $M_{t}$ in (\ref {eq:l913}) can be written as a Stochastic Differential Equation (SDE)  (\ref {eq:l915})
 \begin{align}
dM_{t}&= \sigma \sigma(t) dW(t) \label {eq:l915}
\end{align} 
$Y_{t}^{*}$ is a special semi-martingale \cite{protter2005, barndorff2002financial} and the decomposition $Y_{t}^{*}=A_{t} + M_{t}$  is unique.
\begin{figure}[ht]
    \centering
  \begin{subfigure}[b]{0.45\linewidth}
    \includegraphics[width=\linewidth]{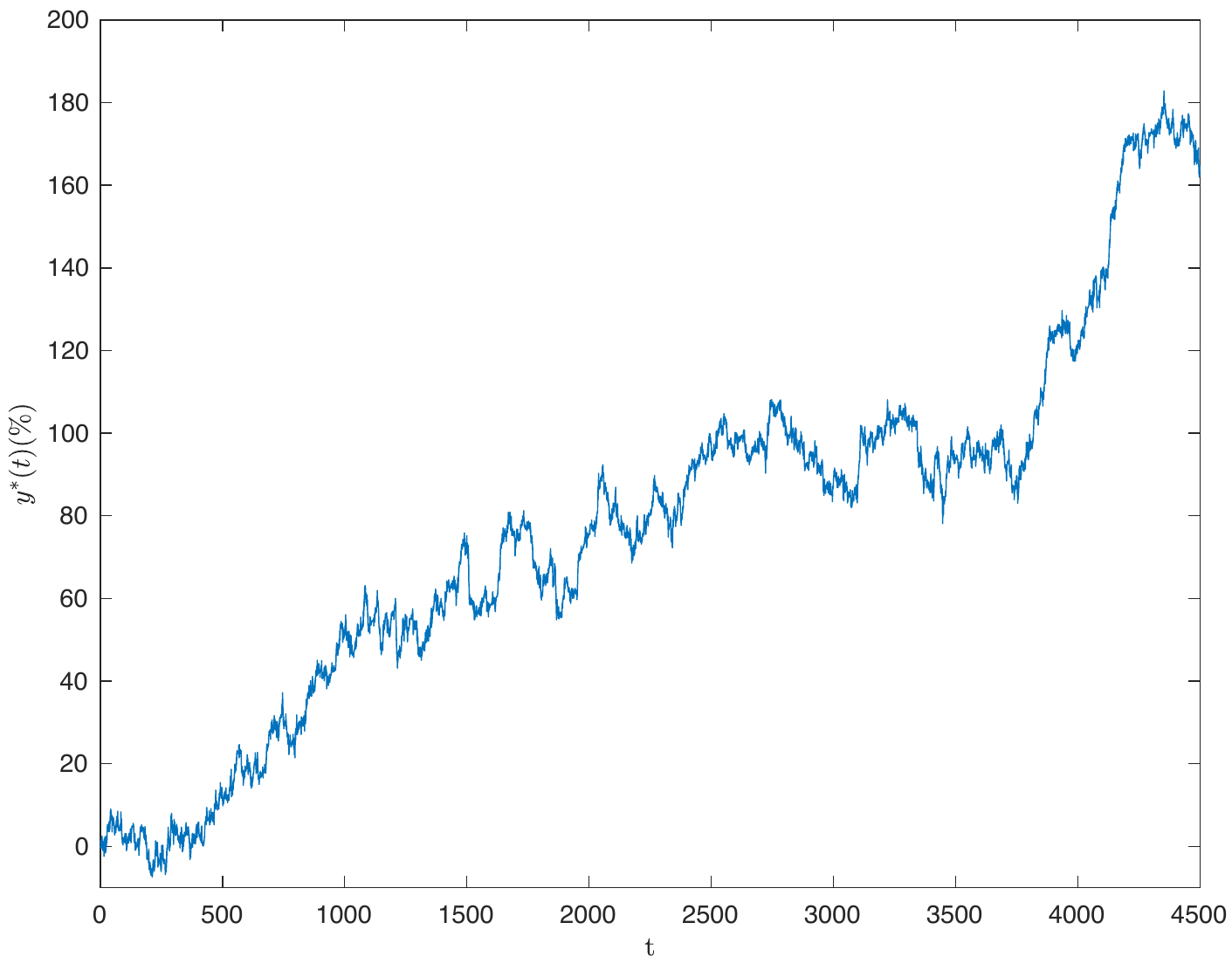}
\vspace{-0.5cm}
     \caption{ Simulations: $Y^{*}= \{Y_{t}^{*}\}$}
         \label{fig05}
  \end{subfigure}
  \begin{subfigure}[b]{0.45\linewidth}
    \includegraphics[width=\linewidth]{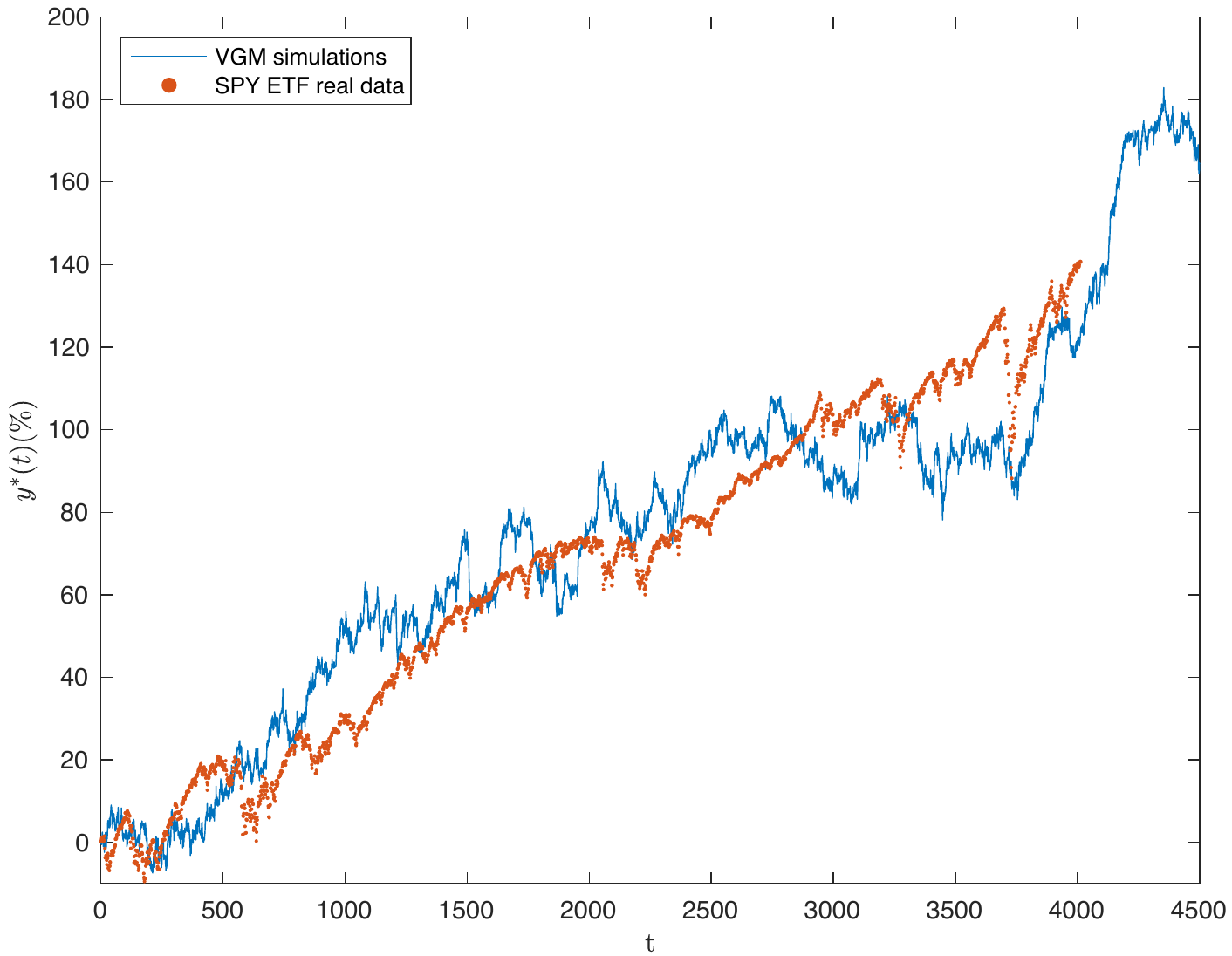}
\vspace{-0.5cm}
     \caption{Simulations versus SPY ETF data}
         \label{fig06}
          \end{subfigure}
\vspace{-0.5cm}
  \caption{VG Model: $\hat{\mu}=0.0848$, $\hat{\delta}=-0.0577$, $\hat{\sigma}=1.0295$, $\hat{\alpha}=0.8845$, $\hat{\theta}= 0.9378$}
  \label{fig07}
\vspace{-0.5cm}
\end{figure}

\noindent
Fig \ref{fig05} and Fig \ref{fig06} display, in blue color, the simulation of the logarithmic of the asset price ($Y^{*}$) in (\ref {eq:l911}). The simulation is compared with the daily SPY ETF historical return data from January 4, 2010, to December 30, 2020. The daily SPY ETF data is displayed in red color in Fig \ref{fig05}.

\subsection{Variance - Gamma Process: Parameter Estimations} 
\noindent
The stochastic process in (\ref {eq:l911})  is the solution of the following Stochastic Differential Equation (SDE):
 \begin{align}
dY_{t}^{*}= (\beta + \delta \sigma^{2} (t))dt + \sigma\sigma(t) dW(t)
\label {eq:l96}
\end{align}

\noindent
Given an interval of length $\Delta$,  we define $\sigma^{2}_{n}$ and $Y_{n}$ over the interval [$(n-1)\Delta$; $n\Delta$].
 \begin{align}
 \sigma^{2}_{n}=\int_{(n-1)\Delta}^{n\Delta}d\sigma^{2*}(s)= \sigma_{n\Delta}^{2*} -  \sigma_{(n-1)\Delta}^{2*} \quad 
Y_{n}=  \int_{(n-1)\Delta}^{n\Delta}dY_{s}^{*}=Y_{n\Delta}^{*} - Y_{(n-1)\Delta}^{*} \label {eq:l98}
\end{align}
The volatility component can be transformed into a normally distributed variable $X(1)$ as follows:
\begin{equation}
 \begin{aligned}
\int_{(n-1)\Delta}^{n\Delta}\sigma(t) dW(t) & \overset{\text{d}}{=} N\left(0, \int_{(n-1)\Delta}^{n\Delta}\sigma^{2}(s)ds \right)  = N\left(0, \sigma_{n\Delta}^{2*} -  \sigma_{(n-1)\Delta}^{2*}\right) = N\left(0, \sigma^{2}_{n}\right) \\
  & \overset{\text{d}}{=} {\sigma_{n}}N\left(0,1\right)  \overset{\text{d}}{=} {\sigma_{n}} X(1). \label {eq:l99}
   \end{aligned}
   \end{equation}
where $X(1)\overset{\text{d}}{=}N(0,1)$ and  $N(0,1)$ denotes a standard normal distribution.\\

\noindent
By integrating the instantaneous return rate  (\ref {eq:l96}) per component, we have:
\begin{align*}
\int_{(n-1)\Delta}^{n\Delta}dY_{s}^{*}=\beta \Delta + \delta\int_{(n-1)\Delta}^{n\Delta}d\sigma^{2*}(s) +  \sigma\int_{(n-1)\Delta}^{n\Delta}\sigma(t) dW(t)
 \end{align*}
\noindent
Based on (\ref{eq:l98}) and (\ref{eq:l99}), we have the following equation over the interval [$(n-1)\Delta$; $n\Delta$]
 \begin{align}
Y_{n} = \mu +  \delta\sigma^{2}_{n} + \sigma{\sigma_{n}} X(1)  \quad  \quad \quad  \mu=\beta\Delta \quad \quad  \sigma^{2}_{n}\overset{\text{d}}{=}\Gamma(\alpha, \theta) \label {eq:l10}
\end{align}

\noindent
In case $\Delta$ is a daily length, $Y_{n}$ becomes the daily return rate. The equation (\ref{eq:l10}) was analyzed in \cite{nzokem2021fitting,nzokem_2021b} as a daily return rate, and the parameters were estimated. The data came from the daily SPY ETF historical data and the period spans from January 4, 2010, to December 30, 2020. See \cite{nzokem2021fitting,nzokem_2021b,nzokem2021sis} for more details on the methodology and the results. \\

\noindent
Table \ref {tab1} presents the estimation results of the five parameters $(\mu, \delta,\alpha,\theta,\sigma)$ of $Y_{n}$ in (\ref{eq:l10}) along with some statistical indicators.
\begin{table}[ht]
\caption{FRFT Maximum Likelihood VG Parameters Estimations}
 \label{tab1}
\centering
\begin{tabular}{@{}lll@{}}
\toprule
\textbf{Model} &
 \textbf{Parameters} &
  \textbf{Statistics} \\ \midrule
 \multirow{2}{*}{} &
  \multirow{2}{*}{$\hat{\mu}=0.0848$} &
  \multirow{2}{*}{$\hat{E(Y)}=0.0369 $} \\
 \multirow{2}{*}{} &
  \multirow{2}{*}{$\hat{\delta}=-0.0577$} &
  \multirow{2}{*}{$\hat{Var(Y)}=0.8817$} \\
  \multirow{2}{*}{\textbf{VG}} &
  \multirow{2}{*}{$\hat{\sigma}=1.0295$} &
  \multirow{2}{*}{$\hat{Skew(Y)}=-0.173 $} \\ 
  \multirow{2}{*}{} &
  \multirow{2}{*}{$\hat{\alpha}=0.8845$} &
  \multirow{2}{*}{$\hat{Kurt(Y)}=6.412$} \\  
  \multirow{2}{*}{} &
  \multirow{2}{*}{$\hat{\theta}=0.9378$} &
  \multirow{2}{*}{} \\  
\multirow{2}{*}{} &
  \multirow{2}{*}{} &
  \multirow{2}{*}{} \\
  \bottomrule
  \footnotesize Source: Nzokem(2021) \cite{nzokem2021fitting,nzokem_2021b}
 \end{tabular}
\end{table}

\noindent
As shown in Table \ref {tab2}, with initial parameter values ($\sigma = \alpha = \theta= 1$, $\delta=\mu=0$),  the maximization procedure convergences after 21 iterations. $\log(ML)$ is the function to maximize and $||\frac{d\log(ML)}{dV}||$ is the norm of the partial derivative function ($\log(ML)$). During the maximization process, both quantities converge respectively to $-3549.692$ and $0$; where the parameter vector is stable. The location parameter $\mu$ is positive, the symmetric parameter $\delta$ is negative, and other parameters have the expected sign.
\begin{table}[ht]
\vspace{-0.3cm}
\caption{Results of VG Model Parameter Estimations}
 \label{tab2}
\centering
\centering
\resizebox{13cm}{!}{%
\begin{tabular}{cccccccc}
\textbf{Iterations} &
  \textbf{$\mu$} &
  \textbf{$\delta$} &
  \textbf{$\sigma$} &
  \textbf{$\alpha$} &
  \textbf{$\theta$} &
  \textbf{$Log(ML)$} &
  \textbf{$||\frac{dLog(ML)}{dV}||$} \\
1  & 0          & 0          & 1          & 1          & 1          & -3582.8388 & 598.743231 \\
2  & 0.05905599 & -0.0009445 & 1.03195903 & 0.9130208  & 1.03208412 & -3561.5099 & 833.530396 \\
3  & 0.06949925 & 0.00400035 & 1.04101444 & 0.88478895 & 1.05131996 & -3559.5656 & 447.807305 \\
4  & 0.07514039 & 0.00055771 & 1.17577397 & 0.67326429 & 1.17778666 & -3569.6221 & 211.365781 \\
5  & 0.08928373 & -0.0263716 & 1.03756321 & 0.83842661 & 0.94304967 & -3554.4434 & 498.289445 \\
6  & 0.08676498 & -0.0521887 & 1.03337015 & 0.85591875 & 0.95066351 & -3550.6419 & 204.467192 \\
7  & 0.086995   & -0.0608517 & 1.02788937 & 0.87382621 & 0.95054954 & -3549.8465 & 66.8039738 \\
8  & 0.08542912 & -0.058547  & 1.02705241 & 0.88258411 & 0.94321299 & -3549.7023 & 15.3209117 \\
9  & 0.08478622 & -0.0576654 & 1.02995166 & 0.88447791 & 0.93670036 & -3549.6921 & 1.14764198 \\
10 & 0.08477798 & -0.0577736 & 1.02922308 & 0.88449072 & 0.93831041 & -3549.692  & 0.17287708 \\
11 & 0.08476475 & -0.0577271 & 1.02960343 & 0.88450434 & 0.93755549 & -3549.692  & 0.07850459 \\
12 & 0.08477094 & -0.0577488 & 1.02942608 & 0.8844984  & 0.93790784 & -3549.692  & 0.03723941 \\
13 & 0.08476804 & -0.0577386 & 1.02950937 & 0.88450117 & 0.93774266 & -3549.692  & 0.01732146 \\
14 & 0.0847694  & -0.0577434 & 1.02947043 & 0.88449987 & 0.93781995 & -3549.692  & 0.00813465 \\
15 & 0.08476876 & -0.0577411 & 1.02948868 & 0.88450048 & 0.93778375 & -3549.692  & 0.00380345 \\
16 & 0.08476906 & -0.0577422 & 1.02948014 & 0.88450019 & 0.9378007  & -3549.692  & 0.00178206 \\
17 & 0.08476892 & -0.0577417 & 1.02948414 & 0.88450033 & 0.93779276 & -3549.692  & 0.00083415 \\
18 & 0.08476898 & -0.0577419 & 1.02948226 & 0.88450026 & 0.93779648 & -3549.692  & 0.00039063 \\
19 & 0.08476895 & -0.0577418 & 1.02948314 & 0.88450029 & 0.93779474 & -3549.692  & 0.00018289 \\
20 & 0.08476897 & -0.0577419 & 1.02948273 & 0.88450028 & 0.93779555 & -3549.692  & 8.56E-05   \\
21 & 0.08476896 & -0.0577418 & 1.02948292 & 0.88450029 & 0.93779517 & -3549.692  & 4.01E-05  
\end{tabular}%
}
\end{table}

\section{Variance - Gamma Process: Probability  versus L\'evy Density} \ \\
\noindent
Based on (\ref{eq:l911}) and (\ref{eq:l913}), the  VG Process $Y=\{Y_{t}\}_{t\geq0}$  with five parameters $(\mu, \delta, \sigma, \alpha,\theta)$ can be written as follows:
 \begin{align}
Y_{t}= \mu t + \delta \sigma^{2*} (t) + \sigma \int_{0}^{t}\sigma(s) dW(s). \label {eq:l939}
\end{align} 
where $\mu, \delta \in R$ , $\sigma>0$, $\alpha>0$, $\theta>0$, $t$ represents the continuous time clock, and $W(t)$ is the standard Brownian motion and independent of  $\sigma^{2}(t)$. 
\begin{align}
\sigma(t)=\sqrt{\sigma^{2} (t)} \quad \quad
\sigma^{2*}(t)&=\int_{0}^{t}\sigma^{2}(s)ds. \label {eq:l949}
\end{align}
where $\sigma (t)$ is the spot or instantaneous volatility, $\sigma^{2} (t)$ is the spot or instantaneous variance, and $\sigma^{2*}(t)$ is the chronometer or the integrated variance of the process.\\

\noindent
We consider the characteristic function of the VG process $Y= \{Y_{t}\}$
\begin{align}
E\left[e^{i \xi Y_{t}}\right]= E\left[e^{i \xi (\mu t + \delta \sigma^{2*} (t) + \sigma \int_{0}^{t}\sigma(s) dW(s))}\right]= e^{i t\mu \xi}E\left[e^{i \xi (\delta \sigma^{2*} (t) + \sigma \int_{0}^{t}\sigma(s) dW(s))}\right].\label {eq:l9399}  \end{align}
$\int_{0}^{t}\sigma(s) dW(s)$ is the It$\hat{o}$ integral with respect to the Brownian motion, and we have :
\begin{align}
\int_{0}^{t}\sigma(s)dW(s) &\overset{\text{d}}{=} N\left(0, \int_{0}^{t}\sigma^{2}(s)ds \right) = N\left(0, \sigma^{2*}(t) \right).\label{eq:l93992}  \end{align}
where $N(0,1)$ is a standard normal distribution.\\

\noindent
From expressions (\ref{eq:l9399}) and (\ref{eq:l93992}), we have 
\begin{equation}
\begin{aligned}
E\left[e^{i \xi (\delta \sigma^{2*} (t) + \sigma \int_{0}^{t}\sigma(s) dW(s))}\right]&=E\left[e^{i \xi N\left(\delta \sigma^{2*} (t), \sigma^2\sigma^{2*}(t)\right)}\right] =E\left[E\left[e^{i \xi N\left(\delta \sigma^{2*} (t), \sigma^2\sigma^{2*}(t)\right)}|\sigma^{2*}(t)\right]\right] \\
&=E\left[e^{(i\delta\xi - \frac{1}{2}\sigma^2\xi^2)\sigma^{2*}(t)}\right]  \label{eq:l93995} 
  \end{aligned}
\end{equation}
$\sigma^{2*}(t)$ is a L\'evy process generated by the Gamma distribution $\Gamma(\alpha, \theta)$ and we have 
\begin{equation}
\begin{aligned}
E\left[e^{(i\delta\xi - \frac{1}{2}\sigma^2\xi^2)\sigma^{2*}(t)}\right]&=\frac{1}{(1 + \frac{1}{2}\theta\sigma^2\xi^2)^{t\alpha}}E\left[e^{i\delta\xi W }\right] \quad  \hbox{$ \sigma^{2*}(t)\overset{\text{d}}{=} \Gamma(t \alpha, \theta)$} \\
&=\frac{1}{\left( 1 - i\delta\theta\xi + \frac{1}{2}\sigma^2\theta\xi^2\right)^{t\alpha}} \quad  \hbox{$ W\overset{\text{d}}{=} \Gamma(t \alpha,  \frac{1}{2}\sigma^2\xi^2 +\frac{1}{\theta})$} \label{eq:l93998} 
\end{aligned}
\end{equation}
From expressions (\ref{eq:l9399}), (\ref{eq:l93995}) and (\ref{eq:l93998}), we have 
\begin{align}
E\left[e^{i Y_{t}\xi}\right]&= \frac{e^{i t\mu \xi}}{\left( 1 - i\delta\theta\xi + \frac{1}{2}\sigma^2\theta\xi^2\right)^{t\alpha}} \label{eq:l93981}
 \end{align}
We define two related functions $\phi(\xi)$ and $\varphi(\xi,t)$ such that 
\begin{align} 
\phi(\xi)= \frac{e^{i\mu \xi}}{\left( 1 - i\delta\theta\xi +\frac{1}{2} \sigma^2\theta\xi^2\right)^{\alpha}} \quad  & \quad  \quad  \varphi(\xi,t)= -Log(E\left[e^{i Y_{t}\xi}\right])= - t Log(\phi(\xi))\label{eq:l93982}
  \end{align}
  The characteristic function can be written as follows.
  \begin{align} 
    E\left[e^{i Y_{t}\xi}\right]=\left( \phi(\xi)\right)^{t}=E[e^{- t Log(\phi(\xi))}] \label{eq:l93983}  
  \end{align}
    
      \bigskip
  \subsubsection{L\'evy measure and the structure of the jumps} 
\begin{lemma} \ \\
    \label{lem4}
 (Frullani integral) $\forall \alpha, \beta>0$ and $\forall z \in \scrC$ with $\Re (z) \leq 0$.\\
we have
\begin{align*}
 \frac{1}{\left(1-\frac{z}{\alpha}\right)^{\beta}}=e^{-\int_{0}^{\infty}(1-e^{zx})\beta x^{-1}e^{-\alpha x}dx}
  \end{align*}
\end{lemma} 
\noindent
For lemma proof, see  \cite{arias1990theorem}

\begin{theorem}\label{lem5}  (Variance-Gamma Model  representation) \ \\
 Let  $Y=\{Y_{t}\}_{t\geq0}$  be a L\'evy process on $\mathbb{R}$ generated by the VG model  with parameter $(\mu, \delta, \sigma, \alpha, \theta)$. The characteristic exponent of the L\'evy process has the following representation.
 \begin{align}
\varphi(\xi,1)&=- Log\left(Ee^{i Y_{1}\xi}\right)= i\mu \xi + \int_{-\infty}^{+\infty}\left(1-e^{-i\xi u}\right)\Pi(u)du \label {eq:l11}
  \end{align}
\noindent
$\Pi(u)$ is the L\'evy density  of  $Y$ :
\begin{align}
 \Pi(u)=\alpha\left( \frac{1_{\{ u>0 \} }}{u}e^{-x_{1}u} + \frac{1_{\{ u<0 \} }}{\lvert u\rvert}e^{{- x_{2}}u}\right) \label {eq:l111}
   \end{align}
\noindent
 with
   \begin{align}
x_{1}=\frac{\delta}{\sigma^2}+ \sqrt{\frac{\delta^2}{\sigma^4}+\frac{2}{\theta \sigma^2}} \quad \quad x_{2}= \frac{\delta}{\sigma^2} - \sqrt{\frac{\delta^2}{\sigma^4}+\frac{2}{\theta \sigma^2}} \label {eq:l112}
   \end{align}
\noindent
and $\Pi(u)$ satisfies the properties
\begin{align}
  \int_{-\infty}^{+\infty}\Pi(u)du =+\infty  \quad \hbox{and}  \quad  \int_{-\infty}^{+\infty}Min(1,|u|)\Pi(u)du <+\infty \label {eq:l113}
   \end{align}
\end{theorem} 
\begin{proof} \ \\
\noindent
We consider the characteristic function $\phi(\xi)$ in (\ref{eq:l93982}) of the VG model with parameter $(\mu, \delta, \sigma, \alpha, \theta)$ developed previously
\begin{align*}
 \phi(\xi)= \frac{e^{i\mu \xi}}{\left( 1 - i\delta\theta\xi + \frac{1}{2} \sigma^2\theta\xi^2\right)^{\alpha}}
\end{align*}
\noindent
We factor the quadratic function in the denominator of $\phi(\xi)$.
\begin{align}
\left(1+\frac{1}{2}\theta \sigma^{2}x^{2} - i\delta\theta x\right)^{\alpha}=\left(\frac{1}{2}\theta \sigma^{2}\right)^{\alpha} \left(x - i x_{1}\right)^{\alpha}\left(x - i x_{2}\right)^{\alpha}. \label{eq:l114}
\end{align}
with
\begin{align*}
x_{1}=\frac{\delta}{\sigma^2}+ \sqrt{\frac{\delta^2}{\sigma^4}+\frac{2}{\theta \sigma^2}} \quad  \quad \quad 
x_{2}= \frac{\delta}{\sigma^2}- \sqrt{\frac{\delta^2}{\sigma^4}+\frac{2}{\theta \sigma^2}}
\end{align*}
\noindent
We apply the lemma \ref {lem4} on each factor of the quadratic function (\ref{eq:l114}).
\begin{align*}
\left(\frac{1}{2}\theta \sigma^{2}\right)^{\alpha} {\left(x - i x_{1}\right)^{\alpha}}{\left(x - i x_{2}\right)^{\alpha} }&={\left(1+\frac{i x}{x_{1}}\right)^{\alpha}}{\left(1+\frac{i x}{x_{2}}\right)^{\alpha}}\\
&=\left(e^{\int_{0}^{\infty}{(1-e^{-i x u})\frac{\alpha}{u}e^{-x_{1}u}du}}\right)  \left(e^{\int_{0}^{\infty}{(1-e^{ix u})\frac{\alpha}{u}e^{x_{2}u}du}}\right)\\
&=e^{\int_{0}^{\infty}{(1-e^{-i x u})\frac{\alpha}{u}e^{-x_{1}u}du} +  \int_{-\infty}^{0}{(1-e^{-i x v})\frac{\alpha}{|v|}e^{-x_{2}v}dv}}\\
&=e^{\int_{-\infty}^{+\infty}{(1-e^{-i x u})\Pi(u)du}}
\end{align*}
\noindent
we take into account  the expression (\ref {eq:l114}) and have
\begin{align}
\left(1+\frac{1}{2}\theta \sigma^{2}x^{2} - i\delta\theta x\right)^{\alpha}=e^{\int_{-\infty}^{+\infty}{(1-e^{-i x u})\Pi(u)du}}.\label {eq:l1114}
\end{align}
where  $\Pi(u)=\alpha\left( \frac{1_{\{ u>0 \}}}{u} e^{-x_{1}u} + \frac{1_{\{ u<0 \} }}{\lvert u\rvert}e^{{-x_{2}}u}\right)$ \\
 \noindent
 From expression (\ref {eq:l93982}), we have:
 \begin{align*}
 \varphi(\xi,t)= - t Log(\phi(\xi))&= -it\mu\xi +t Log\left(1+\frac{1}{2}\theta \sigma^{2}x^{2} - i\delta\theta x\right)^{\alpha}\\
 &= -it\mu\xi +t Log\left(1+\frac{1}{2}\theta \sigma^{2}x^{2} - i\delta\theta x\right)^{\alpha}\\
  &= -it\mu\xi + \int_{-\infty}^{+\infty}{(1-e^{-i x u})t\Pi(u)du}
  \end{align*}
 \noindent
 We have 
 \begin{align}
 \varphi(\xi,t)= - t Log(\phi(\xi))= -it\mu\xi + \int_{-\infty}^{+\infty}{(1-e^{-i x u})t\Pi(u)du}\label {eq:l1134}  \end{align} 
\noindent  
For $t=1$, we have the expression (\ref{eq:l11})\\
\noindent
We can check some properties of  $\Pi(u)$
 \begin{align}
 \int_{-\infty}^{+\infty}\Pi(u) du &=+\infty \quad \quad  \hbox{ in fact  \quad  $\lim_{|u| \to 0} \Pi(u) = +\infty$ } \label {eq:l115}
  \end{align}
  \begin{align*}
  \int_{-\infty}^{+\infty}Min(1,|u|)\Pi(du) &=\int_{-1}^{1}Min(1,|u|)\Pi(du) + \int_{1}^{+\infty}Min(1,|u|)\Pi(du) \\ &+\int_{-\infty}^{-1}Min(1,|u|)\Pi(du) \\
  &=\alpha \left(\frac{1-e^{-x_{1}}}{x_{1}} + \frac{1-e^{x_{2}}}{-x_{2}} + \Gamma(0,x_{1}) + \Gamma(0, -x_{2}) \right).   \end{align*}
with \ \ $\Gamma(s,u)=\int_{u}^{+\infty}{y^{s-1} }e^{-y}dy$  \\
And we have:   
    \begin{align}
 \int_{-\infty}^{+\infty}Min(1,|u|)\Pi(du) <+\infty \label {eq:l116}  
  \end{align}
 \end{proof}
 
 \noindent
The results in (\ref {eq:l115}) show that the VG process is not a finite activities process and can not be written as a Compound Poisson process \cite{barndorff2002financial}. The VG process is an infinite activity process with an infinite number of jumps in any given time interval. The arrival rate of jumps of all sizes in the VG process is defined by the L\'evy density (\ref{eq:l117})
 \begin{align}
   \Pi(u)= 
\begin{cases}
    \frac{\alpha}{\lvert u\rvert}{e^{-x_{2}u}}  & \text{if } u<0\\
    \frac{\alpha}{u} {e^{-x_{1}u}} & \text{if } u>0.  \label{eq:l117}
\end{cases}
 \end{align}
\noindent
As shown in Fig \ref{fig911}, the high arrival rates of jumps are concentrated around the origin $0$. The smaller the jump size, the higher the arrival rate for the VG model. The steepness parameters \cite{boyarchenko2002non},$-x_2$ and $x_1$, defined the rate of exponential decay of the tails on each side. As shown in Fig \ref{fig911} and (\ref {eq:l117}), the L\'evy density is asymmetric, and the left tail is heavier as $-x_2<x_1$. On the other hand, the result in (\ref {eq:l116}) proves that the VG process is a finite variation process, which is contrary to the Brownian motion process. The Gamma distribution Parameter ($\alpha$), called the process intensity \cite{boyarchenko2002non}, plays an important role in the L\'evy density. The intensity of the process $(\alpha)$ has a similar role as the variance parameter in the Brownian motion process. The L\'evy density function  (\ref{eq:l117}) is different for negative and positive jump size. The difference has led \cite{madan1998variance} to see the VG process as the difference between two increasing processes, with one process providing the upward movement and another the downward movement in the market.
\begin{figure}[ht]
\vspace{-0.3cm}
    \centering
  \begin{subfigure}[b]{0.4\linewidth}
    \includegraphics[width=\linewidth]{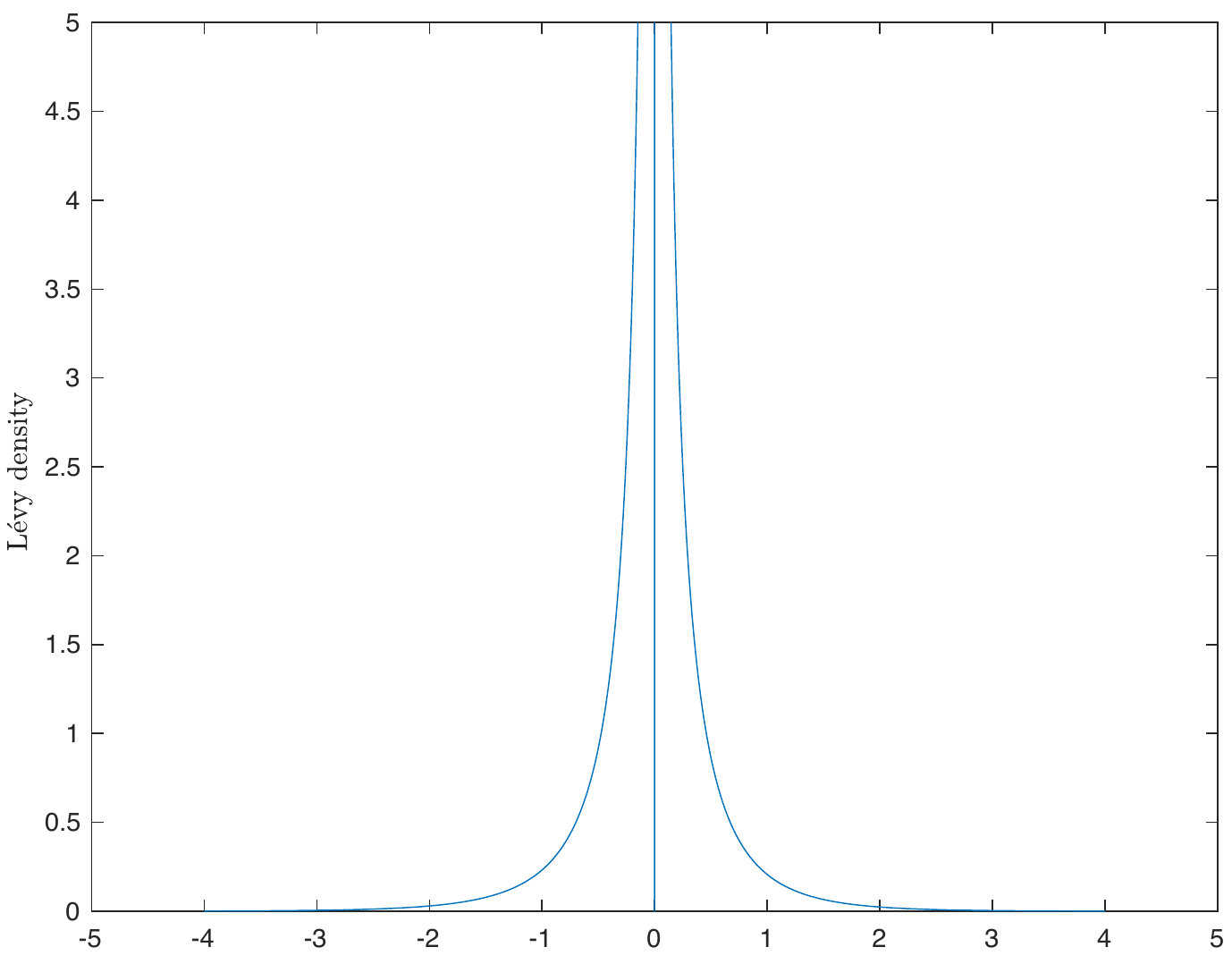}
\vspace{-0.5cm}
     \caption{L\'evy density of the VG model}
         \label{fig911}
  \end{subfigure}
  \begin{subfigure}[b]{0.4\linewidth}
    \includegraphics[width=\linewidth]{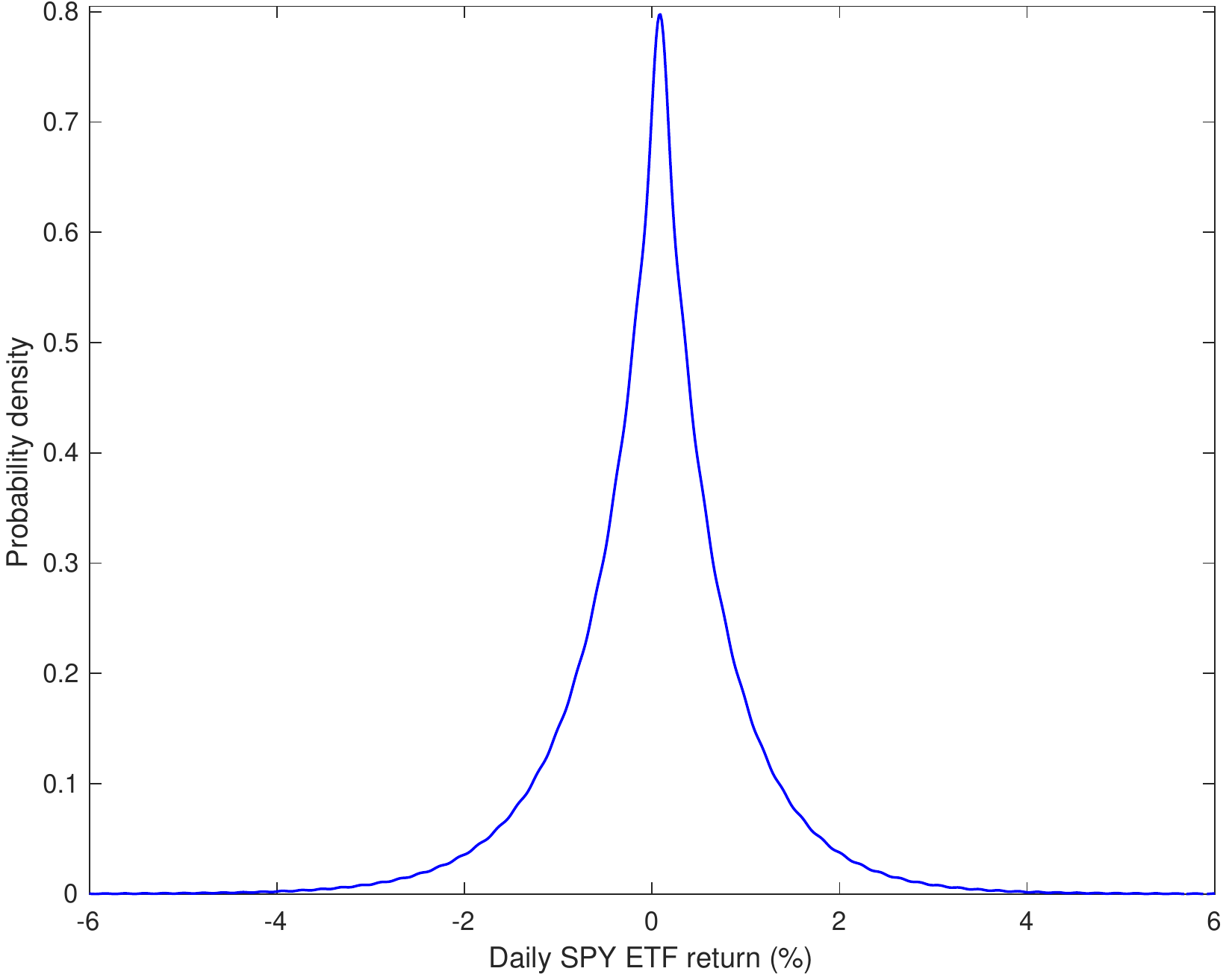}
\vspace{-0.6cm}
     \caption{Probability density of the VG model}
         \label{fig921}
          \end{subfigure}
\vspace{-0.7cm}
  \caption{VG Model: $\hat{\mu}=0.0848$, $\hat{\delta}=-0.0577$, $\hat{\sigma}=1.0295$, $\hat{\alpha}=0.8845$, $\hat{\theta}= 0.9378$}
  \label{fig912}
\vspace{-0.5cm}
\end{figure}

\noindent
Using the VG parameter estimations in table \ref{tab1}, we have  $x_{1}=1.4775$ and $x_{2}=-1.3640$. Fig \ref{fig911} and Fig \ref{fig921} display the L\'evy and the probability densities. As shown in Fig \ref{fig912}, the shape of the density functions are different; even-though, both densities are linked by the same characteristic function.\\
\noindent
Variance-Gamma (VG) Process can be described as a subfamily of the KoBoL family, which is the extension of Koponen's family by Boyarchenko and Levendorskii \cite{boyarchenko2002non}.  The KoBoL family is also called CGMY- model (named after Carr, German, Madan, and Yor) \cite{carr2003stochastic}. Under the KoBoL family, the L\'evy density has the following general form. See \cite{boyarchenko2002non} for more details
\begin{align}
   \Pi(u)= 
\begin{cases}
 C_{-}{|u|^{-\nu-1}} {e^{\lambda_{-}u}}  & \text{if } u<0 \\
 C_{+} {u^{-\nu-1}} {e^{-\lambda_{+}u}} & \text{if } u>0 \label{eq:l1171}
\end{cases}
 \end{align}
 where $C_{+}>0$, $C_{-}>0$, $\nu>0$ and $\lambda_{-}<0<\lambda_{+}$\\
  As a subfamily of the KoBoL family, the VG process belongs to the process class of order  $\nu=0$, intensity $C_{+}=C_{-}=\alpha$ and steepness parameters $\lambda_{-}=-x_{2}=-\frac{\delta}{\sigma^2} + \sqrt{\frac{\delta^2}{\sigma^4}+\frac{2}{\theta \sigma^2}}$ and $\lambda_{+}=x_{1}=\frac{\delta}{\sigma^2}+ \sqrt{\frac{\delta^2}{\sigma^4}+\frac{2}{\theta \sigma^2}}$. For $0<\nu<1$, see \cite{nzokemmay} for a general case of tempered stable distribution.
  
  \subsection{Variance - Gamma Process: Asymptotic distribution} 
 \begin{theorem}\label{lem6} (Variance-Gamma process probability density) \ \\
Let $Y=\{Y_{t}\}_{t\geq0}$  be a L\'evy process on $\mathbb{R}$ generated by the VG model  with parameter $(\mu, \delta, \sigma, \alpha, \theta)$. The probability density function can be written as follows
\begin{align}
f(y, t) &=\frac {1} {\sigma\Gamma(t\alpha) \theta^{t\alpha}}\int_{0}^{+\infty} \frac{1}{\sqrt{2\pi v}}e^{-\frac{(y-t\mu-\delta v)^2}{2v\sigma^2}}v^{t\alpha -1}e^{-\frac{v}{\theta}} \,dv  \quad \quad  t\geq 0  \quad  y \in \mathbb{R} \label{eq:l201}
 \end{align}
\end{theorem}
\begin{proof}[Proof:]  \ \\
$\varphi(\xi,t)$ in (\ref{eq:l1134}) provides the relation between  the characteristic exponent  and the L\'evy density. the expression is used as follows:
 \begin{align*}
\varphi(\xi,t)= - Log\left(Ee^{i Y_{t}\xi}\right)&= -it\mu\xi + \int_{-\infty}^{+\infty}{(1-e^{-i x u})t\Pi(u)du}\\
t \Pi(u)&=t\alpha\left( \frac{1_{\{ u>0 \} }}{u}e^{-x_{1}u} + \frac{1_{\{ u<0 \} }}{\lvert u\rvert}e^{{\lvert x_{2}\rvert}u}\right) 
   \end{align*}
 \begin{align}
\mu_{t}=t\mu \quad \quad  \alpha_{t}= t \alpha \label{eq:l202}
   \end{align}

\noindent
It was shown in \cite{nzokem2021fitting}  that the probability density of a VG model with parameter $(\mu, \delta, \sigma, \alpha, \theta)$ can be written
\begin{align*}
f(y) &=\frac {1} {\sigma\Gamma(\alpha) \theta^{\alpha}}\int_{0}^{+\infty} \frac{1}{\sqrt{2\pi v}}e^{-\frac{(y-\mu-\delta v)^2}{2v\sigma^2}}v^{\alpha -1}e^{-\frac{v}{\theta}} \,dv 
 \end{align*}
\noindent
By replacing the parameters in (\ref{eq:l202}), we have the result in Theorem \ref{lem6}.
\end{proof} 

 \begin{theorem}\label{lem7} (Asymptotic distribution of Variance-Gamma process) \ \\
Let $Y=\{Y_{t}\}_{t\geq0}$ be a L\'evy process on $\mathbb{R}$ generated by the VG model  with parameter $(\mu, \delta, \sigma, \alpha, \theta)$.\\
Then  $Y_{t}$  converges in distribution to a L\'evy process driving by a Normal distribution with mean $a=\mu + \alpha \theta \delta$ and variance $\sigma^2=\alpha (\theta^2\delta^2 + \sigma^2\theta)$.
\begin{align}
Y_{t}\overset{\text{d}}{\sim}N(ta,t\sigma^2) \quad &\text{as}\label{eq:l203} \quad t \to +\infty
 \end{align}
\end{theorem}
\begin{proof}[Proof:]  \ \\
Let us have
\begin{align*}
b_{t}& =\sqrt{t}b  \ & \ a_{t}&=ta\\
b&=\sqrt{\alpha (\theta^2\delta^2 + \sigma^2\theta)} \  & \ a&=\mu + \alpha \theta \delta 
 \end{align*}
\noindent
$\phi(\xi,t)$ is the characteristic function of the process $Y=\{Y_{t}\}_{t\geq0}$ , we use the expression (\ref{eq:l93981}).
\begin{align*}
\phi(\xi,t)=E\left[e^{i Y_{t}\xi}\right]&= \frac{e^{i t\mu \xi}}{\left( 1 - i\delta\theta\xi + \frac{1}{2}\sigma^2\theta\xi^2\right)^{t\alpha}} 
\end{align*}
\noindent
$\phi^{T}(\xi,t)$ is the characteristic function of the stochastic  process $\{\frac{{Y_{t}}-a_{t}}{b_{t}}\}_{t\geq0}$ and we have 
 \begin{align*}
\phi^{T}(\xi,t)= E\{e^{i \frac{{Y_{t}}-a_{t}}{b_{t}}\xi}\} =e^{-i \frac{a_{t}}{b_{t}}\xi}E\{e^{i \frac{\xi}{b_{t}}Y_{t}}\}&=e^{-i \frac{a_{t}}{b_{t}}\xi}\phi(\frac{\xi}{b_{t}},t)=\frac{e^{i t\alpha \theta \delta\frac{\xi}{b_{t}}}}{\left(1+\frac{1}{2}\theta \sigma^{2}\frac{\xi^2}{b_{t}^2} - i\delta\theta \frac{\xi}{b_{t}}\right)^{t\alpha}}\\
&=e^{i t\alpha \theta \delta\frac{\xi}{b_{t}}}{\left(1+\frac{1}{2}\theta \sigma^{2}\frac{\xi^2}{tb^2} - i\delta\theta \frac{\xi}{\sqrt{t}b}\right)^{-t\alpha}}
   \end{align*}
Let us have
\begin{align*}
u(t) =\frac{1}{2}\theta \sigma^{2}\frac{\xi^2}{tb^2} - i\delta\theta \frac{\xi}{\sqrt{t}b}  \quad \quad \quad  \lim_{t \to +\infty}u(t)=0 
 \end{align*}
We can use the Taylor expansions of $\ln(1+u)$
 \begin{align*}
 \ln(1+\frac{1}{2}\theta \sigma^{2}\frac{\xi^2}{tb^2} - i\delta\theta \frac{\xi}{\sqrt{t}b})&=\frac{1}{2}(\theta \sigma^{2} + \delta^2\theta^2)\frac{\xi^2}{tb^2} - i\delta\theta \frac{\xi}{\sqrt{t}b} + o\left(\frac{1}{t\sqrt{t}}\right)\\
\lim_{t \to +\infty} o\left(\frac{1}{t\sqrt{t}}\right)&=0   
\end{align*}
The characteristic function, $\phi^{T}(\xi,t)$, developed previously  becomes:
\begin{align*}
\phi^{T}(\xi,t)= e^{i t\alpha \theta \delta\frac{\xi}{b_{t}}}{\left(1+\frac{1}{2}\theta \sigma^{2}\frac{\xi^2}{tb^2} - i\delta\theta \frac{\xi}{\sqrt{t}b}\right)^{-t\alpha}}&=e^{i t\alpha \theta \delta\frac{\xi}{b_{t}}}e^{-t\alpha\ln(1+\frac{1}{2}\theta \sigma^{2}\frac{\xi^2}{tb^2} - i\delta\theta \frac{\xi}{\sqrt{t}b})}\\
&=e^{-\frac{1}{2}\alpha(\theta \sigma^{2} + \delta^2\theta^2)\frac{\xi^2}{b^2} + o\left(\frac{1}{\sqrt{t}}\right)}\\
&=e^{-\frac{1}{2}\xi^2 + o\left(\frac{1}{\sqrt{t}}\right)}
   \end{align*}
We have 
 \begin{align}
\lim_{t \to +\infty}\phi^{T}(\xi,t)=\lim_{t \to +\infty}E\{e^{i \frac{{Y_{t}}-a_{t}}{b_{t}}\xi}\}=e^{-\frac{1}{2}\xi^2} \label{eq:l555a}
\end{align}
\noindent
By applying the limit in (\ref{eq:l555a}),  we produce the cumulant-generating function \cite{kendall1946advanced} of the Normal distribution. We  have the following convergence in distribution
\begin{align*}
\frac{{Y_{t}}-a_{t}}{b_{t}}\overset{\text{d}}{\sim}N(0,1) \quad &\text{as}\quad t \to +\infty
 \end{align*}
\end{proof} 
\noindent
As shown in (\ref {eq:l202}), the dynamic of the probability density $f(y, t)$ is carried by two parameters: $t\mu$ and $t\alpha$. $f(y, t)$ can be compared to the histogram of the daily SPY ETF return data, as shown in Fig \ref{fig221}. Fig \ref{fig222} shows the shape of the probability densities (\ref{eq:l201})  adjusted at different timeframes:  Quarterly ($\tau=0.25$), Semi-Annual ($\tau=0.5$), Third-Quarterly ($\tau=0.75$), and Annual ($\tau=1$).
\begin{figure}[ht]
 \vspace{-0.2cm}
  \centering
  \begin{subfigure}[b]{0.45\linewidth}
    \includegraphics[width=\linewidth]{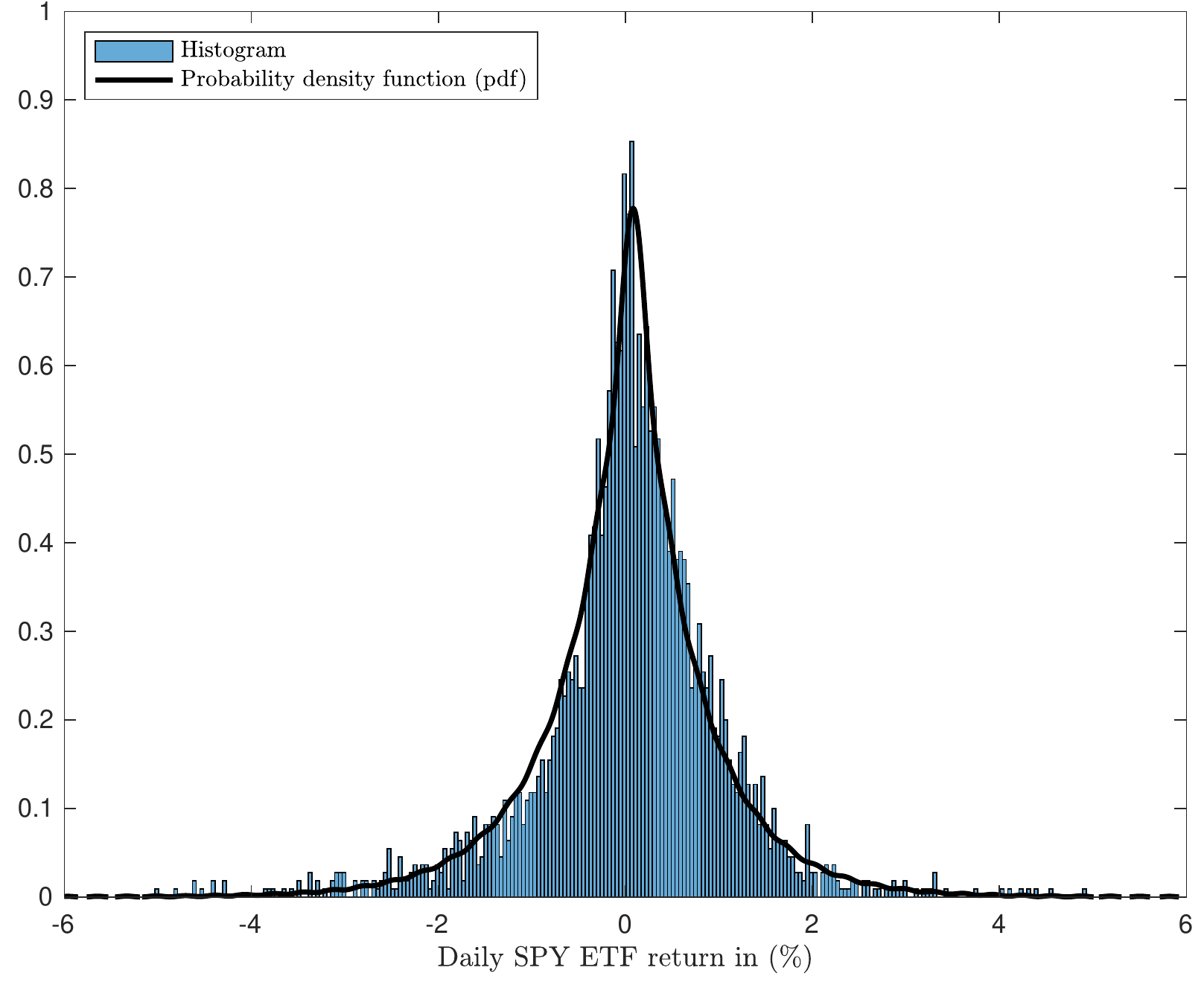}
\vspace{-0.5cm}
     \caption{Daily SPY ETF Return in (\%)}
         \label{fig221}
  \end{subfigure}
  \begin{subfigure}[b]{0.45\linewidth}
    \includegraphics[width=\linewidth]{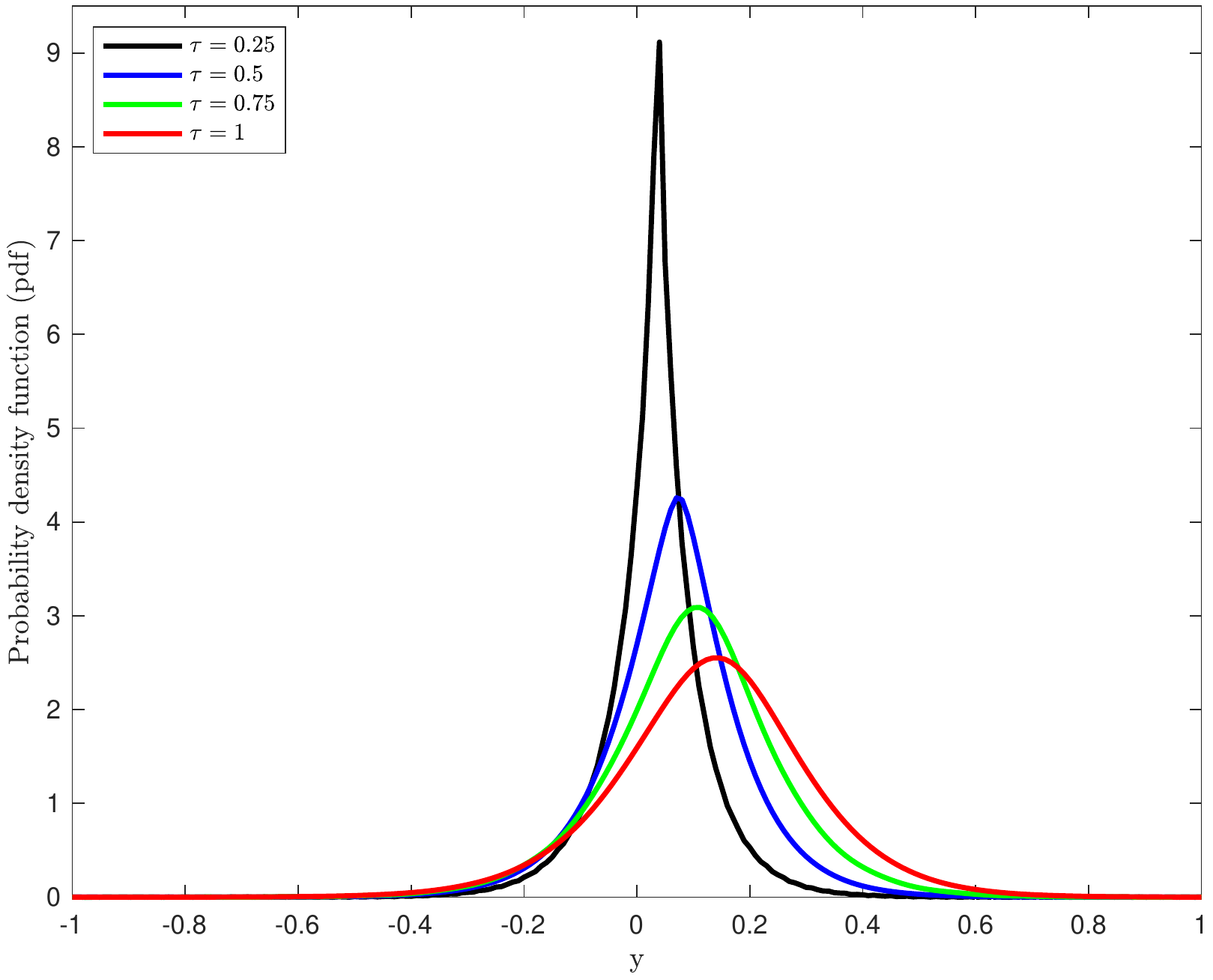}
\vspace{-0.5cm}
     \caption{$f(y,\tau)$: $\tau$ in years}
         \label{fig222}
          \end{subfigure} 
\vspace{-0.5cm}
  \caption{$f(y, t)$: $\hat{\mu}=0.0848$, $\hat{\delta}=-0.0577$, $\hat{\sigma}=1.0295$, $\hat{\alpha}=0.8845$, $\hat{\theta}= 0.9378$}
  \label{fig220}
\vspace{-0.5cm}
\end{figure}
\noindent
The discrepancy between the shape of the probability densities (\ref{eq:l201}) can be explained by the Asymptotic distribution in Theorem \ref{lem7}. When the time frame becomes large, The VG probability density generated by the L\'evy Process changes its nature and becomes a Normal distribution Process. Empirically, the convergence is illustrated in Fig \ref{fig222}.
\section{Variance - Gamma Process: Pricing European Options}
\subsection{Variance - Gamma Process:  Risk Neutral Esscher Measure}
\noindent
The method of Esscher transforms introduced by \cite{gerber1993option}  as an efficient technique for pricing derivative securities if a L\'evy process models the logarithms of the underlying asset prices. An Esscher transform of a stock-price process provides an equivalent martingale measure; under such measure, the price of any derivative security is calculated as the expectation of the discounted payoffs. In some cases, the Esscher transform of a distribution \cite{gerber1993option} remains in the family of the original distributions. In particular, Gamma, Exponential, Normal, Inverse Gaussian, Negative Binomial, Geometric, Poisson, and Compound Poisson distributions are examples of such conservative distributions. The existence of the equivalent Esscher transform measure is not always guaranteed, and the issue of the unicity of the equivalent martingale measure remains recurrent when pricing an option with a L\'evy process.\\
From the characteristic function $\phi(\xi)$ in (\ref{eq:l93982}),  We have the Moment generating function of the VG model.
\begin{equation}
\begin{aligned}
M(h,t)&= \phi(-ih)^{t} =\frac{e^{t \mu h}}{\left(1 - \frac{1}{2}\theta \sigma^{2}h^{2} - \delta\theta h\right)^{t \alpha}}={M(h,1)}^{t} \label {eq:l31} \hspace{5mm}  \hbox{with $h_{1} < h < h_{2}$ } \\
M(h,1)&=\frac{e^{\mu h}}{\left(1 - \frac{1}{2}\theta \sigma^{2}h^{2} - \delta\theta h\right)^{\alpha}} \hspace{5mm}  \hbox{with $h_{1} < h < h_{2}$ } \\
 h_{1}&= -\frac{\delta}{\sigma^2} - \sqrt{\frac{\delta^2}{\sigma^4}+\frac{2}{\theta \sigma^2}}  \quad h_{2}=-\frac{\delta}{\sigma^2}+ \sqrt{\frac{\delta^2}{\sigma^4}+\frac{2}{\theta \sigma^2}} 
\end{aligned}
\end{equation}
\noindent
Under the Esscher transform with parameter h, the probability density of $Y={Y_{t}}$  becomes:
\begin{align}
\hat{f}(x,t,h)= \frac{e^{hx}{f(x, t)}}{M(h,t)}   \quad  \quad \hbox{with $h_{1} < h < h_{2}$ } \label {eq:l3021}
\end{align} 
\noindent
The Moment generating function  of the Esscher transform VG model with $h_{1} < h < h_{2}-z$
\begin{equation}
\begin{aligned}
M(z,t,h)= E^{h}\left[e^{zX_{t}}\right]=\int_{0}^{+\infty}e^{zx}\hat{f}(x,t,h)dx&= \int_{0}^{+\infty}\frac{e^{(h+z)x}{f(x, t)}}{M(x,t)} dx  \label {eq:l321}\\
&=\frac{M(h+z,t)}{M(h,t)}  \quad \hbox{with $h_{1}  < h < h_{2} - z$ }\\
&=\left(\frac{M(h+z,1)}{M(h,1)}\right)^{t} =M(z,1,h)^{t} 
\end{aligned}
\end{equation}
with 
\begin{equation}
\begin{aligned}
M(z,1,h)=\frac{M(h+z,1)}{M(h,1)} = e^{\mu z} {\left(M_{**}(z,1,h)\right)}^{\alpha} \\ M_{**}(z,1,h)=\frac{1 - \frac{1}{2}\theta \sigma^{2}h^{2} - \delta\theta h}{1 - \frac{1}{2}\theta \sigma^{2}(h+z)^{2} - \delta\theta (h+z)} \label {eq:l332}
\end{aligned}
\end{equation}

\noindent
$\hat{f}(x,t,h)$ is the modified probability density of $f(x,t)$ defined in (\ref{eq:l201}). The function $exp(x)$ is a strictly increasing function, and the probability measure generated by $\hat{f}(x,t,h)$ is equivalent to the original probability measure generated by $f(x,t)$. Both probability measures have the same null sets \cite{gerber1993option}(sets of probability measure zero).\\
Given the process $\left\{e^{-r\tau}S(\tau)\right\}_{\tau \geq 0}$ with $r$ the constant risk-free rate of interest. We look into the conditions to have   $h=h^*$ such that
\begin{align}
E^{h^*}\left[e^{-r\tau}S(\tau)\right]=S(0) \label {eq:l33}
\end{align}
\noindent
We have $S(\tau)=S(0)e^{Y_{\tau}}$,  with $Y_{\tau}$ is the Variance - Gamma process. The equation (\ref {eq:l33}) becomes
\begin{align}
e^{r\tau}=E^{h^*}\left[e^{Y_{\tau}}\right]=M(1,1, h^*)^{\tau}=\left(\frac{M(h^*+1,1)}{M(h^*,1)}\right)^{\tau}  \quad \hbox{with $h_{1} < h^{*} < h_{2}-1$ } \label {eq:l34}
\end{align}
\noindent
The first condition is that 
\begin{align*}
h_{2}  - h_{1} >1  \quad   \quad   \frac{\delta^2}{\sigma^4}+\frac{2}{\theta \sigma^2} >\frac{1}{4}
\end{align*}
\noindent
The equation (\ref {eq:l34}) is equivalent to (\ref {eq:l35}).
\begin{align}
e^{\frac{r-\mu}{\alpha}}&=M_{**}(1,1,h^*)=\frac{M(h^*+1,1)}{M(h^*,1)}=\frac{1 - \frac{1}{2}\theta \sigma^{2}h^{*2} - \delta\theta h^*}{1 - \frac{1}{2}\theta \sigma^{2}(h^*+1)^{2} - \delta\theta (h^*+1)} \label {eq:l35}
\end{align}
We consider the function  $g(h)$ defined as follows
$$g(h)=\frac{1 - \frac{1}{2}\theta \sigma^{2}h^{2} - \delta\theta h}{1 - \frac{1}{2}\theta \sigma^{2}(h+1)^{2} - \delta\theta (h+1)} $$
\begin{align*}
\frac{dg}{dh}(h)=\frac{\frac{1}{2}\theta^{2} \sigma^{4}h^{2} + \delta\theta^{2}(\frac{1}{2}\sigma^{2}+\delta)h + \delta\theta^{2}(\frac{1}{2}\sigma^{2}+\delta)+\theta \sigma^{2} }{(1 - \frac{1}{2}\theta \sigma^{2}(h+1)^{2} - \delta\theta (h+1))^{2}} 
\end{align*}
we have 
\begin{align}
\frac{dg}{dh}(h)>0 \quad \hbox{$h_{1} < h < h_{2}-1$}\quad  \quad \quad \lim_{h \to h_{1}} g(h)&= 0  \quad  \quad \quad \lim_{h \to {h_{2}-1}^{-}} g(h)=+\infty \label {eq:l362}
\end{align}
(\ref {eq:l362}) shows the existence and the unicity of $h^*$ in $[h_{1}, h_{2}-1[$ such that 
\begin{align*}
e^{\frac{r-\mu}{\alpha}}=g(h^*)
\end{align*}
\noindent
For VG Model in Table \ref {tab1} \cite{nzokem2021fitting,nzokem_2021b}, the existence and unicity of $h^*$ can be studied empirically as shown in Fig \ref {fig11}. Over the interval $[h_{1}; h_{2}-1]$, g(h) is an increasing function, as shown in Fig \ref{fig11a}. Fig \ref{fig11b} provides the solution $h^*$ of equation (\ref{eq:l35}) for free interest rate less than 10\%. The solution $h^*$ increases with the free interest rate $r$.
\begin{figure}[ht]
    \centering
  \begin{subfigure}[b]{0.4\linewidth}
    \includegraphics[width=\linewidth]{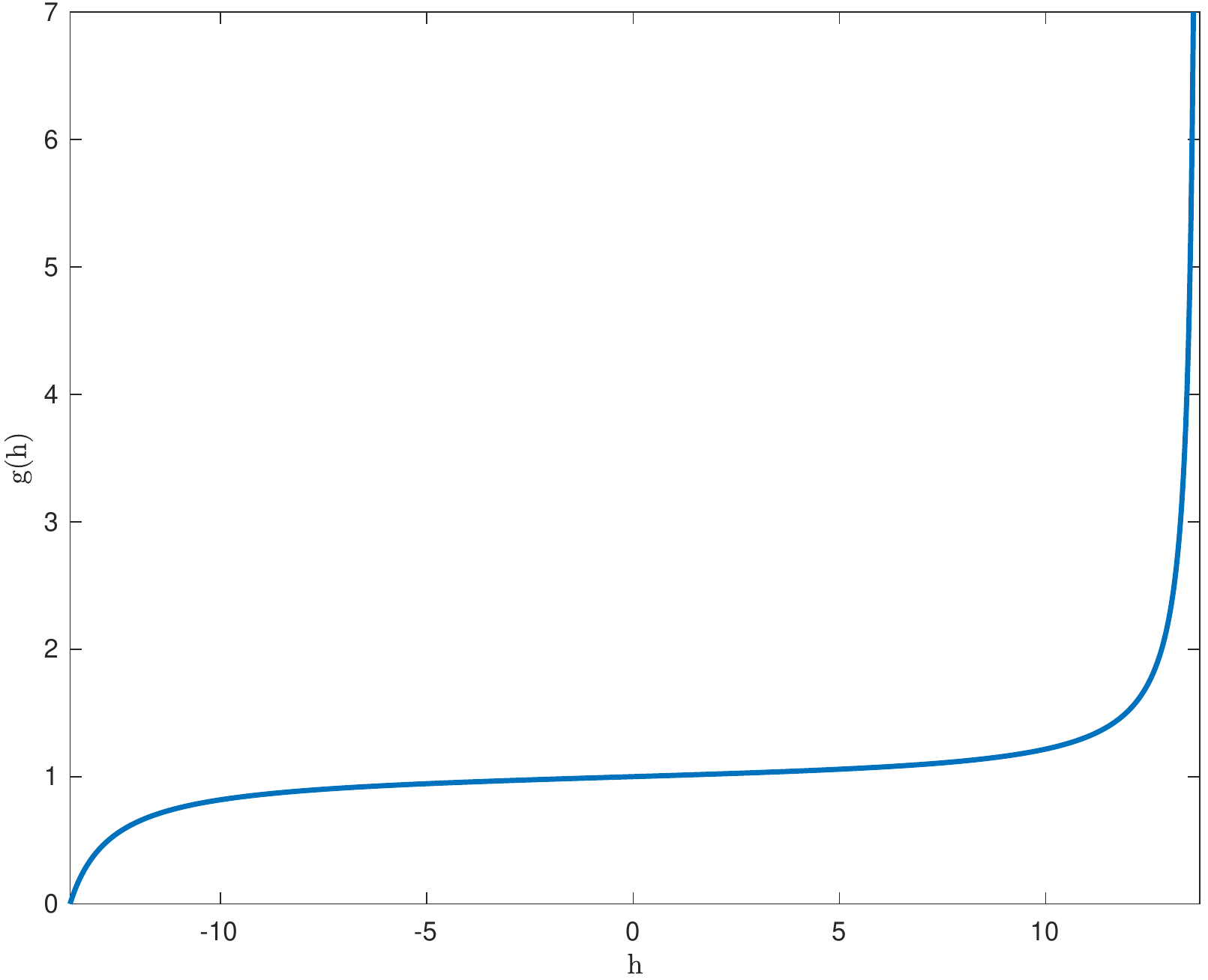}
\vspace{-0.5cm}
     \caption{$g(h)$}
         \label{fig11a}
  \end{subfigure}
  \begin{subfigure}[b]{0.4\linewidth}
    \includegraphics[width=\linewidth]{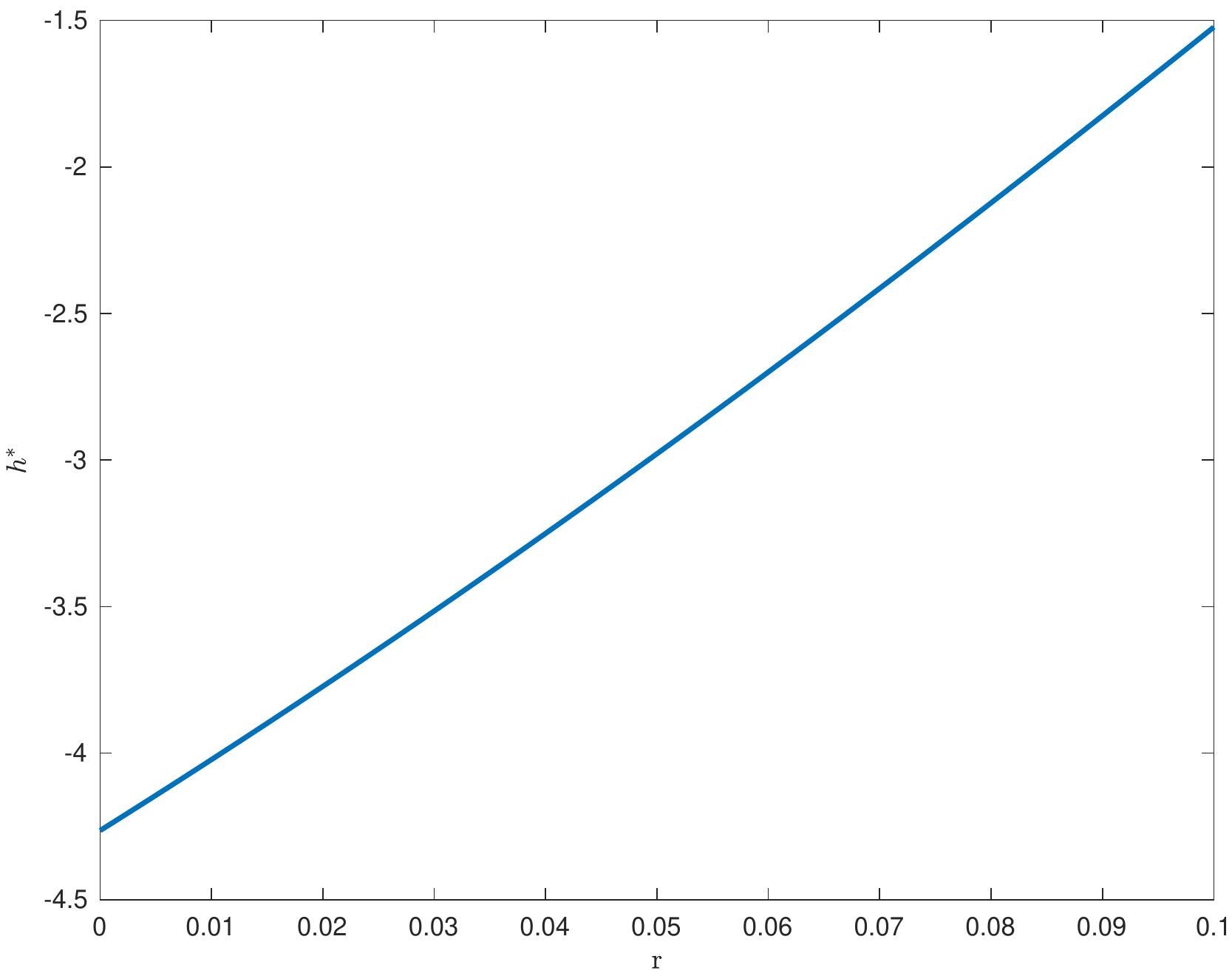}
\vspace{-0.5cm}
     \caption{$e^{\frac{r-\mu}{\alpha}}=g(h)$ }
         \label{fig11b}
          \end{subfigure}
\vspace{-0.5cm}
  \caption{VG Model:  $\hat{\mu}=0.0848$, $\hat{\delta}=-0.0577$, $\hat{\sigma}=1.0295$, $\hat{\alpha}=0.8845$, $\hat{\theta}= 0.9378$, $h_{1}=-13.6511$ and  $h_{2}=14.7399$}
  \label{fig11}
\vspace{-0.5cm}
\end{figure}

\noindent
From the Esscher transform, we have the Equivalent Martingale Measure (EMM) $\mathbf{Q}$, which can be written as the Radon-Nikodym
derivative:
\begin{align}
\frac{d\mathbf{Q}}{d\mathbf{P}}=\frac{e^{h^* Y_{\tau}}}{M(h^*,\tau)}=e^{h^* Y_{\tau} - log(M(h^*,\tau))} \label {eq:l37}
\end{align}
and $E^{\mathbf{Q}}$ for expectation with respect to $\mathbf{Q}$
\begin{align*}
 E^{\mathbf{Q}}\left[e^{-r\tau}S(\tau)\right] &= E^{\mathbf{P}}\left[e^{-r\tau}S(\tau)\frac{d\mathbf{Q}}{d\mathbf{P}}\right] = S(0)E^{P}\left[e^{(1+h^*)Y_{\tau} - log(M(h^*,\tau)) - r \tau}\right]\\
  &= S(0)E^{\mathbf{P}}\left[e^{(1+h^*)Y_{\tau}}\right] e^{- log(M(h^*,\tau)) - r \tau}\\
   &= S(0)\frac{M(1+h^*,\tau)}{M(h^*,\tau)}e^{- r \tau} \quad \quad (\ref {eq:l34}) \quad e^{r\tau}=\left(\frac{M(h^*+1,\tau )}{M(h^*,\tau)}\right) \\
    &= S(0)
   \end{align*}
we have 
\begin{align}
 E^{\mathbf{Q}}\left[e^{-r\tau}S(\tau)\right] = S(0) \label {eq:l377}
    \end{align}
\smallskip
\begin{theorem}\label{lem71} (Variance - Gamma Esscher transform distribution) \ \\
\noindent
Esscher transform of  Variance - Gamma process $Y=\{Y_{t}\}_{t \geq0} $  with parameter $( t \mu, \delta, \sigma, t \alpha, \theta)$ is also a  Variance - Gamma process  with parameter $( t \mu, \tilde{\delta}, \sigma, t \alpha, \tilde{\theta})$ 
 \begin{align}
\tilde{\delta}=\delta + h\sigma^{2}\quad \quad
\tilde{\theta}=\frac{\theta}{1 - \frac{1}{2}\theta \sigma^{2}{h}^{2} - \delta\theta {h}} \label {eq:l377a}
\end{align}
\end{theorem}

\begin{proof}[Proof:] \ \\
From (\ref{eq:l332}), we have
\noindent
\begin{align}
M_{**}(z,1,h) =\frac{1 - \frac{1}{2}\theta \sigma^{2}h^{2} - \delta\theta h}{1 - \frac{1}{2}\theta \sigma^{2}(h+z)^{2} - \delta\theta (h+z)} \label {eq:l377b}\end{align}
\noindent
 We can divide the denominator by the numerator of the function $M_{**}(z,1,{h})$ in (\ref{eq:l377b}) and rearrange the resulting expression.
\begin{align}
M_{**}(z,1,h) = \frac{1}{1 - \frac{1}{2}\tilde{\theta} \sigma^{2}z^{2} - \tilde{\delta}\tilde{\theta} z} \quad \tilde{\theta} = \frac{\theta}{1 - \frac{1}{2}\theta \sigma^{2}h^{2} - \delta\theta h} \quad  \tilde{\delta}=\delta + h\sigma^{2} \label {eq:l3321}
\end{align}
$M(z,1, h)$ in (\ref{eq:l332}) becomes
\begin{align}
M(z,1,{h})=  \frac{e^{\mu z}}{\left(1 - \frac{1}{2}\tilde{\theta} \sigma^{2}z^{2} - \tilde{\delta}\tilde{\theta} z\right)^{\alpha}} \label {eq:l3311}
\end{align}
\noindent
Using the Esscher transform  method, the moment generating function for Variance - Gamma process $Y=\{Y_{t}\}_{t \geq0} $ becomes:
\begin{align}
M(z,t,{h})= E^{{h}}\left[e^{zY_{t}}\right]=M(z,1,{h})^{t} = \frac{e^{ t \mu z}}{\left(1 - \frac{1}{2}\tilde{\theta} \sigma^{2}z^{2} - \tilde{\delta}\tilde{\theta} z\right)^{t \alpha}}\quad \hbox{with $\tilde{h}_{1} < z < \tilde{h}_{2}$ } \label {eq:l3312}
\end{align}
We have a new Variance - Gamma process  with parameter $( t \mu, \tilde{\delta}, \sigma, t \alpha, \tilde{\theta})$ 
  \end{proof} 
\noindent
The Esscher transform method preserves the structure of the five-parameter VG process; it introduces an addition symmetric parameter ($h\sigma^{2}$) and inflates the Gamma scale parameter by $\frac{1}{1 - \frac{1}{2}\theta \sigma^{2}h^{2} - \delta\theta h}$ factor.\\
\subsection{Variance - Gamma Model: Extended Black-Scholes Formula}
\begin{corollary}\label{lem8} (Extended Black-Scholes) \ \\
 Let $r$ a continuously compounded risk-free rate of interest; $Y=\{Y_{t}\}_{t \geq0} $, a VG Process with parameter $(\mu t, \delta, \sigma, \alpha t, \theta)$; and $(S(0)e^{X_{T}} - K)^{+}$ the terminal payoff for a contingent claim with the expiry date $T$.\\
Then at time $t <T$, the arbitrage price of the European call option with the strike price $K$ can be written as follows.
 \begin{align}
F^{GV}_{call}(S_{t},t)&= S(t)\left[ 1-\hat{F}(log(\frac{K}{S(t)}),\tau, h^{*}+1) \right] - Ke^{-r\tau}\left[ 1-\hat{F}(log(\frac{K}{S(t)}),\tau, h^{*})\right] \label {eq:l5}\\
&\hat{F}(log(\frac{K}{S(t)}),\tau, h^{*})=\int_{-\infty }^{log(\frac{K}{S(t)})}\hat{f}(\xi, \tau, h^{*})d\xi \quad \hbox{$\hat{f}(\xi, \tau, h^{*})$ in (\ref{eq:l3021})}
\end{align}
where $\tau=T-t$, and $\hat{F}(k,\tau, h^{*})$ and $\hat{F}(k,\tau, h^{*}+1)$ are the cumulative distribution of VG Model with parameter $(\tau \mu, \tilde{\delta}, \sigma, \tau \alpha, \tilde{\theta})$  and parameter $(\tau \mu, \tilde{\delta} + \sigma^{2}, \sigma, \tau \alpha, \tilde{\theta}e^{\frac{r-\mu}{\alpha}})$ respectively 
\end{corollary}

\begin{proof}[Proof:] 
 \begin{align*}
f(Y_{T},K)= (S(0)e^{Y_{T}} - K)^{+}&= S(t)(e^{Y_{\tau}} - k)^{+} \quad  \quad  \hbox{$Y_{T}=Y_{\tau} + Y_{t}$}\\
&=S(t)g(Y_{\tau})  \quad  \quad  \hbox{$S(t)=S(0)e^{Y_{t}}$ and $k=\frac{K}{S(t)}$}
 \end{align*}
\noindent
Under the Equivalent Martingale Measure (EMM), $\hat{f}(\xi, \tau, h^{*})$ is the probability density of VG model with parameter $( \tau \mu, \tilde{\delta}, \sigma, \tau \alpha, \tilde{\theta})$. We note $k=\frac{K}{S(t) }$
 \begin{align*}
S(t)e^{-r\tau} \int_{log(k)}^{+\infty} e^{\xi} \hat{f}(\xi, \tau, h^{*}) d\xi &= S(t)e^{-r\tau} \int_{log(k)}^{+\infty} e^{\xi} \frac{e^{h^{*}\xi}{f(\xi, \tau)}}{M(h,t)} d\xi \quad \hbox{$\hat{f}(\xi, \tau, h^{*})$ in (\ref{eq:l3021})  }\\
&= S(t)e^{-r\tau} \int_{log(k)}^{+\infty}  \frac{e^{(1+h^{*})\xi}{f(\xi, \tau)}}{M(h^*,t)} d\xi  \quad \hbox{$e^{r\tau}=\frac{M(h^*+1,\tau)}{M(h^*,\tau)}$ in (\ref{eq:l34})}\\
&= S(t) \int_{log(k)}^{+\infty}  \frac{e^{(1+h^{*})\xi}{f(\xi, \tau)}}{M(1+h^*,t)} d\xi =S(t) \int_{log(k)}^{+\infty} \hat{f}(\xi, \tau, h^{*}+1) d\xi
  \end{align*}
We can now show the relation in (\ref{eq:l5})  
 \begin{align*}
F^{GV}_{call}(S_{t},t)&=S(t)e^{-r\tau}E^{Q}\left[g(X_{\tau}) \right]=S(t)e^{-r\tau}\int_{-\infty}^{+\infty}\hat{f}(\xi, \tau, h^{*})(e^{\xi} - k)^{+}dy\\
&=S(t)e^{-r\tau} \int_{log(k)}^{+\infty} e^{\xi} \hat{f}(\xi, \tau, h^{*}) d\xi - Ke^{-r\tau}\int_{log(k)}^{+\infty}\hat{f}(\xi, \tau, h^{*})d\xi \\
&=S(t) \int_{log(k)}^{+\infty} \hat{f}(\xi, \tau, h^{*}+1) d\xi - Ke^{-r\tau}\int_{log(k)}^{+\infty}\hat{f}(\xi, \tau, h^{*})d\xi  \\
&=S(t)\left[ 1-\hat{F}(log(k),\tau, h^{*}+1) \right] - Ke^{-r\tau}\left[ 1-\hat{F}(log(k),\tau, h^{*})\right]   \end{align*}
with 
 \begin{align*}
\hat{F}(log(k),\tau, h^{*})=\int_{-\infty}^{log(k)}\hat{f}(\xi, \tau, h^{*})d\xi \quad \quad \hat{F}(log(k),\tau, h^{*} +1)=\int_{-\infty}^{log(k)}\hat{f}(\xi, \tau, h^{*}+1)d\xi  \end{align*} \end{proof} 
\noindent
From Theorem \ref{lem6} and Theorem \ref{lem71}, $\hat{f}(\xi, \tau, h^{*})$ is the probability density of the VG model with parameter $( \tau\mu, \tilde{\delta}, \sigma, \tau \alpha, \tilde{\theta})$.\\
 \begin{align}
\hat{f}(\xi, \tau, h^{*})&=\frac {1} {\sigma\Gamma(\tau\alpha) \tilde{\theta}^{\tau\alpha}}\int_{0}^{+\infty} \frac{1}{\sqrt{2\pi v}}e^{-\frac{(y-\mu-\tilde{\delta} v)^2}{2v\sigma^2}}v^{\tau\alpha -1}e^{-\frac{v}{\tilde{\theta}}} \,dv    \quad \quad \hbox{($\tilde{\delta}$, $\tilde{\theta}$)  in ($\ref{eq:l377a}$) } \label {eq:l554}
\end{align}
\noindent
Following the same methodology, $\hat{f}(\xi, \tau, h^{*} + 1)$ is the probability density of the VG model with parameter $( \tau\mu, \tilde{\delta}', \sigma, \tau \alpha, \tilde{\theta}')$. we have  
\begin{align}
\tilde{\delta}'&=\tilde{\delta} + \sigma^{2}  \quad \quad \tilde{\theta}'=\tilde{\theta}e^{\frac{r-\mu}{\alpha}}\label {eq:l556}\end{align}
In fact, as in (\ref{eq:l554}), we have :
\begin{align*}
\tilde{\delta}'&=\delta + (h^* +1)\sigma^{2} =\tilde{\delta} + \sigma^{2}  
\end{align*}
And  
\begin{align*}
\tilde{\theta}'&=\frac{\theta}{1 - \frac{1}{2}\theta \sigma^{2}{(h^{*}+1)}^{2} - \delta\theta {(h^{*}+1)}}= \frac{\theta}{1 - \frac{1}{2}\theta \sigma^{2}{h^{*}}^{2} - \delta\theta {h^{*}}}e^{\frac{r-\mu}{\alpha}}=\tilde{\theta}e^{\frac{r-\mu}{\alpha}} 
\end{align*}
We have the expression of  the probability density, 
\begin{align}
\hat{f}(\xi, \tau, h^{*}+1)&=\frac {1} {\sigma\Gamma(\tau\alpha) \tilde{\theta}'^{\tau\alpha}}\int_{0}^{+\infty} \frac{1}{\sqrt{2\pi v}}e^{-\frac{(y-\mu-\tilde{\delta}' v)^2}{2v\sigma^2}}v^{\tau\alpha -1}e^{-\frac{v}{\tilde{\theta}'}} \,dv  \label {eq:l556}
\end{align}

\smallskip

\subsubsection{Equivalent Martingale Measure (EMM) Computation}

\noindent
Under the Equivalent Martingale Measure (EMM), $\hat{f}(\xi, \tau, h^{*})$ is the probability density of VG model with parameter $( \tau \mu, \tilde{\delta}, \sigma, \tau \alpha, \tilde{\theta})$. The Fourier transform is:
  \begin{align*}
\scrF[\hat{f}] (y,\tau,h^{*})&=E\left[e^{-i y X_{\tau}}\right]=\left(\frac{e^{ -i\mu y}}{\left(1 + \frac{1}{2}\tilde{\theta} \sigma^{2}y^{2} + i\tilde{\delta}\tilde{\theta} y\right)^{\alpha}}\right)^{\tau}  \quad  \hbox{see (\ref{eq:l3312})} \\
& =\phi(y)^{\tau}=e^{\tau log(\phi(y))}=e^{-\tau \varphi(-y)}
\end{align*}
 \begin{align}
\scrF[\hat{f}] (y,\tau,h^{*})=e^{-\tau \varphi(-y)} \quad  \quad 
\varphi(y)=-i\mu y + \alpha log(1 + \frac{1}{2}\tilde{\theta} \sigma^{2}y^{2} -i\tilde{\delta}\tilde{\theta} y) \label{eq:l04222}
\end{align}
\noindent
$\phi(y)$ and $\varphi(y)$ are, respectively, the Fourier Transform of the probability density and the characteristic exponent of the VG model with parameter $( \mu, \tilde{\delta}, \sigma,  \alpha, \tilde{\theta})$\\

\noindent
$\hat{f}(\xi, \tau, h^{*})$ can be written as the inverse Fourier Transform from (\ref{eq:l04222})
 \begin{align*}
\hat{f}(\xi, \tau, h^{*})= \frac{1}{2\pi}\int_{-\infty}^{+\infty}e^{i \xi z}\scrF[\hat{f}](z, \tau, h^{*}) dz &=\frac{1}{2\pi}\int_{-\infty}^{+\infty}e^{i \xi z - \tau \varphi(-z)}dz\\
&=\frac{1}{2\pi}\int_{-\infty}^{+\infty}e^{-i \xi z - \tau \varphi(z)}dz  
\end{align*}
It was shown in \cite{nzokem2021fitting} that we can have 
  \begin{align}
\scrF[\hat{F}](\xi,\tau, h^{*})= \frac{\scrF[\hat{f}](\xi,\tau, h^{*})}{i\xi}+\pi\scrF[\hat{f}](0)\delta(\xi)\label{eq:l43}
\end{align}
\noindent
Based on (\ref{eq:l43}), we deduce
 \begin{align*}
\hat{F}(\xi, \tau, h^{*})= \frac{1}{2\pi}\int_{-\infty}^{+\infty}e^{i \xi z}\scrF[\hat{F}](z,\tau, h^{*}) dz=\frac{-1}{2\pi}\int_{-\infty}^{+\infty}\frac{e^{-i \xi z - \tau \varphi(z)}}{iz}dz +\frac{1}{2}   
\end{align*}
We have the probability density and cumulative functions
 \begin{align}
 \hat{f}(\xi, \tau, h^{*})=\frac{1}{2\pi}\int_{-\infty}^{+\infty}e^{-i \xi z - \tau \varphi(z)}dz \quad  \quad \hat{F}(\xi, \tau, h^{*})=\frac{-1}{2\pi}\int_{-\infty}^{+\infty}\frac{e^{-i \xi z - \tau \varphi(z)}}{iz}dz +\frac{1}{2} \label {eq:l45} 
\end{align}
\noindent
For VG Model in Table \ref{tab1} \cite{nzokem2021fitting,nzokem_2021b}: : $\hat{\mu}=0.0848$, $\hat{\delta}=-0.0577$, $\hat{\sigma}=1.0295$, $\hat{\alpha}=0.8845$, $\hat{\theta}= 0.9378$; and a risk-free rate of interest  $6\%$, we have a Esscher transform parameter $h^{*}=-2.6997$ in (\ref{eq:l34}). $\hat{f}(\xi,\tau, h^{*})$ and $\hat{f}(\xi, \tau, h^{*} + 1)$ were computed by the Fractional Fast Fourier (FRFT) \cite{nzokem2021fitting} as shown in Fig \ref{fig72} and Fig \ref{fig82}.\\
\noindent
 Using the numerical integration technique, $\hat{f}(\xi,\tau, h^{*})$ and $\hat{f}(\xi, \tau, h^{*} + 1)$ were computed by implementing the following  12-point rule Composite Newton-Cotes Quadrature Formulas \cite{Nzokem_2021, aubain2020}.
 \begin{equation}
  \begin{aligned}
\hat{f}(\xi, \tau, h)& \approx \frac{b}{n}\sum_{p=0}^{\frac{n}{Q} - 1} \sum_{j=0}^{Q} W_{j} g(x_{Qp+j}, \tau, h)\label {eq:l4644} \\
g(x, \tau, h)&=\frac{1}{\sigma \tilde{\theta}^{\tau\alpha}\Gamma(\tau\alpha)\sqrt{2\pi}}e^{-\frac{(x-\tau \mu-\tilde{\delta} v)^2}{2v\sigma^2}}v^{\tau\alpha -\frac{1}{2}}e^{-\frac{v}{\tilde{\theta}}}
\end{aligned}
\end{equation}
In order to compute (\ref{eq:l4644}) for $h=h^*$ and $h=h^* +1$, the following parameter values are used: $a=0$,
$b=20$, $Q=12$, $n=5000Q$, $n_{0}=5000$ and the weights $\{W_j\}_{0\leq j \leq Q}$ values come from Table 1 \cite{Nzokem_2021} and table 4.1 \cite{aubain2020}. The results are shown in Fig \ref{fig71} and Fig \ref{fig81}.
\begin{figure}[ht]
\vspace{-0.2cm}
    \centering
\hspace{-1cm}
  \begin{subfigure}[b]{0.32\linewidth}
    \includegraphics[width=\linewidth]{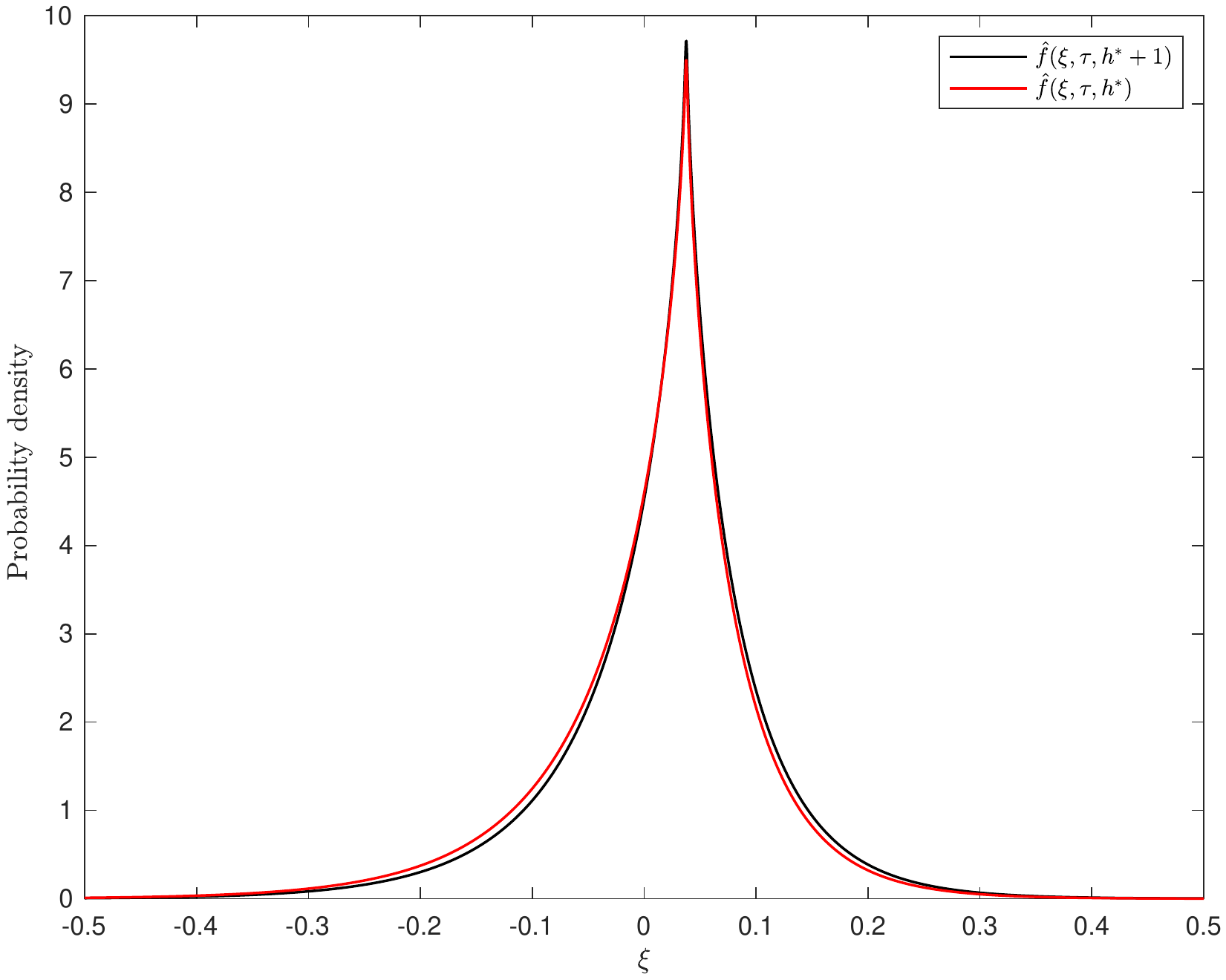}
\vspace{-0.5cm}
     \caption{Newton-Cote Martingale}
         \label{fig71}
  \end{subfigure}
\hspace{-0.2cm}
  \begin{subfigure}[b]{0.32\linewidth}
    \includegraphics[width=\linewidth]{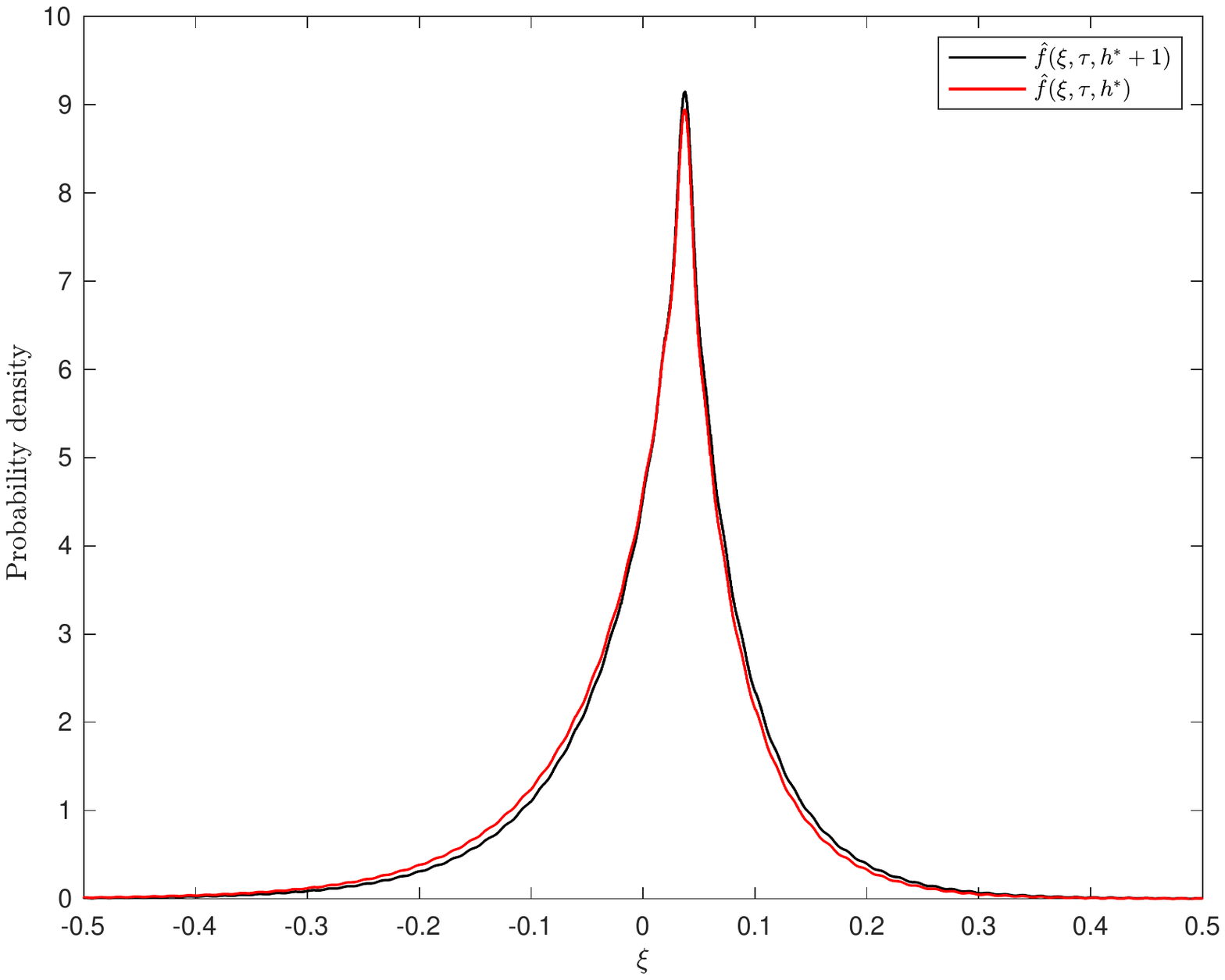}
\vspace{-0.5cm}
     \caption{FRFT Martingale Measure}
         \label{fig72}
          \end{subfigure}
\hspace{-0.2cm}
  \begin{subfigure}[b]{0.33\linewidth}
    \includegraphics[width=\linewidth]{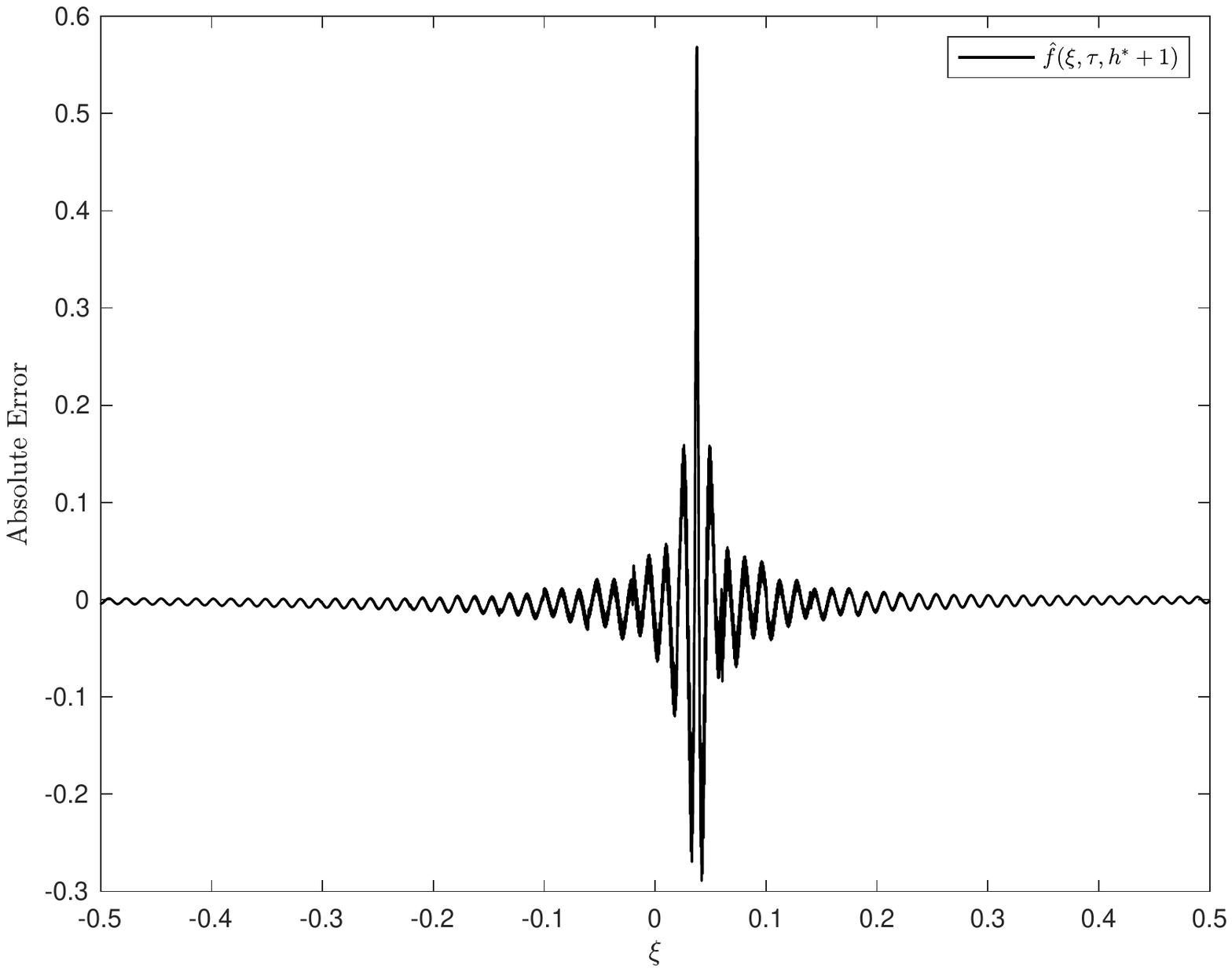}
\vspace{-0.5cm}
     \caption{Error: FRFT versus Newton}
         \label{fig73}
          \end{subfigure}
\vspace{-0.5cm}
  \caption{Estimations: $\hat{f}(\xi,\tau, h^{*})$ versus $\hat{f}(\xi, \tau, h^{*} + 1)$, $\tau=0.25$}
  \label{fig12}
\vspace{-0.3cm}
\end{figure}

\noindent
Both methods produce smooth density functions as shown in Fig \ref{fig12} and Fig \ref{fig13}.  Fig \ref{fig73} and Fig \ref{fig83} provide the estimation error of $\hat{f}(\xi, \tau, h^{*} + 1)$. The Fractional Fast Fourier (FRFT) underestimates the peakedness of the density function when the timeframe is small ($\tau=0.25$) in Fig \ref{fig73}. The estimation error decreases significantly when the timeframe increases. see Fig \ref{fig83} when $\tau=0.5$.
\begin{figure}[ht]
    \centering
\hspace{-1cm}
  \begin{subfigure}[b]{0.32\linewidth}
    \includegraphics[width=\linewidth]{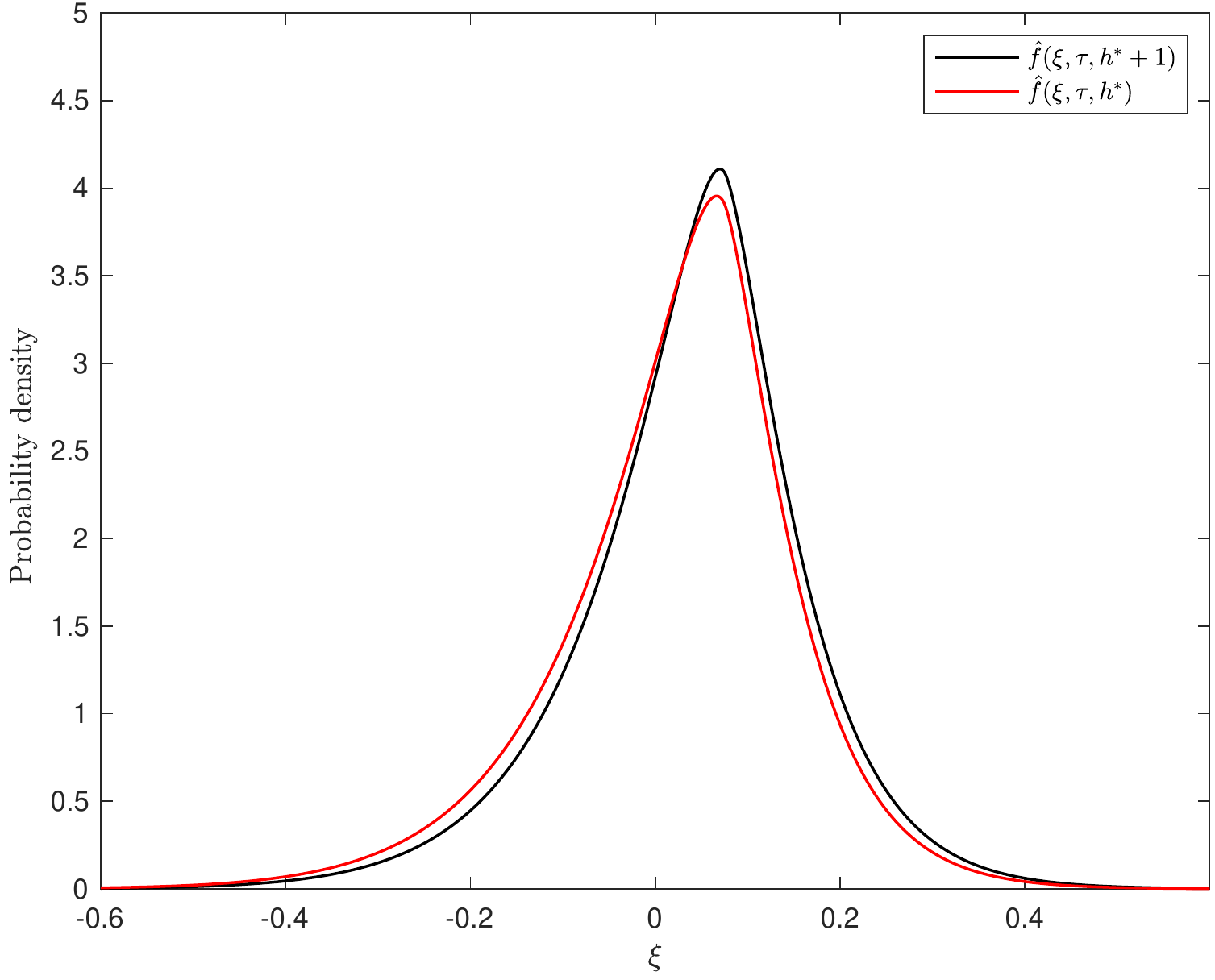}
\vspace{-0.5cm}
     \caption{Newton-Cote Martingale}
         \label{fig81}
  \end{subfigure}
\hspace{-0.2cm}
  \begin{subfigure}[b]{0.32\linewidth}
    \includegraphics[width=\linewidth]{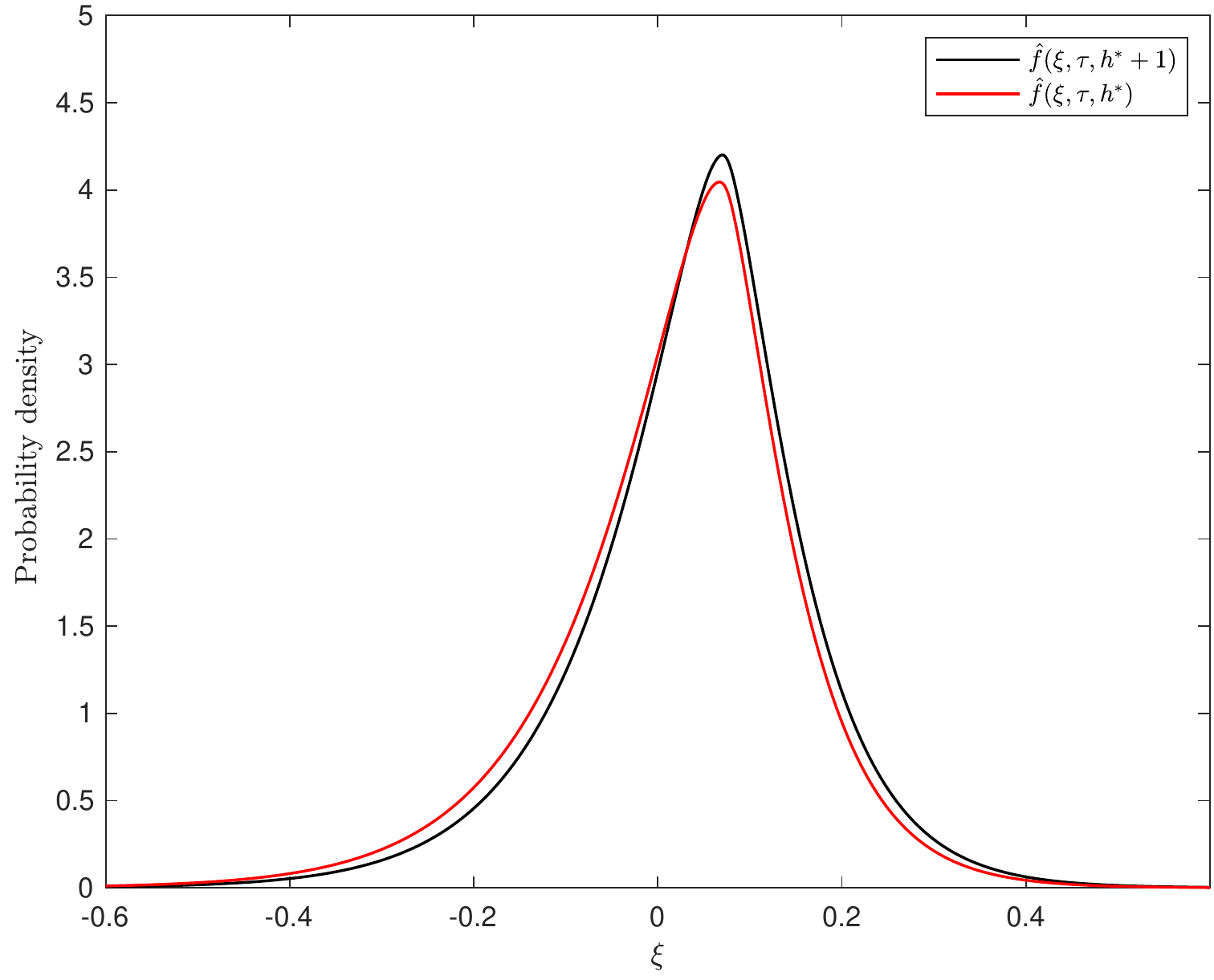}
\vspace{-0.5cm}
     \caption{FRFT Martingale Measure}
         \label{fig82}
          \end{subfigure}
\hspace{-0.2cm}
  \begin{subfigure}[b]{0.33\linewidth}
    \includegraphics[width=\linewidth]{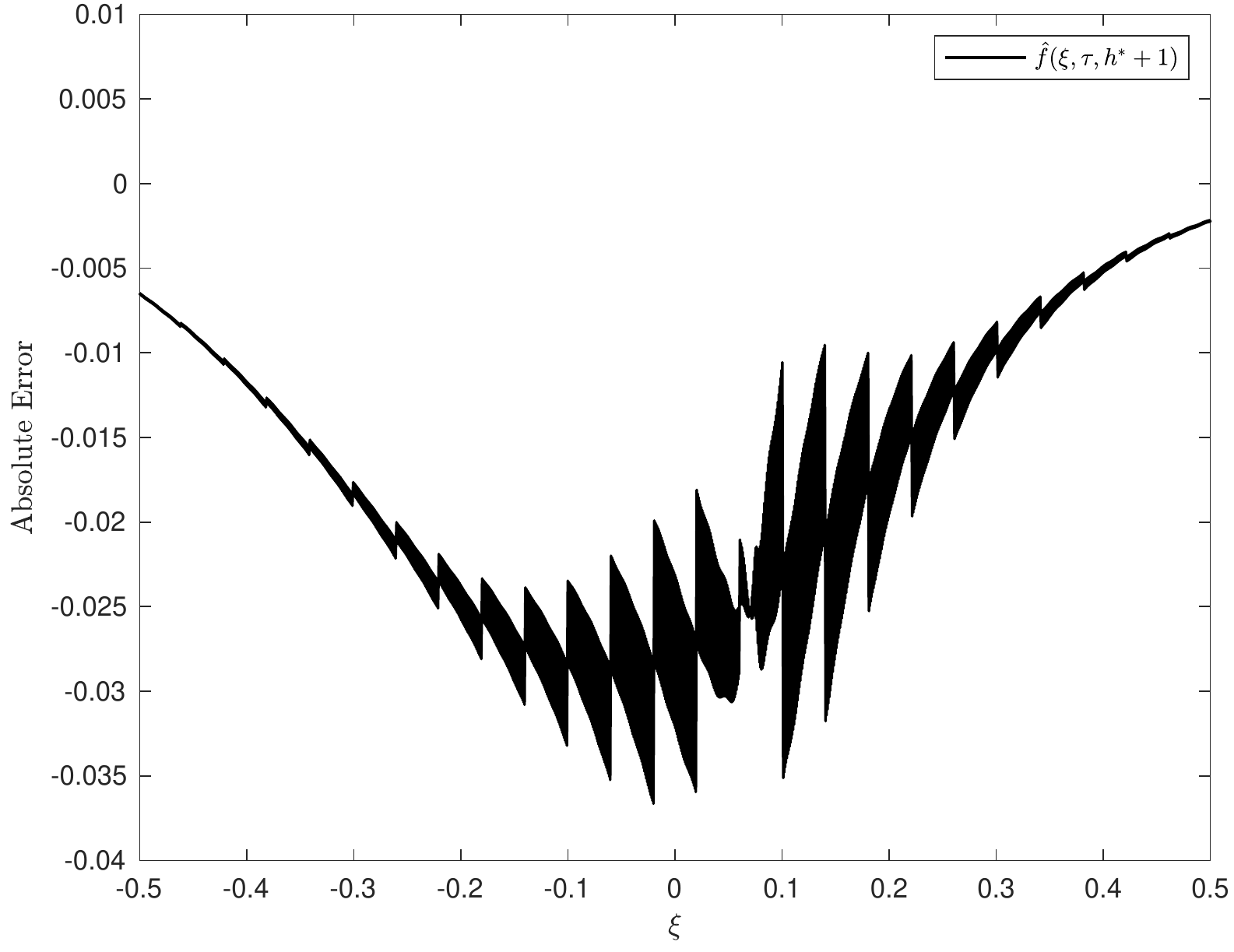}
\vspace{-0.5cm}
     \caption{Error: FRFT versus Newton}
         \label{fig83}
          \end{subfigure}
\vspace{-0.5cm}
  \caption{Estimations: $\hat{f}(\xi,\tau, h^{*})$ versus $\hat{f}(\xi, \tau, h^{*} + 1)$, $\tau=0.5$}
  \label{fig13}
\vspace{-0.3cm}
\end{figure}

\noindent
Both methods will be used to compute the arbitrage price of the European call option.

\subsection{ Variance - Gamma Model: Generalized Black-Scholes Formula}
 \begin{theorem} \label{lem9} \ \\
 \noindent
 Let $r$ a continuously compounded risk-free rate of interest, $Y=\{Y_{t}\}_{t \geq0} $, a VG Process with parameter $(\mu t, \delta, \sigma, \alpha t, \theta)$, and $(S(0)e^{X_{T}} - K)^{+}$, the terminal payoff for a contingent claim with the expiry date $T$. Then at time $t <T$, the arbitrage price of the European call option with the strike price $K$ can be written as follows.
\begin{align}
F^{GV}_{call}(S_{t},t)=\frac{K}{2\pi} \int_{-\infty + i q}^{+\infty + i q}{\frac{e^{\left(i \xi log(\frac{S(t)}{K}) - \tau (r + \varphi(\xi))\right)}}{i \xi(i\xi -1)} d\xi} \label{eq:l8}
\end{align}
where $\varphi(z)$ is the characteristic exponent of VG model with parameter $( \mu, \tilde{\delta}, \sigma,  \alpha, \tilde{\theta})$ in (\ref{eq:l377a}), $\tau=T-t$ and $q < -1$.
\end{theorem}

\begin{proof}[Proof:] 
 \begin{align}
 (S(0)e^{Y_{T}} - K)^{+}&= S(t)(e^{Y_{\tau}} - k)^{+} \quad  \quad  \hbox{$Y_{T}=Y_{\tau} + Y_{t}$ }  \label{eq:l91}\\
&=S(t)g(Y_{\tau},k)  \quad  \quad  \hbox{$S(t)=S(0)e^{Y_{t}}$ and $k=\frac{K}{S(t)}$}\label{eq:l92}  
 \end{align}
  \noindent
$(S(0)e^{Y_{T}} - K)^{+}$ is the payoff  of the call option. The Fourier transform can be written
\begin{align*}
\scrF[g] (y,k)&=\int_{0}^{+\infty}e^{-iyx}g(x,k)dx =\int_{0}^{+\infty}e^{-iyx}{(e^{x} - k)^{+}}dx=\int_{log(k)}^{+\infty}e^{-iyx}{(e^{x} - k)}dx\\
&=\int_{log(k)}^{+\infty}e^{(1-iy)x}dx  - k\int_{log(k)}^{+\infty}e^{-iyx}dx=\frac{1}{1-iy}\left[e^{(1 - iy)x}\right]_{ln(k)}^{+\infty} + \frac{k}{iy}\left[e^{-iyx}\right]_{ln(k)}^{+\infty}\\
&=\frac{ke^{-iy log(k)}}{iy(iy -1)}  \hspace{5mm}  \hbox{for $\Im(y) < -1$} 
\end{align*}
 \noindent
 We have the Fourier  transform of call payoff 
  \begin{align}
\hat{g}(y,k)=\scrF[g] (y,k)=\frac{ke^{-iylog(k)}}{iy(iy -1)}  \hspace{5mm}  \hbox{for $\Im(y) < -1$} 
 \label {eq:l101}
\end{align}
It is shown in (\ref{eq:l45}) and in (\ref{eq:l04222}) that $\hat{f}(\xi, \tau, h^{*})$  and $\varphi(y)$ can be written as follows
 \begin{align}
\hat{f}(\xi, \tau, h^{*})=\frac{1}{2\pi}\int_{-\infty}^{+\infty}e^{-i \xi z - \tau \varphi(z)}dz \quad \quad   
\varphi(y)=-i\mu y + \alpha log(1 + \frac{1}{2}\tilde{\theta} \sigma^{2}y^{2} -i\tilde{\delta}\tilde{\theta} y) \label {eq:102} 
\end{align}
with ($\tilde{\delta}$, $\tilde{\theta}$) defines in (\ref{eq:l377a})\\

\noindent
$F^{GV}_{call}(S_{t},t)$  is the call function under the Equivalent Martingale Measure(EMM), and we can now find a good expression of the function
\begin{align*}
F^{GV}_{call}(S_{t},t)&=e^{-r\tau}E^{h^*}\left[(S(0)e^{X_{T}} - K)^{+} \right]=S(t)e^{-r\tau}E^{h^*}\left[g(X_{\tau}) \right] \quad \quad   \quad  \hbox{recall (\ref{eq:l92})}\\
&=S(t)e^{-r\tau}\int_{-\infty}^{+\infty}\hat{f}(y, \tau, h^{*})g(y,k)dy=\frac{S(t)}{2\pi}\int_{-\infty}^{+\infty}\int_{-\infty}^{+\infty}e^{-i y z - \tau (r + \varphi(z))}g(y,k)dydz  \\ 
&=\frac{S(t)}{2\pi}\int_{-\infty}^{+\infty}e^{-\tau (r + \varphi(z))}\hat{g}(z,k) dz   \quad \quad   \quad  \hbox{recall (\ref{eq:l101})} \\ 
&=\frac{K}{2\pi}\int_{-\infty + iq}^{+\infty + iq}\frac{exp\left[{-iz log(k) - \tau (r + \varphi(z)))}\right]}{iz(iz - 1)}dz  \quad \quad \hbox{$\Im(z)\leq q < -1$ and $k=\frac{K}{S(t)}$} \\ 
&=\frac{K}{2\pi}\int_{-\infty + iq}^{+\infty + iq}\frac{exp\left[{i z log(\frac{S(t)}{K}) - \tau (r + \varphi(z))}\right]}{iz(iz - 1)}dz 
\end{align*}
\noindent
We have the formula (\ref{eq:l8})
  \begin{align}
F^{GV}_{call}(S_{t},t)=\frac{K}{2\pi}\int_{-\infty + iq}^{+\infty + iq}\frac{exp\left[{i z log(\frac{S(t)}{K}) - \tau (r + \varphi(z))}\right]}{iz(iz - 1)}dz \label {eq:l243}\end{align}
  \end{proof} 
 \subsection{European Option Pricing by Fractional Fast Fourier (FRFT)}
 
 \subsubsection{Parameter $q$ Evaluation}
\noindent
Let us considerate the stock or index price $S=S_{0} e^{Y}$ and the strike price $K$; it was shown in (\ref{eq:l101}) that the Fourier transform of the call payoff can be written as follows.
  \begin{align*}
\hat{g}(y,k)=\scrF[g] (y,k)=\frac{ke^{-iylog(k)}}{iy(iy -1)}  \hspace{5mm}  \hbox{for $\Im(y) < -1$} 
\end{align*}

\noindent
We can recover the call payoff from the inverse of Fourier in (\ref{eq:l101})
  \begin{align}
  \check{g}(x,k)=\frac{1}{2\pi}\int_{-\infty + iq}^{+\infty + iq}e^{iyx}\scrF[g] (y,k)dy \hspace{5mm}  \hbox{for $q < -1$} 
 \label {eq:l23}
\end{align}

\noindent
The payoff in (\ref {eq:l23}) depends on the parameter $q$. As shown in Fig \ref {fig711},  for $q=-2$,  the inverse of Fourier in (\ref {eq:l23}) produces poor results; in fact, the curve in red fluctuates around real call payoff $(e^{Y} - k)^{+}$. For $q=-1.002$,  the inverse of Fourier over-estimates the call payoff. 
\begin{figure}[ht]
    \centering
  \begin{subfigure}[b]{0.4\linewidth}
    \includegraphics[width=\linewidth]{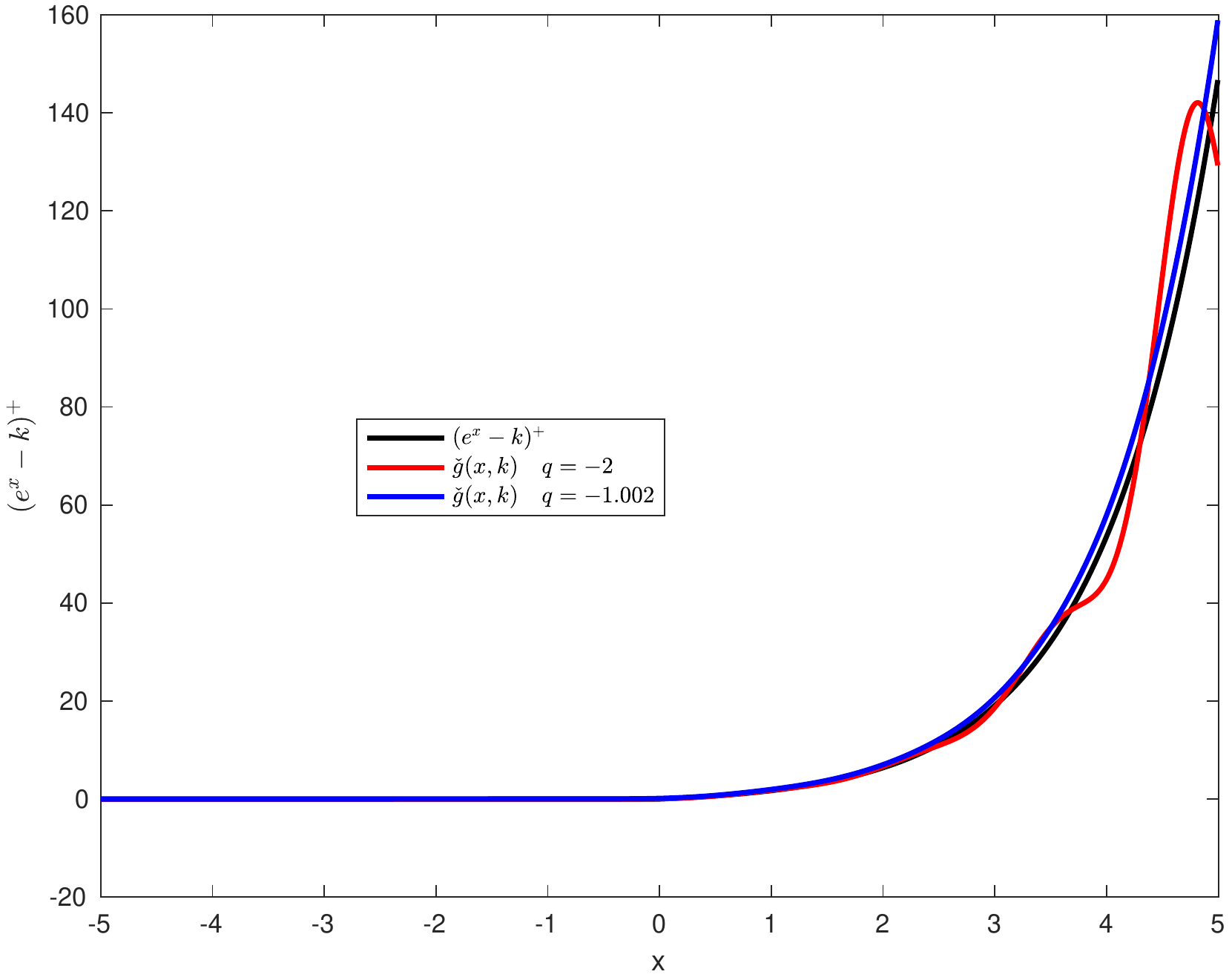}
\vspace{-0.5cm}
   \caption{ $(e^{x} - k)^{+}$ versus $\check{g}(x,k)$ (ATM)}
         \label{fig711}
  \end{subfigure}
\hspace{0.02cm}
  \begin{subfigure}[b]{0.4\linewidth}
    \includegraphics[width=\linewidth]{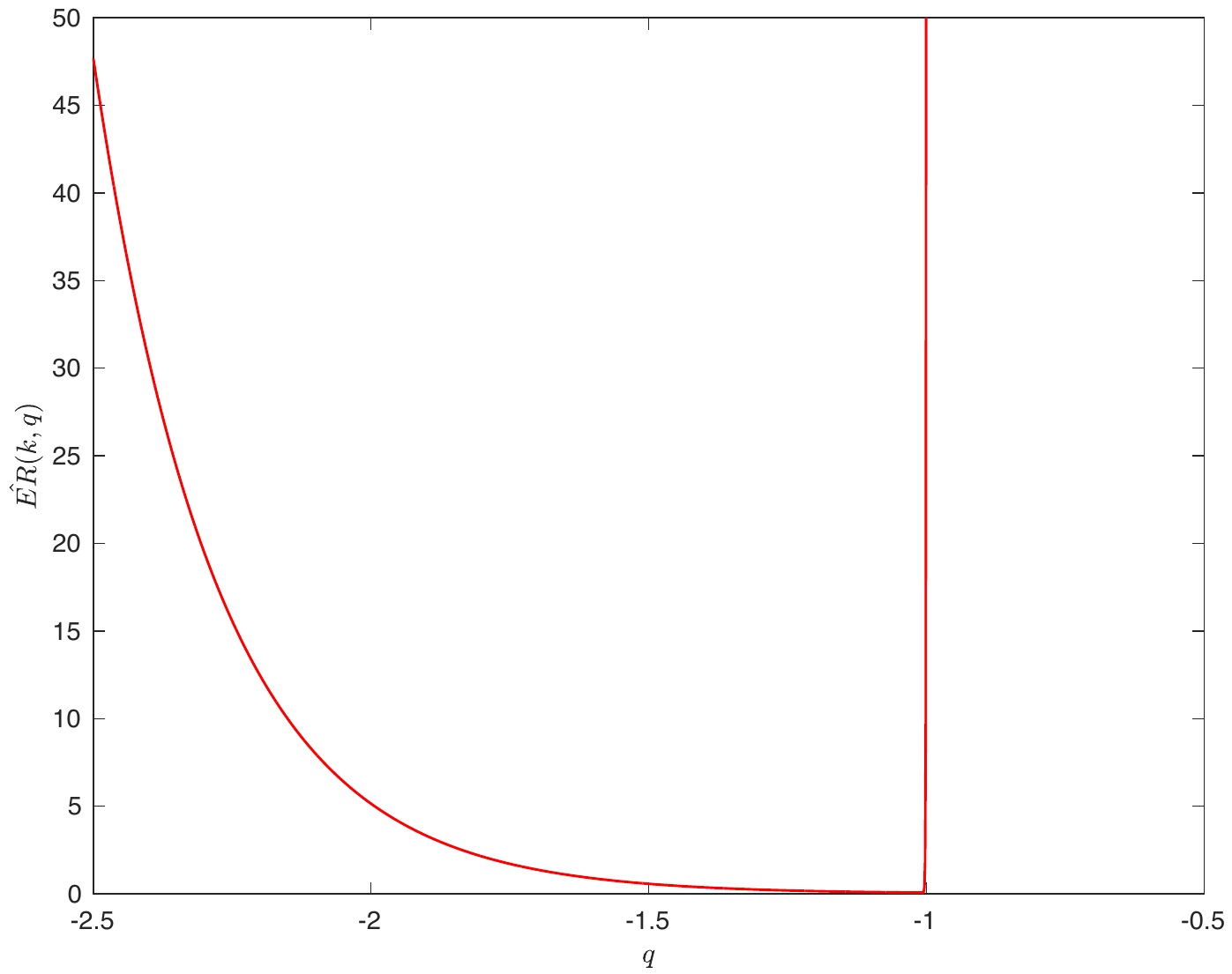}
\vspace{-0.5cm}
     \caption{$ER(k,q)$ and $q$ (ATM)}
         \label{fig712}
          \end{subfigure}
\vspace{-0.5cm}
  \caption{Optimal $q$ parameter}
\end{figure}

\noindent
To find $q$ value with a high level of accuracy, we define the error function ($ER(k,q)$) between the real call payoff and the inverse Fourier payoff, with $k$ (strike price) and the parameter $q$ as inputs.
  \begin{align}
ER(k,q)=\sqrt{ \frac{1}{m}\sum_{j=1}^{m} \left[ ( e^{x_{j}} - k)^{+} -  \check{g}(x_{j},k) \right]^2}  \hspace{5mm}  \hbox{with $-M \leq x_{j}\leq M$} 
 \label {eq:l24}
\end{align}
\noindent
At the money (ATM) option, the strike price $k=1$ and the $ER(k,q)$ can be analyzed as a function of one variable $q$.  Fig \ref {fig712} displays the error function ($ER$) graph as a function of $q$.  $ER$ is a convex function, which decreases and increases over the interval $]-\infty, -1[$. The section method was applied to determine $q=-1.0086$, which minimizes $ER(1,q)$.\\
\noindent
Fig \ref{fig713} displays the $ER(k,q)$ minimum value as a function of the strike price $k$; and the correspondent optimal parameter $q$ as a function of the strike price $k$ (Fig \ref{fig714}). Both graphs display almost a constant function with respect to the strike price.
\begin{figure}[ht]
    \centering
  \begin{subfigure}[b]{0.45\linewidth}
    \includegraphics[width=\linewidth]{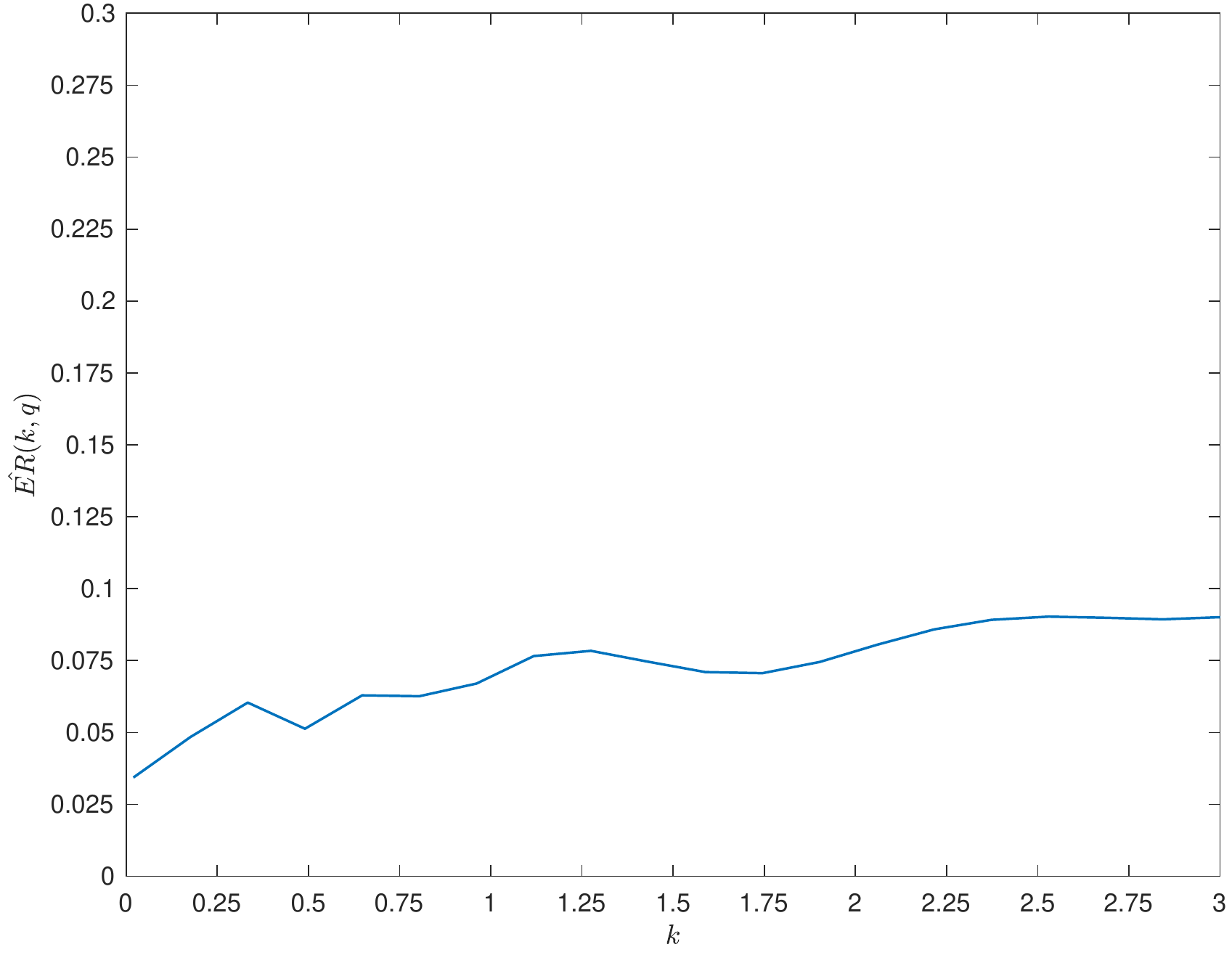}
\vspace{-0.5cm}
     \caption{$ER(k,q)$ minimum value}
         \label{fig713}
          \end{subfigure}
\hspace{0.02cm}
  \begin{subfigure}[b]{0.45\linewidth}
    \includegraphics[width=\linewidth]{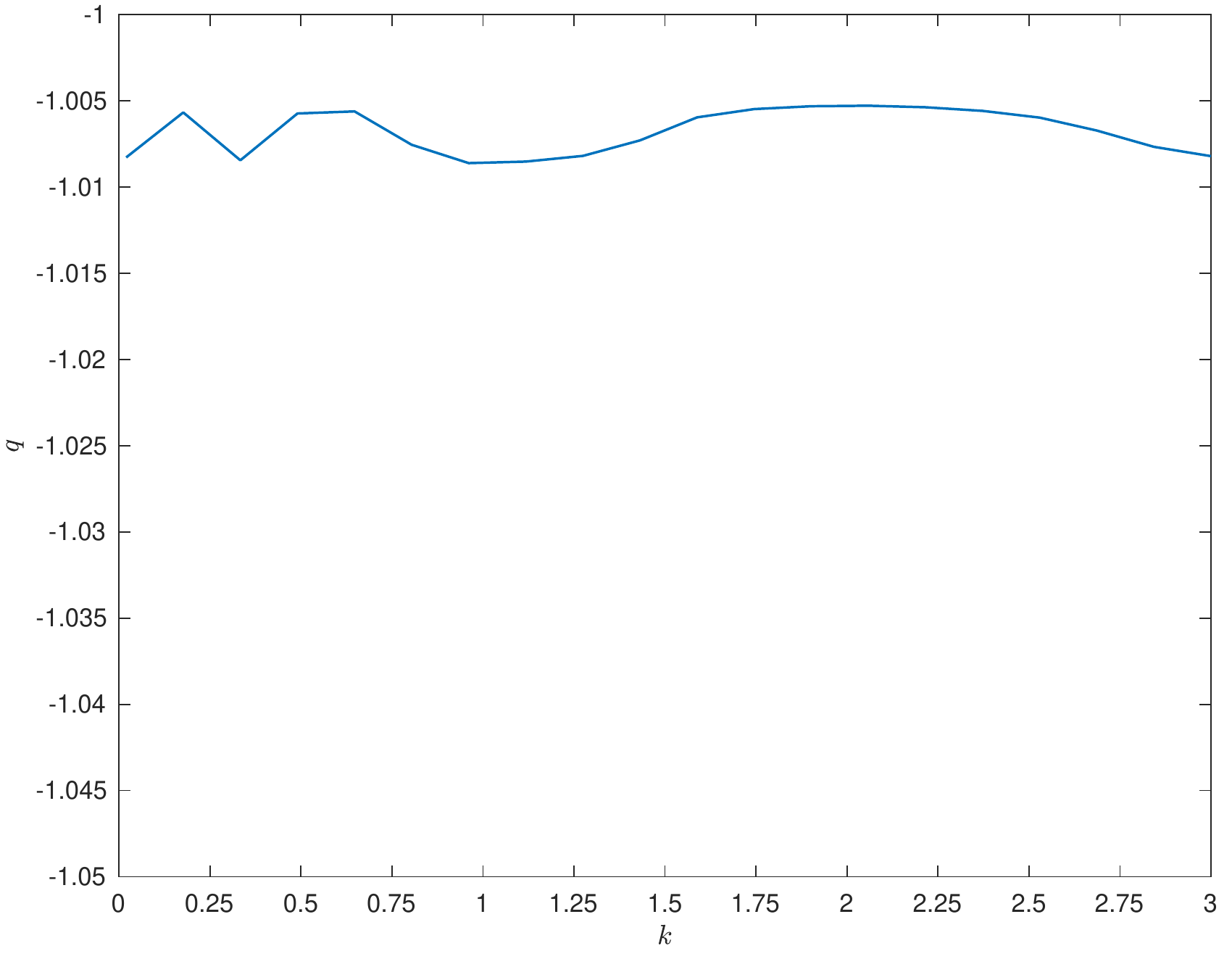}
\vspace{-0.5cm}
     \caption{optimal parameter $q$}
         \label{fig714}
          \end{subfigure}
\vspace{-0.5cm}
  \caption{Optimal $q$ parameter and Error $ER(k,q)$ value}
\vspace{-0.3cm}
\end{figure}

 \subsubsection{Calculating the Fourier Integral by FRFT}
\noindent
The method of Fourier transform \cite{li2020pricing} provides valuable and powerful tools for option pricing under a class of L\'evy processes when the characteristic function is much simpler than their density function. We compute the call option's value on the SPY ETF with the Fractional Fast Fourier (FRFT).  \\
  For $x=log(\frac{S(t)}{K})$, $\frac{F^{GV}_{call}(S(t), t)}{K}$ is the price per one dollar of the strike price. We have 
\begin{align*}
\frac{F^{GV}_{call}(S_{t},\tau)}{K} &= \frac{1}{2\pi}\int_{-\infty + iq}^{+\infty + iq}\frac{exp\left[{i z log(\frac{S(t)}{K}) - \tau (r + \varphi(z))}\right]}{iz(iz - 1)}dz\\
&= \frac{1}{2\pi}\int_{-\infty + iq}^{+\infty + iq}e^{i z x}
\frac{exp\left[{-\tau (r + \varphi(z))}\right]}{iz(iz -1)}dz 
\end{align*}
\noindent
we assume
\begin{align}
f(\xi) = \frac{exp\left[{-\tau (r + \varphi(z))}\right]}{i\xi(i\xi -1)} \quad \quad
F(x)= \frac{1}{2\pi}\int_{-\infty + iq}^{+\infty + iq}e^{i \xi x}
f(\xi)d\xi
\end{align}
\noindent
Following the notation of the Fractional Fast Fourier (FRFT) developed  in \cite{nzokem2021fitting} (Section 2 and Appendix A.1.)
\begin{align*}
F(x_{k})&= \frac{1}{2\pi}\int_{-\infty + iq}^{+\infty + iq}e^{i \xi x_{k}}f(\xi)d\xi =\frac{e^{-q x_{k}}}{2\pi}\int_{-\infty }^{+\infty}e^{i \xi x_{k}}f(\xi+iq)d\xi \\
&\approx \frac{e^{-q x_{k}}}{2\pi}\int_{-\frac{a}{2}}^{\frac{a}{2}} e^{i\xi x_{k}}f(\xi+iq)d\xi \approx \frac{\gamma}{2\pi}e^{-(q +\frac{m}{2}\beta)x_{k}} G_{k}(f(\xi_{j}+iq)e^{\pi i jn\delta},-\delta)
\end{align*}
\begin{align}
F(x_{k})\approx \frac{\gamma}{2\pi}e^{-(q +\frac{m}{2}\beta)x_{k}} G_{k}(f(\xi_{j}+iq)e^{\pi i jn\delta},-\delta)
\end{align}

\subsection{Empirical Analysis}
\noindent
 Based on parameter data from VG Model \cite{nzokem2021fitting,nzokem_2021b}:  $\hat{\mu}=0.0848$, $\hat{\delta}=-0.0577$, $\hat{\sigma}=1.0295$, $\hat{\alpha}=0.8845$, $\hat{\theta}= 0.9378$; $6\%$ risk-free interest rate is added and the Esscher transform parameter ($h^{*}=-2.6997$) was computed. The VG  option pricing will be calculated across maturity and option moneyness using Extended and Generalized Black-Scholes formulas. The closed-form Black-Scholes model \cite{hull2003options} was added to the analysis as a benchmark.
\begin{equation}
\begin{aligned}
F^{BS}_{call}(S_{t},\tau)&=S_{t}N(d_{1}) - K e^{-r\tau}N(d_{2}) \label {eq:l244}\\
d_{1}=\frac{Ln(\frac{S_{t}}{K}) + (r+\frac{1}{2}\sigma^{*2})\tau}{\sigma^{*} \sqrt{\tau}}
 \quad d_{2}&=d_{1} -\sigma^{*}\sqrt{\tau} \quad  N(x)=\frac{1}{\sqrt{2\pi}}\int_{-\infty}^x\exp\left( -\frac{t^2}{2}\ \right)\,dt 
\end{aligned}
\end{equation}
\noindent
The variance $\sigma^{*2}=0.1848$  is the annualized variance computed from the daily SPY ETF return variance in \cite{nzokem2021fitting}.\\
 Option Moneyness describes the intrinsic value of an option in its current state. It indicates whether the option would make money if exercised immediately. Option moneyness can be classified into three categories: At-The-Money (ATM) option ($k=\frac{S_{t}}{K}=1$), Out-of-The Money (OTM) option ($k=\frac{S_{t}}{K}<1$), and In-The-Money (ITM) option ($k=\frac{S_{t}}{K}>1$). On August 04, 2021, the SPY ETF market price closed at $438.98$. We compute the VG call option price on SPY ETF with the spot price ($S_{0}$) $438.98$.  The results are summarised in Table \ref{tab:1}. 
 
\begin{table}[ht]
\centering
\caption{The price of a European call on SPY ETF}
\vspace{-0.2cm}
\label{tab:1}
\resizebox{\columnwidth}{!}{%
\begin{tabular}{@{}lc|ccc|ccc|ccc|ccc|ccc|ccc@{}}
\toprule
\multirow{1}{*}{\textbf{Strike   Price}} & \multirow{1}{*}{\textbf{Moneyness}} & \multirow{1}{*}{\textbf{BSM}}& \multirow{1}{*}{\textbf{VGM(\ref{eq:l5})}} & \multirow{1}{*}{\textbf{VGM(\ref{eq:l8})}} & \multirow{1}{*}{\textbf{BSM}} & \multirow{1}{*}{\textbf{VGM(\ref{eq:l5})}} & \multirow{1}{*}{\textbf{VGM(\ref{eq:l8})}} & \multirow{1}{*}{\textbf{BSM}} & \multirow{1}{*}{\textbf{VGM(\ref{eq:l5})}} &\multirow{1}{*}{\textbf{VGM(\ref{eq:l8})}} & \multirow{1}{*}{\textbf{BSM}} & \multirow{1}{*}{\textbf{VGM(\ref{eq:l5})}} & \multirow{1}{*}{\textbf{VGM(\ref{eq:l8})}} & \multirow{1}{*}{\textbf{BSM}} & \multirow{1}{*}{\textbf{VGM(\ref{eq:l5})}} & \multirow{1}{*}{\textbf{VGM(\ref{eq:l8})}} & \multirow{1}{*}{\textbf{BSM}} & \multirow{1}{*}{\textbf{VGM(\ref{eq:l5})}} &\multirow{1}{*}{\textbf{VGM(\ref{eq:l8})}} \\ \midrule

\multicolumn{2}{c|}{\textbf{Period ( in year)}} & \multicolumn{3}{c|}{\textbf{0.0625}} & \multicolumn{3}{c|}{\textbf{0.125}} & \multicolumn{3}{c|}{\textbf{0.25}} & \multicolumn{3}{c|}{\textbf{0.5}} & \multicolumn{3}{c|}{\textbf{0.75}} & \multicolumn{3}{c}{\textbf{1}} \\ \midrule

\multirow{1}{*}{\textbf{219.49}} & \multirow{1}{*}{\textbf{2.00}} & \multirow{1}{*}{220.31} & \multirow{1}{*}{220.28} & \multirow{1}{*}{219.86} & \multirow{1}{*}{221.13} & \multirow{1}{*}{221.10} & \multirow{1}{*}{220.48} & \multirow{1}{*}{222.76} & \multirow{1}{*}{222.72} & \multirow{1}{*}{221.71} & \multirow{1}{*}{225.98} & \multirow{1}{*}{225.93} & \multirow{1}{*}{224.14} & \multirow{1}{*}{229.15} & \multirow{1}{*}{229.13} & \multirow{1}{*}{226.53} & \multirow{1}{*}{232.27} & \multirow{1}{*}{232.26} & \multirow{1}{*}{228.88} \\
\multirow{1}{*}{\textbf{225.12}} & \multirow{1}{*}{\textbf{1.95}} & \multirow{1}{*}{214.70} & \multirow{1}{*}{214.74} & \multirow{1}{*}{214.10} & \multirow{1}{*}{215.54} & \multirow{1}{*}{215.58} & \multirow{1}{*}{214.74} & \multirow{1}{*}{217.21} & \multirow{1}{*}{217.25} & \multirow{1}{*}{216.01} & \multirow{1}{*}{220.52} & \multirow{1}{*}{220.54} & \multirow{1}{*}{218.53} & \multirow{1}{*}{223.77} & \multirow{1}{*}{223.82} & \multirow{1}{*}{221.01} & \multirow{1}{*}{226.97} & \multirow{1}{*}{227.04} & \multirow{1}{*}{223.45} \\
\multirow{1}{*}{\textbf{231.04}} & \multirow{1}{*}{\textbf{1.90}} & \multirow{1}{*}{208.80} & \multirow{1}{*}{208.83} & \multirow{1}{*}{208.18} & \multirow{1}{*}{209.66} & \multirow{1}{*}{209.69} & \multirow{1}{*}{208.84} & \multirow{1}{*}{211.38} & \multirow{1}{*}{211.41} & \multirow{1}{*}{210.16} & \multirow{1}{*}{214.77} & \multirow{1}{*}{214.80} & \multirow{1}{*}{212.77} & \multirow{1}{*}{218.10} & \multirow{1}{*}{218.16} & \multirow{1}{*}{215.34} & \multirow{1}{*}{221.39} & \multirow{1}{*}{221.47} & \multirow{1}{*}{217.88} \\
\multirow{1}{*}{\textbf{237.29}} & \multirow{1}{*}{\textbf{1.85}} & \multirow{1}{*}{202.58} & \multirow{1}{*}{202.53} & \multirow{1}{*}{202.10} & \multirow{1}{*}{203.47} & \multirow{1}{*}{203.42} & \multirow{1}{*}{202.79} & \multirow{1}{*}{205.23} & \multirow{1}{*}{205.18} & \multirow{1}{*}{204.16} & \multirow{1}{*}{208.71} & \multirow{1}{*}{208.67} & \multirow{1}{*}{206.86} & \multirow{1}{*}{212.13} & \multirow{1}{*}{212.13} & \multirow{1}{*}{209.53} & \multirow{1}{*}{215.51} & \multirow{1}{*}{215.53} & \multirow{1}{*}{212.16} \\
\multirow{1}{*}{\textbf{243.88}} & \multirow{1}{*}{\textbf{1.80}} & \multirow{1}{*}{196.02} & \multirow{1}{*}{196.06} & \multirow{1}{*}{195.38} & \multirow{1}{*}{196.92} & \multirow{1}{*}{196.97} & \multirow{1}{*}{196.10} & \multirow{1}{*}{198.73} & \multirow{1}{*}{198.78} & \multirow{1}{*}{197.51} & \multirow{1}{*}{202.31} & \multirow{1}{*}{202.37} & \multirow{1}{*}{200.32} & \multirow{1}{*}{205.83} & \multirow{1}{*}{205.93} & \multirow{1}{*}{203.10} & \multirow{1}{*}{209.31} & \multirow{1}{*}{209.43} & \multirow{1}{*}{205.84} \\
\multirow{1}{*}{\textbf{250.85}} & \multirow{1}{*}{\textbf{1.75}} & \multirow{1}{*}{189.07} & \multirow{1}{*}{189.16} & \multirow{1}{*}{188.47} & \multirow{1}{*}{190.01} & \multirow{1}{*}{190.10} & \multirow{1}{*}{189.21} & \multirow{1}{*}{191.87} & \multirow{1}{*}{191.96} & \multirow{1}{*}{190.68} & \multirow{1}{*}{195.55} & \multirow{1}{*}{195.66} & \multirow{1}{*}{193.61} & \multirow{1}{*}{199.17} & \multirow{1}{*}{199.33} & \multirow{1}{*}{196.50} & \multirow{1}{*}{202.75} & \multirow{1}{*}{202.94} & \multirow{1}{*}{199.36} \\
\multirow{1}{*}{\textbf{258.22}} & \multirow{1}{*}{\textbf{1.70}} & \multirow{1}{*}{181.72} & \multirow{1}{*}{181.81} & \multirow{1}{*}{181.36} & \multirow{1}{*}{182.69} & \multirow{1}{*}{182.78} & \multirow{1}{*}{182.13} & \multirow{1}{*}{184.60} & \multirow{1}{*}{184.69} & \multirow{1}{*}{183.66} & \multirow{1}{*}{188.39} & \multirow{1}{*}{188.52} & \multirow{1}{*}{186.70} & \multirow{1}{*}{192.12} & \multirow{1}{*}{192.30} & \multirow{1}{*}{189.71} & \multirow{1}{*}{195.81} & \multirow{1}{*}{196.03} & \multirow{1}{*}{192.70} \\
\multirow{1}{*}{\textbf{266.05}} & \multirow{1}{*}{\textbf{1.65}} & \multirow{1}{*}{173.93} & \multirow{1}{*}{173.98} & \multirow{1}{*}{173.53} & \multirow{1}{*}{174.92} & \multirow{1}{*}{174.97} & \multirow{1}{*}{174.33} & \multirow{1}{*}{176.89} & \multirow{1}{*}{176.95} & \multirow{1}{*}{175.92} & \multirow{1}{*}{180.79} & \multirow{1}{*}{180.91} & \multirow{1}{*}{179.09} & \multirow{1}{*}{184.64} & \multirow{1}{*}{184.83} & \multirow{1}{*}{182.24} & \multirow{1}{*}{188.45} & \multirow{1}{*}{188.69} & \multirow{1}{*}{185.37} \\
\multirow{1}{*}{\textbf{274.36}} & \multirow{1}{*}{\textbf{1.60}} & \multirow{1}{*}{165.64} & \multirow{1}{*}{165.64} & \multirow{1}{*}{165.45} & \multirow{1}{*}{166.67} & \multirow{1}{*}{166.66} & \multirow{1}{*}{166.28} & \multirow{1}{*}{168.70} & \multirow{1}{*}{168.70} & \multirow{1}{*}{167.94} & \multirow{1}{*}{172.73} & \multirow{1}{*}{172.82} & \multirow{1}{*}{171.26} & \multirow{1}{*}{176.70} & \multirow{1}{*}{176.87} & \multirow{1}{*}{174.56} & \multirow{1}{*}{180.63} & \multirow{1}{*}{180.88} & \multirow{1}{*}{177.84} \\
\multirow{1}{*}{\textbf{283.21}} & \multirow{1}{*}{\textbf{1.55}} & \multirow{1}{*}{156.83} & \multirow{1}{*}{156.75} & \multirow{1}{*}{156.56} & \multirow{1}{*}{157.88} & \multirow{1}{*}{157.81} & \multirow{1}{*}{157.43} & \multirow{1}{*}{159.98} & \multirow{1}{*}{159.91} & \multirow{1}{*}{159.18} & \multirow{1}{*}{164.14} & \multirow{1}{*}{164.21} & \multirow{1}{*}{162.66} & \multirow{1}{*}{168.25} & \multirow{1}{*}{168.42} & \multirow{1}{*}{166.14} & \multirow{1}{*}{172.33} & \multirow{1}{*}{172.59} & \multirow{1}{*}{169.60} \\
\multirow{1}{*}{\textbf{292.65}} & \multirow{1}{*}{\textbf{1.50}} & \multirow{1}{*}{147.42} & \multirow{1}{*}{147.28} & \multirow{1}{*}{146.81} & \multirow{1}{*}{148.51} & \multirow{1}{*}{148.38} & \multirow{1}{*}{147.73} & \multirow{1}{*}{150.68} & \multirow{1}{*}{150.56} & \multirow{1}{*}{149.56} & \multirow{1}{*}{154.98} & \multirow{1}{*}{155.05} & \multirow{1}{*}{153.24} & \multirow{1}{*}{159.24} & \multirow{1}{*}{159.45} & \multirow{1}{*}{156.93} & \multirow{1}{*}{163.49} & \multirow{1}{*}{163.79} & \multirow{1}{*}{160.59} \\
\multirow{1}{*}{\textbf{302.75}} & \multirow{1}{*}{\textbf{1.45}} & \multirow{1}{*}{137.37} & \multirow{1}{*}{137.50} & \multirow{1}{*}{136.72} & \multirow{1}{*}{138.50} & \multirow{1}{*}{138.63} & \multirow{1}{*}{137.69} & \multirow{1}{*}{140.74} & \multirow{1}{*}{140.89} & \multirow{1}{*}{139.62} & \multirow{1}{*}{145.20} & \multirow{1}{*}{145.61} & \multirow{1}{*}{143.53} & \multirow{1}{*}{149.64} & \multirow{1}{*}{150.21} & \multirow{1}{*}{147.45} & \multirow{1}{*}{154.08} & \multirow{1}{*}{154.76} & \multirow{1}{*}{151.34} \\
\multirow{1}{*}{\textbf{313.56}} & \multirow{1}{*}{\textbf{1.40}} & \multirow{1}{*}{126.60} & \multirow{1}{*}{126.45} & \multirow{1}{*}{126.29} & \multirow{1}{*}{127.77} & \multirow{1}{*}{127.64} & \multirow{1}{*}{127.31} & \multirow{1}{*}{130.09} & \multirow{1}{*}{129.99} & \multirow{1}{*}{129.36} & \multirow{1}{*}{134.73} & \multirow{1}{*}{135.00} & \multirow{1}{*}{133.53} & \multirow{1}{*}{139.39} & \multirow{1}{*}{139.85} & \multirow{1}{*}{137.71} & \multirow{1}{*}{144.06} & \multirow{1}{*}{144.63} & \multirow{1}{*}{141.86} \\
\multirow{1}{*}{\textbf{325.17}} & \multirow{1}{*}{\textbf{1.35}} & \multirow{1}{*}{115.03} & \multirow{1}{*}{115.01} & \multirow{1}{*}{114.85} & \multirow{1}{*}{116.24} & \multirow{1}{*}{116.25} & \multirow{1}{*}{115.94} & \multirow{1}{*}{118.65} & \multirow{1}{*}{118.71} & \multirow{1}{*}{118.15} & \multirow{1}{*}{123.51} & \multirow{1}{*}{124.07} & \multirow{1}{*}{122.64} & \multirow{1}{*}{128.44} & \multirow{1}{*}{129.20} & \multirow{1}{*}{127.14} & \multirow{1}{*}{133.40} & \multirow{1}{*}{134.25} & \multirow{1}{*}{131.59} \\
\multirow{1}{*}{\textbf{337.68}} & \multirow{1}{*}{\textbf{1.30}} & \multirow{1}{*}{102.57} & \multirow{1}{*}{102.48} & \multirow{1}{*}{102.35} & \multirow{1}{*}{103.83} & \multirow{1}{*}{103.80} & \multirow{1}{*}{103.53} & \multirow{1}{*}{106.34} & \multirow{1}{*}{106.39} & \multirow{1}{*}{105.94} & \multirow{1}{*}{111.50} & \multirow{1}{*}{112.20} & \multirow{1}{*}{110.84} & \multirow{1}{*}{116.79} & \multirow{1}{*}{117.68} & \multirow{1}{*}{115.72} & \multirow{1}{*}{122.10} & \multirow{1}{*}{123.05} & \multirow{1}{*}{120.53} \\
\multirow{1}{*}{\textbf{351.18}} & \multirow{1}{*}{\textbf{1.25}} & \multirow{1}{*}{89.11} & \multirow{1}{*}{89.15} & \multirow{1}{*}{88.69} & \multirow{1}{*}{90.42} & \multirow{1}{*}{90.55} & \multirow{1}{*}{90.00} & \multirow{1}{*}{93.08} & \multirow{1}{*}{93.33} & \multirow{1}{*}{92.69} & \multirow{1}{*}{98.68} & \multirow{1}{*}{99.71} & \multirow{1}{*}{98.12} & \multirow{1}{*}{104.44} & \multirow{1}{*}{105.61} & \multirow{1}{*}{103.48} & \multirow{1}{*}{110.16} & \multirow{1}{*}{111.34} & \multirow{1}{*}{108.71} \\
\multirow{1}{*}{\textbf{365.82}} & \multirow{1}{*}{\textbf{1.20}} & \multirow{1}{*}{74.53} & \multirow{1}{*}{74.60} & \multirow{1}{*}{74.52} & \multirow{1}{*}{75.91} & \multirow{1}{*}{76.15} & \multirow{1}{*}{76.02} & \multirow{1}{*}{78.82} & \multirow{1}{*}{79.18} & \multirow{1}{*}{79.08} & \multirow{1}{*}{85.10} & \multirow{1}{*}{86.32} & \multirow{1}{*}{85.17} & \multirow{1}{*}{91.45} & \multirow{1}{*}{92.73} & \multirow{1}{*}{91.09} & \multirow{1}{*}{97.65} & \multirow{1}{*}{98.88} & \multirow{1}{*}{96.78} \\
\multirow{1}{*}{\textbf{381.72}} & \multirow{1}{*}{\textbf{1.15}} & \multirow{1}{*}{58.69} & \multirow{1}{*}{59.15} & \multirow{1}{*}{58.38} & \multirow{1}{*}{60.22} & \multirow{1}{*}{60.94} & \multirow{1}{*}{60.20} & \multirow{1}{*}{63.67} & \multirow{1}{*}{64.34} & \multirow{1}{*}{63.82} & \multirow{1}{*}{70.94} & \multirow{1}{*}{72.47} & \multirow{1}{*}{70.85} & \multirow{1}{*}{77.99} & \multirow{1}{*}{79.48} & \multirow{1}{*}{77.46} & \multirow{1}{*}{84.70} & \multirow{1}{*}{86.09} & \multirow{1}{*}{83.68} \\
\multirow{1}{*}{\textbf{399.07}} & \multirow{1}{*}{\textbf{1.10}} & \multirow{1}{*}{41.51} & \multirow{1}{*}{42.09} & \multirow{1}{*}{41.76} & \multirow{1}{*}{43.56} & \multirow{1}{*}{44.29} & \multirow{1}{*}{44.08} & \multirow{1}{*}{48.03} & \multirow{1}{*}{48.30} & \multirow{1}{*}{48.53} & \multirow{1}{*}{56.55} & \multirow{1}{*}{57.74} & \multirow{1}{*}{56.68} & \multirow{1}{*}{64.32} & \multirow{1}{*}{65.45} & \multirow{1}{*}{64.03} & \multirow{1}{*}{71.52} & \multirow{1}{*}{72.55} & \multirow{1}{*}{70.80} \\
\multirow{1}{*}{\textbf{418.08}} & \multirow{1}{*}{\textbf{1.05}} & \multirow{1}{*}{23.73} & \multirow{1}{*}{24.37} & \multirow{1}{*}{24.15} & \multirow{1}{*}{27.03} & \multirow{1}{*}{27.33} & \multirow{1}{*}{27.36} & \multirow{1}{*}{32.83} & \multirow{1}{*}{32.25} & \multirow{1}{*}{33.00} & \multirow{1}{*}{42.51} & \multirow{1}{*}{43.25} & \multirow{1}{*}{42.47} & \multirow{1}{*}{50.85} & \multirow{1}{*}{51.63} & \multirow{1}{*}{50.57} & \multirow{1}{*}{58.42} & \multirow{1}{*}{59.16} & \multirow{1}{*}{57.83} \\
\multirow{1}{*}{\textbf{438.98}} & \multirow{1}{*}{\textbf{1.00}} & \multirow{1}{*}{8.92} & \multirow{1}{*}{6.45} & \multirow{1}{*}{6.76} & \multirow{1}{*}{13.11} & \multirow{1}{*}{11.13} & \multirow{1}{*}{11.40} & \multirow{1}{*}{19.53} & \multirow{1}{*}{18.35} & \multirow{1}{*}{18.43} & \multirow{1}{*}{29.61} & \multirow{1}{*}{29.17} & \multirow{1}{*}{29.01} & \multirow{1}{*}{38.11} & \multirow{1}{*}{38.01} & \multirow{1}{*}{37.61} & \multirow{1}{*}{45.79} & \multirow{1}{*}{45.79} & \multirow{1}{*}{45.20} \\
\multirow{1}{*}{\textbf{462.08}} & \multirow{1}{*}{\textbf{0.95}} & \multirow{1}{*}{1.64} & \multirow{1}{*}{1.19} & \multirow{1}{*}{1.27} & \multirow{1}{*}{4.40} & \multirow{1}{*}{2.94} & \multirow{1}{*}{3.02} & \multirow{1}{*}{9.62} & \multirow{1}{*}{7.41} & \multirow{1}{*}{7.50} & \multirow{1}{*}{18.70} & \multirow{1}{*}{17.17} & \multirow{1}{*}{17.09} & \multirow{1}{*}{26.72} & \multirow{1}{*}{25.85} & \multirow{1}{*}{25.50} & \multirow{1}{*}{34.12} & \multirow{1}{*}{33.67} & \multirow{1}{*}{33.01} \\
\multirow{1}{*}{\textbf{487.76}} & \multirow{1}{*}{\textbf{0.90}} & \multirow{1}{*}{0.10} & \multirow{1}{*}{0.35} & \multirow{1}{*}{0.40} & \multirow{1}{*}{0.89} & \multirow{1}{*}{0.96} & \multirow{1}{*}{1.02} & \multirow{1}{*}{3.68} & \multirow{1}{*}{2.82} & \multirow{1}{*}{2.94} & \multirow{1}{*}{10.42} & \multirow{1}{*}{8.69} & \multirow{1}{*}{8.83} & \multirow{1}{*}{17.24} & \multirow{1}{*}{15.73} & \multirow{1}{*}{15.71} & \multirow{1}{*}{23.87} & \multirow{1}{*}{22.75} & \multirow{1}{*}{22.48} \\
\multirow{1}{*}{\textbf{516.45}} & \multirow{1}{*}{\textbf{0.85}} & \multirow{1}{*}{0.00} & \multirow{1}{*}{0.10} & \multirow{1}{*}{0.13} & \multirow{1}{*}{0.09} & \multirow{1}{*}{0.31} & \multirow{1}{*}{0.34} & \multirow{1}{*}{1.02} & \multirow{1}{*}{1.03} & \multirow{1}{*}{1.11} & \multirow{1}{*}{4.96} & \multirow{1}{*}{3.92} & \multirow{1}{*}{4.09} & \multirow{1}{*}{10.02} & \multirow{1}{*}{8.48} & \multirow{1}{*}{8.63} & \multirow{1}{*}{15.46} & \multirow{1}{*}{13.90} & \multirow{1}{*}{13.93} \\
\multirow{1}{*}{\textbf{548.73}} & \multirow{1}{*}{\textbf{0.80}} & \multirow{1}{*}{0.00} & \multirow{1}{*}{0.03} & \multirow{1}{*}{0.04} & \multirow{1}{*}{0.00} & \multirow{1}{*}{0.10} & \multirow{1}{*}{0.12} & \multirow{1}{*}{0.19} & \multirow{1}{*}{0.37} & \multirow{1}{*}{0.41} & \multirow{1}{*}{1.94} & \multirow{1}{*}{1.66} & \multirow{1}{*}{1.78} & \multirow{1}{*}{5.12} & \multirow{1}{*}{4.19} & \multirow{1}{*}{4.37} & \multirow{1}{*}{9.10} & \multirow{1}{*}{7.80} & \multirow{1}{*}{7.96} \\
\multirow{1}{*}{\textbf{585.31}} & \multirow{1}{*}{\textbf{0.75}} & \multirow{1}{*}{0.00} & \multirow{1}{*}{0.01} & \multirow{1}{*}{0.01} & \multirow{1}{*}{0.00} & \multirow{1}{*}{0.03} & \multirow{1}{*}{0.04} & \multirow{1}{*}{0.02} & \multirow{1}{*}{0.12} & \multirow{1}{*}{0.14} & \multirow{1}{*}{0.60} & \multirow{1}{*}{0.64} & \multirow{1}{*}{0.72} & \multirow{1}{*}{2.23} & \multirow{1}{*}{1.85} & \multirow{1}{*}{2.02} & \multirow{1}{*}{4.76} & \multirow{1}{*}{3.91} & \multirow{1}{*}{4.14} \\
\multirow{1}{*}{\textbf{627.11}} & \multirow{1}{*}{\textbf{0.70}} & \multirow{1}{*}{0.00} & \multirow{1}{*}{0.00} & \multirow{1}{*}{0.00} & \multirow{1}{*}{0.00} & \multirow{1}{*}{0.01} & \multirow{1}{*}{0.01} & \multirow{1}{*}{0.00} & \multirow{1}{*}{0.04} & \multirow{1}{*}{0.04} & \multirow{1}{*}{0.14} & \multirow{1}{*}{0.23} & \multirow{1}{*}{0.26} & \multirow{1}{*}{0.80} & \multirow{1}{*}{0.75} & \multirow{1}{*}{0.83} & \multirow{1}{*}{2.15} & \multirow{1}{*}{1.77} & \multirow{1}{*}{1.91} \\
\multirow{1}{*}{\textbf{675.35}} & \multirow{1}{*}{\textbf{0.65}} & \multirow{1}{*}{0.00} & \multirow{1}{*}{0.00} & \multirow{1}{*}{0.00} & \multirow{1}{*}{0.00} & \multirow{1}{*}{0.00} & \multirow{1}{*}{0.00} & \multirow{1}{*}{0.00} & \multirow{1}{*}{0.01} & \multirow{1}{*}{0.01} & \multirow{1}{*}{0.02} & \multirow{1}{*}{0.07} & \multirow{1}{*}{0.09} & \multirow{1}{*}{0.22} & \multirow{1}{*}{0.27} & \multirow{1}{*}{0.31} & \multirow{1}{*}{0.81} & \multirow{1}{*}{0.72} & \multirow{1}{*}{0.81} \\
\multirow{1}{*}{\textbf{731.63}} & \multirow{1}{*}{\textbf{0.60}} & \multirow{1}{*}{0.00} & \multirow{1}{*}{0.00} & \multirow{1}{*}{0.00} & \multirow{1}{*}{0.00} & \multirow{1}{*}{0.00} & \multirow{1}{*}{0.00} & \multirow{1}{*}{0.00} & \multirow{1}{*}{0.00} & \multirow{1}{*}{0.00} & \multirow{1}{*}{0.00} & \multirow{1}{*}{0.02} & \multirow{1}{*}{0.03} & \multirow{1}{*}{0.05} & \multirow{1}{*}{0.09} & \multirow{1}{*}{0.11} & \multirow{1}{*}{0.24} & \multirow{1}{*}{0.26} & \multirow{1}{*}{0.31} \\
\multirow{1}{*}{\textbf{798.15}} & \multirow{1}{*}{\textbf{0.55}} & \multirow{1}{*}{0.00} & \multirow{1}{*}{0.00} & \multirow{1}{*}{0.00} & \multirow{1}{*}{0.00} & \multirow{1}{*}{0.00} & \multirow{1}{*}{0.00} & \multirow{1}{*}{0.00} & \multirow{1}{*}{0.00} & \multirow{1}{*}{0.00} & \multirow{1}{*}{0.00} & \multirow{1}{*}{0.01} & \multirow{1}{*}{0.01} & \multirow{1}{*}{0.01} & \multirow{1}{*}{0.02} & \multirow{1}{*}{0.03} & \multirow{1}{*}{0.06} & \multirow{1}{*}{0.08} & \multirow{1}{*}{0.10} \\
\multirow{1}{*}{\textbf{877.96}} & \multirow{1}{*}{\textbf{0.50}} & \multirow{1}{*}{0.00} & \multirow{1}{*}{0.00} & \multirow{1}{*}{0.00} & \multirow{1}{*}{0.00} & \multirow{1}{*}{0.00} & \multirow{1}{*}{0.00} & \multirow{1}{*}{0.00} & \multirow{1}{*}{0.00} & \multirow{1}{*}{0.00} & \multirow{1}{*}{0.00} & \multirow{1}{*}{0.00} & \multirow{1}{*}{0.00} & \multirow{1}{*}{0.00} & \multirow{1}{*}{0.01} & \multirow{1}{*}{0.01} & \multirow{1}{*}{0.01} & \multirow{1}{*}{0.02} & \multirow{1}{*}{0.03} \\ \bottomrule
\end{tabular}%
}
\footnotesize (\ref{eq:l5}): 12-point rule Composite Newton-Cotes Quadrature  \    \ (\ref{eq:l8}): Fractional Fourier Transform (FRFT) 
\end{table}

\noindent
The Fractional Fourier Transform (FRFT)  algorithm performs poorly for the in-the-money (ITM) option. The FRFT underprices the VG option for in-the-money (ITM), whereas the 12-point rule Composite Newton-Cotes Quadrature produces consistent option pricing results with the Black-Scholes model. Both algorithms yield consistent results for at-the-money and out-of-the-money options. \\
\noindent
To generalize the analysis and account for a large range of option moneyness and maturity, The error (\ref{eq:l245})  was computed as the difference between VG option and BS option prices.
\begin{align}
Error(k,\tau)=\frac{F^{GV}_{call}(S_{t},\tau)}{K}  - \frac{F^{BS}_{call}(S_{t},\tau)}{K}  \quad \quad  \hbox{$(k=\frac{S_{t}}{K})$} \label {eq:l245}
\end{align}
\noindent 
Fig \ref{fig1ab} graphs the error as a function of the time to maturity ($\tau$) and the option moneyness ($k$). The spot price ($S_{t}$) is a constant, and the option moneyness depends on the strike price.\\
 \begin{figure}[ht]
\vspace{-0.5cm}
    \centering
  \begin{subfigure}[b]{0.47\linewidth}
    \includegraphics[width=\linewidth]{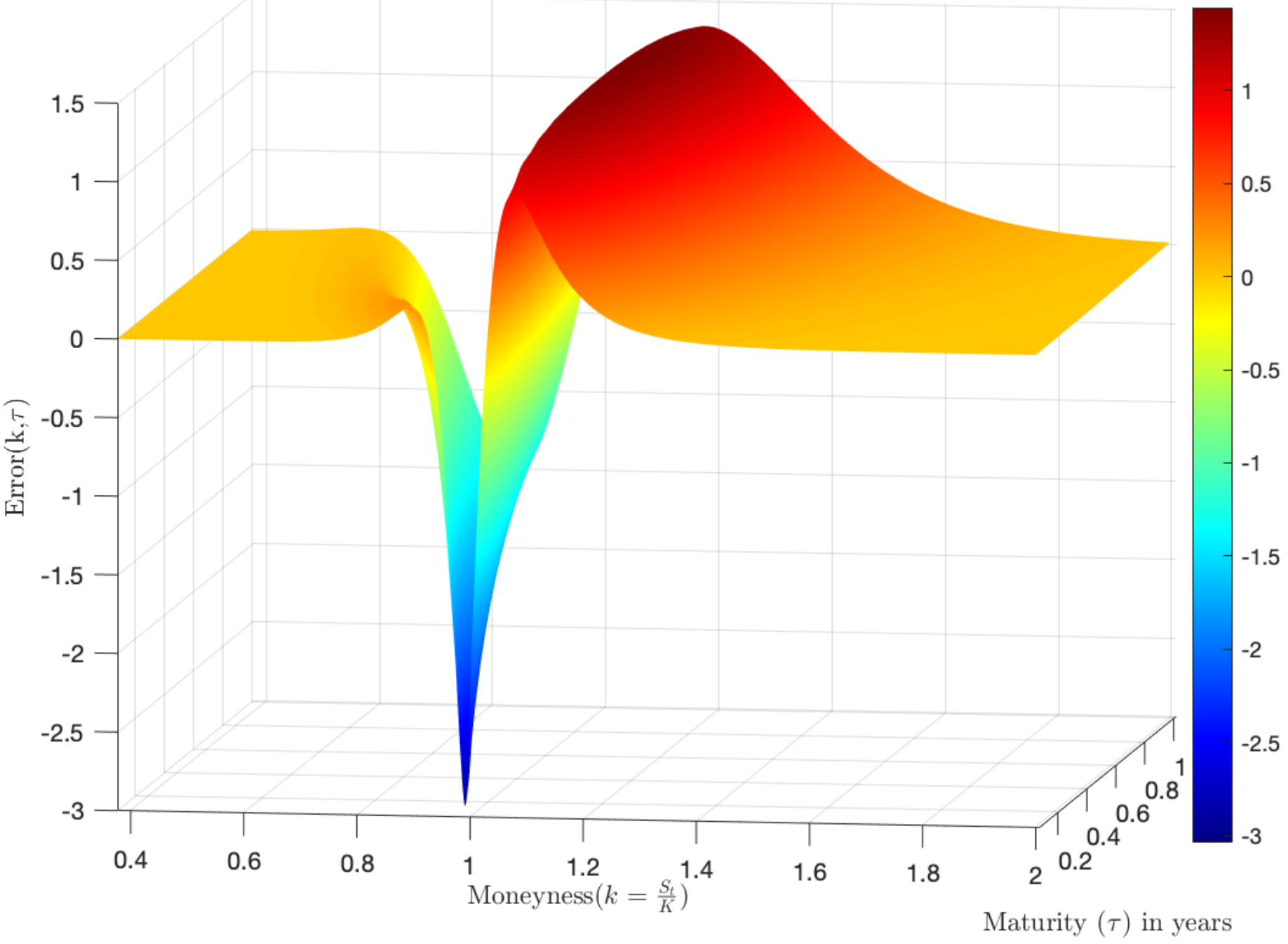}
\vspace{-0.3cm}
     \caption{Error(k,$\tau$)}
         \label{fig101a}
  \end{subfigure}
  \begin{subfigure}[b]{0.43\linewidth}
    \includegraphics[width=\linewidth]{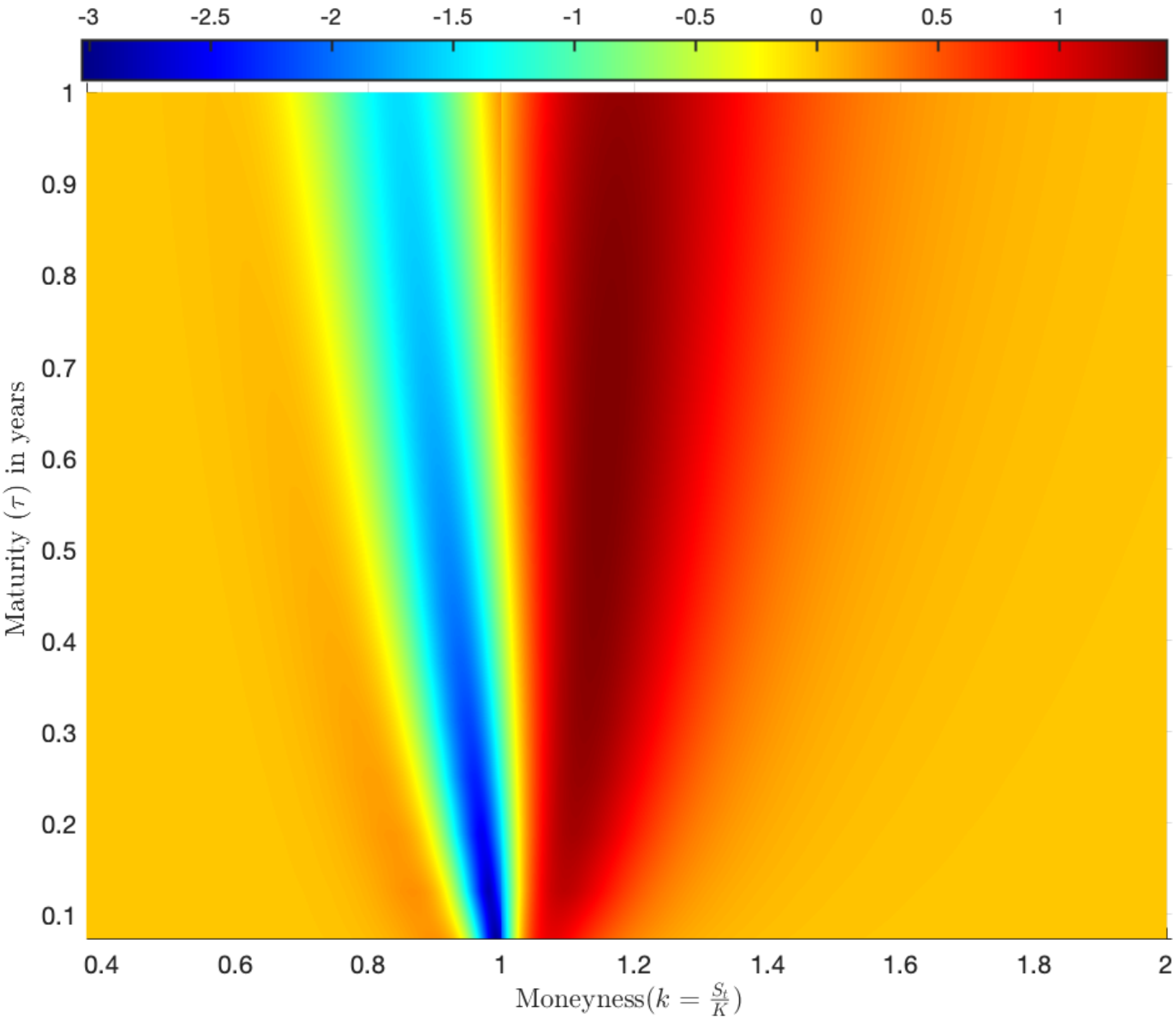}
\vspace{-0.5cm}
     \caption{Error(k,$\tau$) (top view)}
         \label{fig101b}
          \end{subfigure}
\vspace{-0.5cm}
  \caption{Combined Effects of time to Maturity ($\tau$) and Option Moneyness  ($k=\frac{S_{t}}{K}$)}
  \label{fig1ab}
\vspace{-0.6cm}
\end{figure}

\noindent 
The Black-Scholes (BS) and VG models produce different option pricing results. The Black-Scholes model is overpriced for the out-of-the-money (OTM) option (blue color in  Fig \ref{fig1ab}) and underpriced for the in-the-money (ITM) option (red color).\\
\noindent 
The results in Fig \ref{fig1ab} are consistent with \cite{mozumder2015revisiting}, where the VG pricing was performed on S\&P500 index data. The shape in Fig \ref{fig101a} looks similar to Fig {6} in \cite{mozumder2015revisiting} when the Option moneyness variable replaces the strike price.  However, the overpriced Black-Scholes model in bleu color (Fig \ref{fig1ab}) does not support the findings \cite{madan1991option} that the VG option prices are typically higher than the Black-Scholes model prices, with the percentage bias rising when the stock gets out-of-the-money (OTM). One of the limitations of these studies is that the VG model is symmetric and uses three parameters. The five-parameter VG model controls the excess kurtosis and the skewness of the daily SPY ETF return data.
\section {Conclusion} 
\noindent 
In the paper, a $\Gamma(\alpha, \theta)$ Ornstein-Uhlenbeck type process was used to build a continuous sample path of a five-parameter Variance-Gamma (VG) process ($\mu,\delta,\sigma,\alpha,\theta$): location ($\mu$), symmetric ($\delta$), volatility ($\sigma$), and shape ($\alpha$) and scale ($\theta$). The data parameters \cite{nzokem2021fitting,nzokem_2021b} were used to simulate the gamma process ($\sigma^{2}(t)$) and the continuous sample path of the subordinator process ($\sigma^{2*}(t)$). Both simulations were used as inputs to simulate the VG process's continuous sample path. The L\'evy density of the VG process was derived and shown to belong to a KoPoL family of order  $\nu=0$, intensity $\alpha$ and steepness parameters $\frac{\delta}{\sigma^2} - \sqrt{\frac{\delta^2}{\sigma^4}+\frac{2}{\theta \sigma^2}}$ and $\frac{\delta}{\sigma^2}+ \sqrt{\frac{\delta^2}{\sigma^4}+\frac{2}{\theta \sigma^2}}$. It was shown that the VG process converges asymptotically in distribution to a L\'evy process driven by a Normal distribution with mean $(\mu + \alpha \theta \delta)$ and variance $\alpha (\theta^2\delta^2 + \sigma^2\theta)$.
The existence of the Equivalent Martingale Measure (EMM)  of the five-parameter VG process was shown. The EMM preserves the structure of the five-parameter VG process with an inflated Gamma scale parameter and a constant term adjustment symmetric parameter. The extended Black-Scholes formula provides the closed form of the VG option price. The L\'evy process generated by the VG model provides the Generalized Black-Scholes Formula. The daily SPY ETF return data illustrates the computation of the European option price under the five-parameter VG process. The 12-point rule Composite Newton-Cotes Quadrature and the Fractional Fast Fourier (FRFT) algorithms were implemented to compute the European option price. It results that the FRFT yields inconsistent European option prices for in-the-money options. The Black-Scholes (BS) and VG models produce different option pricing results. The Black-Scholes model is overpriced for out-of-the-money (OTM) options and underpriced for in-the-money (ITM) options. However, for deep out-of-the-money (OTM) and deep-in-the-money (ITM) options, Black-Scholes (BS) and VG models yield almost the same option price.

\pagestyle{plain}
\addcontentsline{toc}{chapter}{References}
\bibliographystyle{plain}
\bibliography{fourierbis.bib}

\begin{thebibliography}{10}

\bibitem{adeosun2016variance}
ME~Adeosun, SO~Edeki, and OO~Ugbebor.
\newblock On a variance gamma model (vgm) in option pricing: A difference of
  two gamma processes.
\newblock {\em Journal of Informatics and Mathematical Sciences}, 8(1):1--16,
  2016.

\bibitem{alhagyan2020discussions}
Mohammed Alhagyan, Masnita Misiran, and Zurni Omar.
\newblock Discussions on continuous stochastic volatility models.
\newblock {\em Global and Stochastic Analysis}, 7(1):55--64, 2020.

\bibitem{andrii2020}
Andrii Andrusiv and Hans-J{\"u}rgen Engelbert.
\newblock On the minimal entropy martingale measure for l{\'e}vy processes.
\newblock {\em Stochastics}, 92(8):1223--1243, 2020.

\bibitem{applebaum2009levy}
David Applebaum.
\newblock {\em L{\'e}vy processes and stochastic calculus}.
\newblock Cambridge university press, 2009.

\bibitem{arias1990theorem}
Juan Arias-de Reyna.
\newblock On the theorem of frullani.
\newblock {\em Proceedings of the American Mathematical Society},
  109(1):165--175, 1990.

\bibitem{barndorff1998some}
Ole~E Barndorff-Nielsen, Jens~Ledet Jensen, and Michael S{\o}rensen.
\newblock Some stationary processes in discrete and continuous time.
\newblock {\em Advances in Applied Probability}, 30(4):989--1007, 1998.

\bibitem{barndorff1999non}
Ole~E Barndorff-Nielsen and Neil Shephard.
\newblock {\em Non-Gaussian OU based models and some of their uses in financial
  economics}.
\newblock Nuffield College Oxford, 1999.

\bibitem{barndorff2001modelling}
Ole~E Barndorff-Nielsen and Neil Shephard.
\newblock Modelling by l{\'e}vy processess for financial econometrics.
\newblock {\em L{\'e}vy processes}, pages 283--318, 2001.

\bibitem{barndorff2001non}
Ole~E Barndorff-Nielsen and Neil Shephard.
\newblock Non-gaussian ornstein--uhlenbeck-based models and some of their uses
  in financial economics.
\newblock {\em Journal of the Royal Statistical Society: Series B (Statistical
  Methodology)}, 63(2):167--241, 2001.

\bibitem{barndorff2002financial}
Ole~E Barndorff-Nielsen and Neil Shephard.
\newblock Financial volatility, l{\'e}vy processes and power variation.
\newblock {\em Unpublished book, Nuffield College}, 2002.

\bibitem{barndorff2003integrated}
Ole~E Barndorff-Nielsen and Neil Shephard.
\newblock Integrated ou processes and non-gaussian ou-based stochastic
  volatility models.
\newblock {\em Scandinavian Journal of statistics}, 30(2):277--295, 2003.

\bibitem{black1973}
Fischer Black and Myron Scholes.
\newblock The pricing of options and corporate liabilities.
\newblock {\em Journal of Political Economy}, 81(3):637--654, 1973.

\bibitem{boyarchenko2002non}
Svetlana Boyarchenko and Sergei~Z Levendorskii.
\newblock {\em Non-Gaussian Merton-Black-Scholes Theory}, volume~9.
\newblock World Scientific, 2002.

\bibitem{carr2003stochastic}
Peter Carr, H{\'e}lyette Geman, Dilip~B Madan, and Marc Yor.
\newblock Stochastic volatility for l{\'e}vy processes.
\newblock {\em Mathematical finance}, 13(3):345--382, 2003.

\bibitem{Cont_2001}
R.~Cont.
\newblock Empirical properties of asset returns: stylized facts and statistical
  issues.
\newblock {\em Quantitative Finance}, 1(2):223--236, feb 2001.

\bibitem{gerber1993option}
Hans~U Gerber, Elias~SW Shiu, et~al.
\newblock {\em Option pricing by Esscher transforms}.
\newblock HEC Ecole des hautes {\'e}tudes commerciales, 1993.

\bibitem{heston1993closed}
Steven~L Heston.
\newblock A closed-form solution for options with stochastic volatility with
  applications to bond and currency options.
\newblock {\em The review of financial studies}, 6(2):327--343, 1993.

\bibitem{hull1987pricing}
John Hull and Alan White.
\newblock The pricing of options on assets with stochastic volatilities.
\newblock {\em The journal of finance}, 42(2):281--300, 1987.

\bibitem{hull2003options}
John~C Hull.
\newblock {\em Options futures and other derivatives}.
\newblock Pearson Education India, 2003.

\bibitem{ken1999levy}
Sato Ken-Iti.
\newblock {\em L{\'e}vy processes and infinitely divisible distributions}.
\newblock Cambridge university press, 1999.

\bibitem{kendall1946advanced}
Maurice~George Kendall et~al.
\newblock The advanced theory of statistics.
\newblock {\em The advanced theory of statistics.}, (2nd Ed), 1946.

\bibitem{kou2002jump}
Steven~G Kou.
\newblock A jump-diffusion model for option pricing.
\newblock {\em Management science}, 48(8):1086--1101, 2002.

\bibitem{kyprianou2014fluctuations}
Andreas~E Kyprianou.
\newblock {\em Fluctuations of L{\'e}vy processes with applications:
  Introductory Lectures}.
\newblock Springer Science \& Business Media, 2014.

\bibitem{li2020pricing}
Cuixiang Li, Huili Liu, Mengna Wang, and Wenhan Li.
\newblock The pricing of compound option under variance gamma process by fft.
\newblock {\em Communications in Statistics-Theory and Methods}, pages 1--15,
  2020.

\bibitem{madan1998variance}
Dilip~B Madan, Peter~P Carr, and Eric~C Chang.
\newblock The variance gamma process and option pricing.
\newblock {\em Review of Finance}, 2(1):79--105, 1998.

\bibitem{madan1991option}
Dilip~B Madan and Frank Milne.
\newblock Option pricing with vg martingale components 1.
\newblock {\em Mathematical finance}, 1(4):39--55, 1991.

\bibitem{matsuda2004introduction}
Kazuhisa Matsuda.
\newblock Introduction to merton jump diffusion model.
\newblock {\em Department of Economics. The Graduate Center, The City
  University of New York}, 2004.

\bibitem{mozumder2015revisiting}
Sharif Mozumder, Ghulam Sorwar, and Kevin Dowd.
\newblock Revisiting variance gamma pricing: An application to s\&p500 index
  options.
\newblock {\em International Journal of Financial Engineering}, 2(02):1550022,
  2015.

\bibitem{aubain2020}
A.~H. Nzokem.
\newblock {\em Stochastic and Renewal Methods Applied to Epidemic Models}.
\newblock PhD thesis, York University , YorkSpace institutional repository,
  2020.

\bibitem{nzokemmay}
A.~H. Nzokem and V.~T. Montshiwa.
\newblock Fitting generalized tempered stable distribution: Fractional fourier
  transform (frft) approach.
\newblock ARXIV.2205.00586[q-fin.ST], 2022.

\bibitem{nzokem2021fitting}
A.H. Nzokem.
\newblock Fitting infinitely divisible distribution: Case of gamma-variance
  model, 2021.

\bibitem{nzokem_2021b}
A.H. Nzokem.
\newblock Gamma variance model: Fractional fourier transform ({FRFT}).
\newblock {\em Journal of Physics: Conference Series}, 2090(1):012094, nov
  2021.

\bibitem{Nzokem_2021}
A.H. Nzokem.
\newblock Numerical solution of a gamma - integral equation using a higher
  order composite newton-cotes formulas.
\newblock {\em Journal of Physics: Conference Series}, 2084(1):012019, nov
  2021.

\bibitem{nzokem2021sis}
A.H. Nzokem.
\newblock Sis epidemic model: Birth-and-death markov chain approach.
\newblock {\em International Journal of Statistics and Probability},
  10(4):10--20, July 2021.

\bibitem{protter2005}
Philip Protter.
\newblock {\em Stochastic integration and differential equations. A new
  approach}, volume~21.
\newblock 01 2005.

\bibitem{tankov2003financial}
Peter Tankov.
\newblock {\em Financial modelling with jump processes}.
\newblock Chapman and Hall/CRC, 2003.

\bibitem{uhlenback1930theory}
GE~Uhlenback and LS~Ornstein.
\newblock On the theory of the brownian motion.
\newblock {\em Phys. Rev}, 36:823--841, 1930.

\end{thebibliography}
\end{document}